\documentclass[11pt,twoside]{article}
\usepackage[margin=1in]{geometry}
\usepackage[utf8]{inputenc}
\usepackage{natbib}
\usepackage{usebib}
\usepackage{amsmath}
\usepackage{mathtools}
\usepackage{amsthm}
\usepackage{amsfonts}
\usepackage{amssymb}
\usepackage{graphicx}
\usepackage[ruled,vlined,linesnumbered,hangingcomment]{algorithm2e}
\usepackage{textcomp}
\usepackage{makecell}
\usepackage[flushleft]{threeparttable}
\usepackage{geometry}
\usepackage{fancyhdr}
\usepackage{fontsize}
\usepackage{hyperref}
\usepackage{color}
\usepackage[dvipsnames]{xcolor}
\usepackage{bm}
\usepackage{multicol}
\usepackage{enumerate}
\usepackage[inline]{enumitem}
\usepackage[explicit]{titlesec}
\usepackage{xargs}
\usepackage{xifthen}
\usepackage{xparse}
\usepackage{etoolbox}
\usepackage{tcolorbox}
\usepackage{colortbl}
\usepackage{pdfpages}
\usepackage{sidecap}
\usepackage{xspace}
\usepackage{xfrac}
\usepackage{arydshln}
\usepackage{booktabs}

\bibinput{refs}

\bibpunct{(}{)}{;}{a}{}{,}

\def\BibTeX{{\rm B\kern-.05em{\sc i\kern-.025em b}\kern-.08em
  T\kern-.1667em\lower.7ex\hbox{E}\kern-.125em}}

\allowdisplaybreaks
\let\originalleft\left
\let\originalright\right
\renewcommand{\left}{\mathopen{}\mathclose\bgroup\originalleft}
\renewcommand{\right}{\aftergroup\egroup\originalright}

\newboolean{showcomments}
\newboolean{showcolors}
\setboolean{showcomments}{true}
\setboolean{showcolors}{true}

\definecolor{todocolor}{rgb}{1,0,0}
\definecolor{cmtcolor}{rgb}{0.1,0.1,0.9}
\definecolor{cmtqcolor}{HTML}{0959aa}

\def\ldelim{(}
\def\rdelim{)}

\newcommand{\IfNE}[2]{\ifthenelse{\isempty{#1}}{}{#2}}
\newcommand{\IfENE}[3]{\ifthenelse{\isempty{#1}}{#2}{#3}}
\NewDocumentCommand{\Sub}{s m O{} O{}}{\IfBooleanTF{#1}{\IfENE{#2}{_{#3#4}}{_{#3#2#4}}}{\IfNE{#2}{_{#3#2#4}}}}
\NewDocumentCommand{\Sup}{s m O{} O{}}{\IfBooleanTF{#1}{\IfENE{#2}{^{#3#4}}{^{#3#2#4}}}{\IfNE{#2}{^{#3#2#4}}}}

\newcommand{\negphantom}[1]{\ifmmode\settowidth{\dimen0}{$#1$}\else\settowidth{\dimen0}{#1}\fi\hspace*{-\dimen0}}

\protected\def\<{%
  \ifmmode
    \mskip0.5\thinmuskip
  \else
    \ifhmode
      \kern0.08334em
    \fi
  \fi
}

\makeatletter
\renewcommand\paragraph{\@startsection{paragraph}{8}{\z@}%
  {1ex \@plus1ex \@minus.2ex}%
  {-.1em}%
  {\normalfont\normalsize\itshape}}
\makeatother

\makeatletter
\renewcommand{\paragraph}{\@startsection{paragraph}{8}{\z@}%
  {1ex \@plus1ex \@minus.2ex}%
  {-1em}%
  {\normalfont\normalsize\bfseries}}
\makeatother
\makeatletter
\titleformat{\paragraph}[runin]{\normalfont\itshape}{}{}{\theparagraph #1.}
\makeatother

\makeatletter
\renewenvironment{proof}[1]{\par
  \pushQED{\qed}%
  \normalfont \topsep6\p@\@plus6\p@\relax
  \trivlist
  \item\relax
  {\itshape
  \proofname\ifthenelse{\equal{#1}{}}{\@addpunct{.}}{ (#1)\@addpunct{.}}}\hspace\labelsep\ignorespaces
}{%
  \popQED\endtrivlist\@endpefalse \vskip8\p@\@plus6\p@\relax
}
\makeatother

\makeatletter
\newenvironment{subproof}[1]{\par
  \pushQED{\qed}
  \normalfont \topsep6\p@\@plus6\p@\relax
  \trivlist
  \item\relax
  {\itshape
  \proofname\ifthenelse{\equal{#1}{}}{\@addpunct{.}}{ (#1)\@addpunct{.}}}\hspace\labelsep\ignorespaces
}{%
  \popQED\endtrivlist\@endpefalse \vskip8\p@\@plus6\p@\relax
}
\makeatother



\newtheoremstyle{claim}
{3pt}
{3pt}
{\itshape}
{}
{\bfseries}
{.}
{.5em}
{}

\newtheoremstyle{property}
{3pt}
{3pt}
{\itshape}
{}
{}
{}
{.5em}
{}

\newtheoremstyle{LemmaRepeat}
{3pt}
{3pt}
{\itshape}
{}
{\bfseries}
{.}
{.5em}
{\thmname{#1}\thmnote{ #3}}

\newtheoremstyle{mystyle}
  {}
  {}
  {\itshape}
  {}
  {\bfseries}
  {.}
  { }
  {\thmname{#1}\thmnumber{ #2}\thmnote{ #3}}

\theoremstyle{mystyle}

\newtheorem{theorem}{Theorem}[section]

\newtheorem{lemma}[theorem]{Lemma}

\newtheorem{fact}{Fact}[section]
\newtheorem{remark}{Remark}[section]
\newtheorem*{theorem*}{Restatement of Theorem}
\newtheorem*{lemma*}{Restatement of Lemma}
\newtheorem*{claim*}{Restatement of Claim}
\newtheorem*{corollary*}{Restatement of Corollary}
\newtheorem*{fact*}{Restatement of Fact}
\newtheorem{definition}{Definition}[section]
\theoremstyle{claim}
\newtheorem{claim}[theorem]{Claim}
\theoremstyle{property}

\theoremstyle{LemmaRepeat}

\newcommand{\LemmaLabel}[1]{(#1)}

\makeatletter
\renewenvironment{remark}[1][]{\par
  \pushQED{\qed}%
  \normalfont \topsep6\p@\@plus6\p@\relax
  \trivlist
  \item\relax
  \refstepcounter{remark}
  {{\bfseries Remark \theremark}\ifthenelse{\equal{#1}{}}{\@addpunct{.}}{ (#1)\@addpunct{.}}}\hspace\labelsep\ignorespaces
}{%
  \popQED\endtrivlist\@endpefalse \vskip8\p@\@plus6\p@\relax
}
\makeatother

\DontPrintSemicolon

\SetKwInput{Given}{Given}
\SetKwInput{Solution}{Solution}
\SetKwInput{Argument}{Argument}
\SetKwInput{Require}{Require}
\SetKwInput{Let}{Let}
\SetKwInput{Fix}{Fix}
\SetKwInput{Define}{Define}
\SetKwInput{Design}{Design}
\SetKwInput{Construct}{Construct}

\SetKwIF{If}{ElseIf}{Else}%
  {if}%
  {then:}%
  {else if}{else}{endif}
\SetKwFor{ForEach}%
  {for each}%
  {do:}{endfch}
\SetKwFor{ForAll}%
  {for all}%
  {do:}%
  {endfall}

\SetKw{LetStatement}{let:}
\SetKw{InitStatement}{initialize:}
\SetKw{GivenStatement}{given:}
\SetKw{ConstructStatement}{construct:}
\SetKw{CalculateStatement}{calculate:}
\SetKw{EstimateStatement}{estimate:}
\SetKw{FindStatement}{find:}
\SetKw{SolveStatement}{solve:}
\SetKwProg{Stage}{Stage}{:}{end}

\SetSideCommentLeft
\SetNoFillComment
\SetCommentSty{mycommentfnt}
\SetKwComment{Comment}{$\textcolor{black}{\blacktriangleright}$ \ }{}

\newcommand{\smallsquare}{%
  \text{\fboxsep=-.2pt\raisebox{0.75mm}{\fbox{\rule{0pt}{0.8ex}\rule{0.8ex}{0pt}}}}%
}

\newlist{itemizetriangle}{itemize}{10}
\setlist[itemizetriangle,1]{label=$\textcolor{black}{\blacktriangleright}$}
\setlist[itemizetriangle,2]{label=\textbullet}
\setlist[itemizetriangle,3]{label=\textendash}
\setlist[itemizetriangle,4]{label=$\textcolor{black}{\triangleright}$}

\newlist{Itemize}{itemize}{10}
\setlist[Itemize,1]{label=$\textcolor{black}{\triangleright}$,nosep}
\setlist[Itemize,2]{label=\textendash}
\setlist[Itemize,3]{label=$\textcolor{black}{\diamond}$}
\setlist[Itemize,4]{label=$\textcolor{black}{\smallsquare}$}
\setlist[Itemize,5]{label=$\textcolor{black}{\circ}$}

\newlist{Enumerate}{enumerate}{5}
\setlist[Enumerate,1]{label={(\roman*)}}
\setlist[Enumerate,2]{label={(alph*)}}
\setlist[Enumerate,3]{label={(\arabic*)}}

\newlist{EnumerateText}{enumerate}{5}
\setlist[EnumerateText,1]{label={(\alph*)}}
\setlist[EnumerateText,2]{label={(\arabic*)}}
\setlist[EnumerateText,3]{label={(\roman*)}}

\newlist{EnumerateThm}{enumerate}{5}
\setlist[EnumerateThm,1]{label={\textup{\thetheorem{} (\arabic*)}}, leftmargin=0.08\linewidth}
\setlist[EnumerateThm,2]{label={(\roman*)}}
\setlist[EnumerateThm,3]{label={(\alph*)}}

\newlist{EnumeratePropRef}{enumerate}{5}
\setlist[EnumeratePropRef,1]{label={\textup{\Prop \thetheorem{} (\arabic*) (HD)}}, leftmargin=0.18\linewidth,first=\itshape}
\setlist[EnumeratePropRef,2]{label={(\roman*)}}
\setlist[EnumeratePropRef,3]{label={(\alph*)}}

\newlist{EnumerateSub}{enumerate}{5}
\setlist[EnumerateSub,1]{label={(\roman*)}}
\setlist[EnumerateSub,2]{label={(\arabic*)}} 
\setlist[EnumerateSub,3]{label={(\alph*)}}

\newlist{EnumerateAlg}{enumerate}{5}
\setlist[EnumerateAlg,1]{label={\bfseries\small\arabic*:},left=0pt,itemindent=0pt}
\setlist[EnumerateAlg,2]{label={(\roman*)}}
\setlist[EnumerateAlg,3]{label={(\alph*)}}

\newlist{EnumerateInline}{enumerate*}{3}
\setlist[EnumerateInline,1]{label={(\roman*)}}
\setlist[EnumerateInline,2]{label={(\alph*)}}
\setlist[EnumerateInline,3]{label={(arabic*)}}

\newcommand{\Enum}[2][]{\IfENE{#1}{(#2)}{\Label{#1}{(#2)}}\xspace}
\newcommand{\TagEqn}{\stepcounter{equation}\tag{\theequation}}

\MakeRobust{\ref}

\makeatletter
\NewDocumentCommand{\Label}{s m m}{%
  \@bsphack
  \csname phantomsection\endcsname 
  \IfBooleanF{#1}{#3}%
  \def\@currentlabel{#3}{#2}
  \@esphack
}
\makeatother

\makeatletter
\newcommand{\subalign}[1]{%
  \vcenter{%
    \Let@ \restore@math@cr \default@tag
    \baselineskip\fontdimen10 \scriptfont\tw@
    \advance\baselineskip\fontdimen12 \scriptfont\tw@
    \lineskip\thr@@\fontdimen8 \scriptfont\thr@@
    \lineskiplimit\lineskip
    \ialign{\hfil$\m@th\scriptstyle##$&$\m@th\scriptstyle{}##$\hfil\crcr
      #1\crcr
    }%
  }%
}
\makeatother

\DeclareMathAlphabet{\mathsfit}{T1}{\sfdefault}{\mddefault}{\sldefault}
\SetMathAlphabet{\mathsfit}{bold}{T1}{\sfdefault}{\bfdefault}{\sldefault}

\newcommand{\R}{\mathbb{R}}

\newcommand{\Z}{\mathbb{Z}}

\newcommand{\ZminusO}[1][]{\IfENE{#1}{\Z \setminus \{0\}}{( \Z \setminus \{0\} )^{#1}}}
\newcommand{\RminusO}[1][]{\IfENE{#1}{\R \setminus \{0\}}{( \R \setminus \{0\} )^{#1}}}

\newcommand{\OperatorName}[1]{\operatorname{\mathsf{#1}}}

\NewDocumentCommand{\RV}{+m O{}}{#1\Sub{#2}}
\newcommand{\Distr}[1]{\mathcal{#1}}

\newcommand{\Set}[1]{\mathcal{#1}}

\newcommand{\N}{\mathcal{N}}

\makeatletter
\newcommand\niton{\mathrel{\m@th\mathpalette\canc@l\owns}}
\newcommand\canc@l[2]{{\ooalign{$\hfil#1/\mkern1mu\hfil$\crcr$#1#2$}}}
\makeatother

\newcommand{\T}{\mathsf{T}}

\newcommand{\IOp}{\mathbb{I}}
\newcommand{\IpmOp}{\mathbb{I}_{\pm}}

\NewDocumentCommand{\Log}{s O{} D(){}}{{%
  \ifthenelse{\isempty{#3}}{\log}{
    {
      \log\IfNE{#2}{_{#2}}
      \mathchoice
        {\left( #3 \right)} 
        {\IfBooleanTF{#1}{\left( #3 \right)}{(#3)}} 
        {\left( #3 \right)}
        {\left( #3 \right)}
}}}}

\NewDocumentCommand{\Sign}{t. s O{} D(){}}{{%
  \def\dfmt{\left( #4 \right)}
  \def\tfmt{( #4 )}
  \IfBooleanF{#1}{\,}
  \SignOp\Sub{#3}
  \IfNE{#4}{
    \IfBooleanTF{#2}{\tfmt}{
      \mathchoice
        {\dfmt} 
        {\tfmt} 
        {\tfmt} 
        {\tfmt} 
}}}}

\NewDocumentCommand{\SignO}{s D(){}}{{%
  \ifthenelse{\isempty{#2}}{}{
    \IfBooleanTF{#1}{\SignOOp(#2)}{
      \mathchoice
        {\SignOOp\left( #2 \right)} 
        {\SignOOp(#2)} 
        {\SignOOp(#2)} 
        {\SignOOp(#2)} 
}}}}

\NewDocumentCommand{\Supp}{s D(){}}{{%
  \def\dfmt{\left( #2 \right)}
  \def\tfmt{( #2 )}
  \SuppOp%
  \IfBooleanTF{#1}{\tfmt}{
    \mathchoice
      {\tfmt} 
      {\tfmt} 
      {\tfmt} 
      {\tfmt} 
}}}

\newcommand{\E}{\operatornamewithlimits{\mathbb{E}}}

\NewDocumentCommand{\I}{s D(){}}{{%
  \IOp%
  \ifthenelse{\isempty{#2}}{}{
    \IfBooleanTF{#1}{(#2)}{
      \mathchoice
        {\left( #2 \right)} 
        {(#2)} 
        {(#2)} 
        {(#2)} 
}}}}

\NewDocumentCommand{\Ipm}{s D(){}}{{%
  \IpmOp%
  \ifthenelse{\isempty{#2}}{}{
    \IfBooleanTF{#1}{(#2)}{
      \mathchoice
        {\left( #2 \right)} 
        {(#2)} 
        {(#2)} 
        {(#2)} 
}}}}

\NewDocumentCommand{\BigO}{s t' D(){}}{{%
  \def\dfmt{\IfBooleanTF{#1}{(#3)}{\left( #3 \right)}}
  \def\tfmt{\IfBooleanTF{#1}{\left( #3 \right)}{(#3)}}
  \IfBooleanTF{#2}{\BigOOp'}{\BigOOp}
  \IfBooleanTF{#1}{\tfmt}{
    \mathchoice
      {\dfmt} 
      {\tfmt} 
      {\tfmt} 
      {\tfmt} 
}}}

\NewDocumentCommand{\BigOmega}{s t' D(){}}{{%
  \def\dfmt{\IfBooleanTF{#1}{(#3)}{\left( #3 \right)}}
  \def\tfmt{\IfBooleanTF{#1}{\left( #3 \right)}{(#3)}}
  \IfBooleanTF{#2}{\BigOmegaOp'}{\BigOmegaOp}
  \IfBooleanTF{#1}{\tfmt}{
    \mathchoice
      {\dfmt} 
      {\tfmt} 
      {\tfmt} 
      {\tfmt} 
}}}

\NewDocumentCommand{\BigTheta}{s t' D(){}}{{%
  \def\dfmt{\IfBooleanTF{#1}{(#3)}{\left( #3 \right)}}
  \def\tfmt{\IfBooleanTF{#1}{\left( #3 \right)}{(#3)}}
  \IfBooleanTF{#2}{\BigThetaOp'}{\BigThetaOp}
  \IfBooleanTF{#1}{\tfmt}{
    \mathchoice
      {\dfmt} 
      {\tfmt} 
      {\tfmt} 
      {\tfmt} 
}}}

\NewDocumentCommand{\LittleO}{s t' D(){}}{{%
  \IfBooleanTF{#2}{\LittleOOp'}{\LittleOOp}
  \ifthenelse{\isempty{#3}}{}{
    \IfBooleanTF{#1}{(#3)}{
      \mathchoice
        {\left( #3 \right)} 
        {(#3)} 
        {(#3)} 
        {(#3)} 
}}}}

\NewDocumentCommand{\LittleOmega}{s t' D(){}}{{%
  \IfBooleanTF{#2}{\LittleOmegaOp'}{\LittleOmegaOp}
  \ifthenelse{\isempty{#3}}{}{
    \IfBooleanTF{#1}{(#3)}{
      \mathchoice
        {\left( #3 \right)} 
        {(#3)} 
        {(#3)} 
        {(#3)} 
}}}}

\NewDocumentCommand{\Paren}{s O{} m}{{%
  #2%
  \ifthenelse{\isempty{#3}}{}{
    \IfBooleanTF{#1}{(#3)}{
      \mathchoice
        {\left( #3 \right)} 
        {(#3)} 
        {(#3)} 
        {(#3)} 
}}}}

\NewDocumentCommand{\BigOOp}{t'}{\IfBooleanTF{#1}{\tilde{O}}{O}}
\NewDocumentCommand{\BigThetaOp}{t'}{\IfBooleanTF{#1}{\tilde{\Theta}}{\Theta}}
\NewDocumentCommand{\BigOmegaOp}{t'}{\IfBooleanTF{#1}{\tilde{\Omega}}{\Omega}}
\newcommand{\LittleOOp}{o}
\newcommand{\LittleOmegaOp}{\omega}

\newcommand{\SignOp}{\OperatorName{sign}}
\newcommand{\SignOOp}{\OperatorName{sign}_{0}}
\newcommand{\SuppOp}{\OperatorName{supp}}

\NewDocumentCommand{\Ker}{s O{}}{\IfBooleanTF{#1}{\widehat{\OperatorName{ker}}\Sub{#2}}{\OperatorName{ker}}}

\NewDocumentCommand{\Proj}{s m D(){}}{{%
  \def\dfmt{\IfNE{#3}{\left( #3 \right)}}
  \def\tfmt{\IfNE{#3}{( #3 )}}
  \mathsf{proj}_{#2}
  \IfBooleanTF{#1}{\tfmt}{
    \mathchoice
      {\dfmt} 
      {\tfmt} 
      {\tfmt} 
      {\tfmt} 
}}}

\newcommand{\lnorm}[1]{\ell_{#1}}

\RenewDocumentCommand{\Vec}{s +m}{\IfBooleanTF{#1}{#2}{{\bm{{\mathrm{#2}}}}}}
\NewDocumentCommand{\SVec}{s +m O{} +m}{\IfBooleanTF{#1}{#2^{\IfNE{#3}{(#3);}#4}}{\bm{\mathrm{#2}}^{\IfNE{#3}{(#3);}#4}}}
\NewDocumentCommand{\BVec}{s O{} +m O{}}{\IfNE{#4}{\,(}\IfBooleanTF{#1}{(\mathbf{1}^{#3})\S{#2}}{\bm{\mathrm{1}}^{\IfNE{#2}{(#2);}#3}}\IfNE{#4}{)_{#4}}}
\NewDocumentCommand{\SMat}{s +m O{} +m}{\IfBooleanTF{#1}{#2^{\IfNE{#3}{(#3);}#4}}{\bm{\mathrm{#2}}^{\IfNE{#3}{(#3);}#4}}}
\NewDocumentCommand{\PMat}{s +m O{} +m}{\IfBooleanTF{#1}{#2_{\IfNE{#3}{(#3);}(#4)}}{\bm{\mathrm{#2}}_{\IfNE{#3}{(#3);}(#4)}}}

\NewDocumentCommand{\Mat}{s +m}{\IfBooleanTF{#1}{#2}{{\bm{{\mathrm{#2}}}}}}
\NewDocumentCommand{\MatRow}{s +m +m O{}}{\IfBooleanTF{#1}{#2}{\bm{\mathrm{#2}}}\IfENE{#4}{_{#3}}{_{#3,#4}}}
\NewDocumentCommand{\MatCol}{s +m +m O{}}{\IfBooleanTF{#1}{#2}{\bm{\mathrm{\underline{#2}}}}_{#3}\IfNE{#4}{^{#4}}}
\NewDocumentCommand{\Submat}{s +m +m}{{\bm{\mathrm{#2}}}^{(#3)}}
\NewDocumentCommand{\SubmatRow}{s +m +m +m O{}}{\IfBooleanTF{#1}{#2}{\bm{\mathrm{#2}}}\IfENE{#5}{_{#4}}{_{#4,#5}}^{(#3)}}

\newcommand{\Var}{\OperatorName{Var}}
\newcommand{\Cov}{\OperatorName{Cov}}

\newcommand{\iid}[1][~]{i.i.d.#1}
\newcommand{\MGF}[1][~]{mgf\xspace}

\newcommand{\SphereSym}{S}
\newcommand{\Sphere}[2][]{\ifthenelse{\isempty{#1}}{\SphereSym^{#2-1}}{\mathcal{S}^{#2}}}

\newcommand{\SparseSphereSubspace}[2]{{\SphereSym^{#2-1} \cap \SparseSubspace{#1}{#2}}}

\newcommand{\SparseSubspace}[2]{\Sigma_{#1}^{#2}}

\NewDocumentCommand{\Th}{t'}{\IfBooleanTF{#1}{^{\prime\mathrm{th}}}{\!\,^{\mathrm{th}}}}
\newcommand{\defeq}{\triangleq}

\NewDocumentCommand{\dCmt}{s O{\qquad} m}{#2\blacktriangleright \IfBooleanTF{#1}{\text{#3}}{\IfNE{#3}{\text{#3} \ }}}
\NewDocumentCommand{\dCmtx}{s O{} m}{#2\dCmtIndent \IfBooleanTF{#1}{\text{#3}}{\IfNE{#3}{\text{#3} \ }}}
\newcommand{\dCmtIndent}{\phantom{\qquad \blacktriangleright \ \ }}

\newcommand{\cIf}{\text{if}\ }
\newcommand{\cOtherwise}{\text{otherwise} }

\NewDocumentCommand{\THEOREM}{s O{~}}{\IfBooleanTF{#1}{Theorem}{Theorem}#2\ignorespaces}
\NewDocumentCommand{\COROLLARY}{s O{~}}{\IfBooleanTF{#1}{Corollary}{Corollary}#2\ignorespaces}
\NewDocumentCommand{\LEMMA}{s O{~}}{\IfBooleanTF{#1}{Lemma}{Lemma}#2\ignorespaces}
\NewDocumentCommand{\PROP}{s O{~}}{\IfBooleanTF{#1}{Prop}{Prop}#2\ignorespaces}
\NewDocumentCommand{\PROPS}{s O{~}}{\IfBooleanTF{#1}{Propositions}{Propositions}#2\ignorespaces}
\NewDocumentCommand{\FACT}{s O{~}}{\IfBooleanTF{#1}{Fact}{Fact}#2\ignorespaces}
\NewDocumentCommand{\FACTS}{s O{~}}{\IfBooleanTF{#1}{Facts}{Facts}#2\ignorespaces}
\NewDocumentCommand{\CLAIM}{s O{~}}{\IfBooleanTF{#1}{Claim}{Claim}#2\ignorespaces}
\NewDocumentCommand{\REMARK}{s O{~}}{\IfBooleanTF{#1}{Remark}{Remark}#2}
\NewDocumentCommand{\EQUATION}{s O{~}}{\IfBooleanTF{#1}{Equation}{Equation}#2\ignorespaces}
\NewDocumentCommand{\DEFINITION}{s O{~}}{\IfBooleanTF{#1}{Definition}{Definition}#2\ignorespaces}
\NewDocumentCommand{\ALGORITHM}{s O{~}}{\IfBooleanTF{#1}{Algorithm}{Algorithm}#2\ignorespaces}
\NewDocumentCommand{\LINE}{s O{~}}{\IfBooleanTF{#1}{Line}{Line}#2\ignorespaces}
\NewDocumentCommand{\REGIME}{s O{~}}{\IfBooleanTF{#1}{Regime}{Regime}#2\ignorespaces}
\NewDocumentCommand{\REGIMES}{s O{~}}{\IfBooleanTF{#1}{Regimes}{Regimes}#2\ignorespaces}
\NewDocumentCommand{\PROPERTY}{s O{~}}{\IfBooleanTF{#1}{Property}{Property}#2\ignorespaces}
\NewDocumentCommand{\PROPERTIES}{s O{~}}{\IfBooleanTF{#1}{Properties}{Properties}#2\ignorespaces}
\NewDocumentCommand{\ASSUMPTION}{s O{~}}{\IfBooleanTF{#1}{Assumption}{Assumption}#2\ignorespaces}
\NewDocumentCommand{\ASSUMPTIONS}{s O{~}}{\IfBooleanTF{#1}{Assumptions}{Assumptions}#2\ignorespaces}
\NewDocumentCommand{\STEP}{s O{~}}{\IfBooleanTF{#1}{Step}{Step}#2\ignorespaces}
\NewDocumentCommand{\CASE}{s O{~}}{\IfBooleanTF{#1}{Case}{Case}#2\ignorespaces}

\NewDocumentCommand{\THEOREMS}{s O{~}}{\IfBooleanTF{#1}{Theorems}{Theorems}#2\ignorespaces}
\NewDocumentCommand{\COROLLARIES}{s O{~}}{\IfBooleanTF{#1}{Corollaries}{Corollaries}#2\ignorespaces}
\NewDocumentCommand{\LEMMAS}{s O{~}}{\IfBooleanTF{#1}{Lemmas}{Lemmas}#2\ignorespaces}
\NewDocumentCommand{\CLAIMS}{s O{~}}{\IfBooleanTF{#1}{Claims}{Claims}#2\ignorespaces}
\NewDocumentCommand{\REMARKS}{s O{~}}{\IfBooleanTF{#1}{Remarks}{Remarks}#2\ignorespaces}
\NewDocumentCommand{\EQUATIONS}{s O{~}}{\IfBooleanTF{#1}{Equations}{Equations}#2\ignorespaces}
\NewDocumentCommand{\DEFINITIONS}{s O{~}}{\IfBooleanTF{#1}{Definitions}{Definitions}#2\ignorespaces}
\NewDocumentCommand{\ALGORITHMS}{s O{~}}{\IfBooleanTF{#1}{Algorithms}{Algorithms}#2\ignorespaces}
\NewDocumentCommand{\LINES}{s O{~}}{\IfBooleanTF{#1}{Lines}{Lines}#2\ignorespaces}
\NewDocumentCommand{\SECTION}{s O{~}}{\IfBooleanTF{#1}{Section}{Section}#2\ignorespaces}
\NewDocumentCommand{\SECTIONS}{s O{~}}{\IfBooleanTF{#1}{Sections}{Sections}#2\ignorespaces}
\NewDocumentCommand{\APPENDIX}{s O{~}}{\IfBooleanTF{#1}{Appendix}{Appendix}#2\ignorespaces}
\NewDocumentCommand{\APPENDICES}{s O{~}}{\IfBooleanTF{#1}{Appendices}{Appendices}#2\ignorespaces}
\NewDocumentCommand{\STEPS}{s O{~}}{\IfBooleanTF{#1}{Steps}{Steps}#2\ignorespaces}
\NewDocumentCommand{\CASES}{s O{~}}{\IfBooleanTF{#1}{Cases}{Cases}#2\ignorespaces}

\newcommand{\cf}[1][~]{cf.#1\ignorespaces}
\newcommand{\see}{see,~\ignorespaces}

\newcommand{\tab}{\ensuremath{~~~~}}
\newcommand{\Tab}[1][1]{%
    \ifthenelse{#1>0}{\tab}{}%
    \ifthenelse{#1>1}{\tab}{}%
    \ifthenelse{#1>2}{\tab}{}%
    \ifthenelse{#1>3}{\tab}{}%
    \ifthenelse{#1>4}{\tab}{}%
    \ifthenelse{#1>5}{\tab}{}%
    \ifthenelse{#1>6}{\tab}{}%
    \ifthenelse{#1>7}{\tab}{}%
    \ifthenelse{#1>8}{\tab}{}%
    \ifthenelse{#1>9}{\tab}{}%
    \ifthenelse{#1>10}{\tab}{}%
    \ifthenelse{#1>11}{\tab}{}%
    \ifthenelse{#1>12}{\tab}{}%
    \ifthenelse{#1>13}{\tab}{}%
    \ifthenelse{#1>14}{\tab}{}%
    \ifthenelse{#1>15}{\tab}{}%
    \ifthenelse{#1>16}{\tab}{}%
    \ifthenelse{#1>17}{\tab}{}%
    \ifthenelse{#1>18}{\tab}{}%
    \ifthenelse{#1>19}{\tab}{}%
    \ifthenelse{#1>20}{\tab}{}%
}

\newcommand{\AlignSp}{\tab}

\NewDocumentCommand{\SBinom}{s +m +m}{\IfBooleanTF{#1}{\binom{#2}{#3}}{\binom{[#2]}{#3}}}

\newcommand{\+}{\phantom{+}}

\newcommand{\Ell}{\ell}

\newcommand{\Id}[1]{\Mat{I}_{#1}}

\newcommand{\DistSOp}[1]{d_{\Sphere{#1}}}

\NewDocumentCommand{\Dist}{s +m +m}{d\IfBooleanTF{#1}{( #2, #3 )}{\left( #2, #3 \right)}}
\NewDocumentCommand{\DistH}{s O{} +m +m}{{%
 \def\dfmt{\bigl( #3, #4 \bigr)}
 \def\tfmt{( #3, #4 )}
 d_{H}#2
 \IfBooleanTF{#1}{\dfmt}{
    \mathchoice
     {\tfmt} 
     {\tfmt} 
     {\tfmt} 
     {\tfmt} 
}}}
\NewDocumentCommand{\DistS}{s O{n} O{} +m +m}{{%
 \def\dfmt{\bigl( #4, #5 \bigr)}
 \def\tfmt{( #4, #5 )}
 \DistSOp{#2}#3
 \IfBooleanTF{#1}{\dfmt}{
    \mathchoice
     {\tfmt} 
     {\tfmt} 
     {\tfmt} 
     {\tfmt} 
}}}
\newcommand{\DistAngular}[2]{\theta_{#1,#2}}

\NewDocumentCommand{\LatestUpdate}{s +m}{\IfBooleanTF{#1}{(Latest update on: #2)}{(\emph{Latest update on: #2}.)}}

\newcommand{\Mid}[1]{\,#1\,}

\newcommand{\BIHT}{BIHT\xspace}
\NewDocumentCommand{\NBIHT}{s}{\IfBooleanTF{#1}{(normalized)}{normalized} \BIHT}

\newcommand{\Restriction}[2]{
 \def\dfmt{\left. #1 \right|_{#2}}
 \def\tfmt{#1|_{#2}}
  \mathchoice
   {\dfmt} 
   {\tfmt} 
   {\tfmt} 
   {\tfmt} 
}
\NewDocumentCommand{\VL}{+m O{}}{\Sub{#1\IfNE{#2}{,#2}}}

\newcommand{\IL}[2][]{^{(#2)#1}}
\NewDocumentCommand{\RVL}{+m O{}}{\Sub{#1\IfNE{#2}{,#2}}}

\newcommand{\Net}[2][]{\mathcal{C}_{#2\IfNE{#1}{;#1}}}
\newcommand{\BallNet}[2][]{\mathcal{D}_{#2\IfNE{#1}{;#1}}}
\newcommand{\Ball}[2][]{\mathcal{B}_{#2}\IfNE{#1}{^{(#1)}}}
\newcommand{\BallAngular}[2][]{\mathcal{B}_{\theta\leq #2}\IfNE{#1}{^{(#1)}}}
\NewDocumentCommand{\BallSparseSphere}{O{} +m D(){}}{\mathcal{B}_{#2}\IfNE{#1}{^{(#1)}}( #3 ) \cap \Sphere{n} \cap \SparseSubspace{k}{n}}

\newcommand{\ThresholdOp}{T\!}
\newcommand{\ThresholdMat}{\Mat{T}\!}
\newcommand{\ThresholdMatEntry}{\Mat*{T}\!}
\NewDocumentCommand{\Threshold}{s +m D(){}}{%
  \ThresholdOp_{#2}%
  \ifthenelse{\isempty{#3}}
  {}
  {\IfBooleanTF{#1}{( #3 )}{\left( #3 \right)}}
}

\NewDocumentCommand{\ThresholdSetMat}{s +m O{}}{%
  \IfBooleanTF{#1}{\ThresholdMatEntry_{\IfNE{#3}{#3;}#2}}{\ThresholdMat_{#2}}%
}

\NewDocumentCommand{\ThresholdSet}{s t' O{\Supp( \Vec{\uV} ) \cup } +m D(){}}{%
  \ThresholdOp_{\IfBooleanF{#2}{#3}#4}%
  \ifthenelse{\isempty{#5}}
  {}
  {\IfBooleanTF{#1}{( #5 )}{\left( #5 \right)}}
}
\NewDocumentCommand{\TS}{s +m D(){}}{%
  \ThresholdOp_{#2}%
  \ifthenelse{\isempty{#3}}
  {}
  {\IfBooleanTF{#1}{( #3 )}{\left( #3 \right)}}
}

\newcommand{\pointwise}{point-wise\xspace}

\newcommand{\SACap}[2][]{\mathsf{A}\Sub{#2\IfNE{#1}{,#1}}}

\newcommand{\RegularizedIncompleteBetaFunctionOp}[1]{I_{#1}}
\NewDocumentCommand{\RegularizedIncompleteBetaFunction}{s m m m}{{%
  \def\dfmt{\left( #3, #4 \right)}
  \def\tfmt{( #3, #4 )}
  \RegularizedIncompleteBetaFunctionOp{#2}
  \IfBooleanTF{#1}{\tfmt}{
    \mathchoice
      {\dfmt} 
      {\tfmt} 
      {\tfmt} 
      {\tfmt} 
}}}

\NewDocumentCommand{\RHS}{s}{\IfBooleanTF{#1}{RHS}{right-hand-side}\xspace}
\NewDocumentCommand{\LHS}{s}{\IfBooleanTF{#1}{LHS}{left-hand-side}\xspace}

\newcommand{\BinarySet}{\{ -1,1 \}}

\renewcommand{\k}{k}
\newcommand{\m}{m}
\newcommand{\mO}[1][\epsilonRAICadv]{m_{0}\IfNE{#1}{( #1 )}}

\newcommand{\n}{n}
\newcommand{\Iter}{t}
\newcommand{\IterX}{t'}
\newcommand{\Err}[1]{\varepsilon\IfNE{#1}{( #1 )}}
\newcommand{\epsilonX}{\epsilon}
\newcommand{\epsilonO}{\epsilon_{0}}

\newcommand{\tauX}{\tau}

\newcommand{\Eta}{\eta}

\newcommand{\epsilonRAICadv}{\delta}

\newcommand{\epsilonRAIC}{\epsilon'}

\newcommand{\kX}{k}

\newcommand{\qX}[1][]{q\Sub{#1}}
\newcommand{\rhoX}[1][]{\rho\Sub{#1}}
\newcommand{\rhoXX}[1][]{\rho'}
\newcommand{\rhoRAIC}{\rho'}

\newcommand{\aC}{a}
\newcommand{\bC}{b}
\NewDocumentCommand{\cC}{s t' t' O{}}{c\IfBooleanT{#2}{\IfBooleanTF{#3}{''}{'}}\Sub{#4}\IfBooleanF{#1}{\<}}

\newcommand{\cX}[1]{c\Sub{#1}}
\newcommand{\CC}{c}

\newcommand{\iIx}{i}

\newcommand{\jIx}{j}

\newcommand{\IX}{I}
\newcommand{\JX}{J}
\newcommand{\Jx}{J_{0}}

\newcommand{\lX}{\Ell}

\newcommand{\dX}{d}

\newcommand{\betaX}{\beta}

\newcommand{\phiX}[2]{\phi\Sub{#1} \IfNE{#2}{( #2 )}}
\newcommand{\varphiX}[1]{\Vec{\uV}}
\NewDocumentCommand{\xvarphiX}{s +m}{{%
  \def\dfmt{\IfNE{#2}{\left( \phantom{\bigl|} \negphantom{\bigl|} #2 \right)}}
  \def\tfmt{\IfNE{#2}{( #2 )}}
  \varphi
  \IfBooleanTF{#1}{\tfmt}{
    \mathchoice
      {\dfmt} 
      {\tfmt} 
      {\tfmt} 
      {\tfmt} 
}}}

\NewDocumentCommand{\QuI}{O{\Vec{\uV}} O{\IX}}{Q\Sub{#1\IfNE{#2}{,#2}}}

\NewDocumentCommand{\barZI}{s O{\IX}}{\IfBooleanTF{#1}{\bar{\ZX}}{\mathbf{\bar{\ZX}}}\Sub{#2}}
\NewDocumentCommand{\sI}{s s O{} O{\IX}}{\IfBooleanTF{#1}{s}{\mathbf{s}}\Sub{\IfBooleanTF{#2}{\Vec{\uV*}}{\Vec{\uV}},#4 \IfNE{#3}{;#3}}}
\NewDocumentCommand{\barZsI}{s O{} O{} O{\IX}}{\IfBooleanTF{#1}{\bar{\ZX}}{\mathbf{\bar{\ZX}}}\Sub{\IfNE{#3}{#3,}\Vec{\uV},#4\IfNE{#2}{;#2}}}
\NewDocumentCommand{\barZsIx}{s O{} O{} O{\IX'}}{\IfBooleanTF{#1}{\bar{\ZX}}{\mathbf{\bar{\ZX}}}\Sub{\IfNE{#3}{#3,}\Vec{\uV},#4\IfNE{#2}{;#2}}}
\NewDocumentCommand{\barZsIX}{s}{\frac{1}{\m} \sum_{\iIx \in \IX} \ThresholdSet{\JX}( \Vec{\ZX}\VL{\iIx} ) \IfBooleanTF{#1}{\sI**[\iIx]}{\sI*[\iIx]}}
\NewDocumentCommand{\barZsIXX}{s}{\frac{1}{\m} \sum_{\iIx \in \IX'} \ThresholdSet{\JX}( \Vec{\ZX}\VL{\iIx} ) \IfBooleanTF{#1}{\sI**[\iIx]}{\sI*[\iIx]}}
\NewDocumentCommand{\zuI}{s O{} O{\Vec{\uV}} O{\IX}}{\IfBooleanTF{#1}{s}{\mathbf{s}}\Sub{#3,#4\IfNE{#2}{;#2}}}

\newcommand{\rX}{r}

\newcommand{\uX}{u}
\newcommand{\uXx}{\hat{u}}
\newcommand{\uO}{u}

\newcommand{\vX}{v}
\newcommand{\vXx}{\hat{v}}
\newcommand{\wX}{w}

\newcommand{\zX}{z}

\newcommand{\TJZ}{\ThresholdSet{\JX}( \Vec{\ZX} )}

\newcommand{\TJZi}[1][\iIx]{\ThresholdSet{\JX}( \Vec{\ZX}\VL{#1} )}

\NewDocumentCommand{\uV}{s}{u\IfBooleanT{#1}{_{\ast}}}
\newcommand{\vV}{v}
\newcommand{\wV}{w}
\newcommand{\zV}{z}

\newcommand{\WV}{W}
\newcommand{\HV}{H}
\NewDocumentCommand{\HVx}{s O{} O{} O{\f,\Set{\ZX}}}{\IfBooleanTF{#1}{\HV}{\mathbf{\HV}}_{#4\IfNE{#2}{;#2}\IfNE{#3}{;#3}}}
\NewDocumentCommand{\WVx}{s O{} O{} O{1}}{\IfBooleanTF{#1}{\WV}{\mathbf{\WV}}_{#4\IfNE{#2}{;#2}\IfNE{#3}{;#3}}}
\NewDocumentCommand{\WVxx}{s O{} O{} O{2}}{\IfBooleanTF{#1}{\WV}{\mathbf{\WV}}_{#4\IfNE{#2}{;#2}\IfNE{#3}{;#3}}}

\newcommand{\YRV}[1][]{U\Sub{#1}}
\newcommand{\barYRV}[1][]{\bar{U}\Sub{#1}}
\newcommand{\WRV}{W}

\newcommand{\DRV}{D}
\newcommand{\Dx}{\DRV\RVL{1}}
\newcommand{\Dxx}{\DRV\RVL{2}}
\newcommand{\Dxu}[1][\JX]{\DRV'_{1;#1}( \Vec{\uV},\Vec{\uV} )}
\newcommand{\Dxxu}[1][\JX]{\DRV'_{2;#1}( \Vec{\uV},\Vec{\uV} )}
\newcommand{\DX}[3][\JX]{\DRV_{1;#1}( #2,#3 )}
\newcommand{\DXX}[3][\JX]{\DRV_{2;#1}( #2,#3 )}
\newcommand{\DXXu}[1][\JX]{\DRV_{2;#1}( \Vec{\uV},\Vec{\uV} )}

\newcommand{\US}[1][]{\Set{U}\IfNE{#1}{_{#1}}}
\newcommand{\HS}[1][]{\Set{H}_{\f;\Set{\ZX}\IfNE{#1}{;#1}}}

\newcommand{\ZX}{A}

\newcommand{\WX}{W}

\newcommand{\alphaX}[1][]{\alpha\Sub{#1}}

\NewDocumentCommand{\UX}{O{\betaX} O{\Vec{\uX}}}{\Set{U}_{#1}( #2 )}
\NewDocumentCommand{\VX}{O{\betaX} O{\Vec{\uX}}}{\Set{V}_{#1}( #2 )}

\newcommand{\AV}[1][]{A}
\newcommand{\AM}{A}
\NewDocumentCommand{\Resp}{t'}{\IfBooleanTF{#1}{\hat{y}}{y}}
\newcommand{\f}[1][]{f\Sub{#1}}

\newcommand{\gammaX}{\gamma}
\newcommand{\x}{x}
\newcommand{\xApprox}[1][]{\hat{x}}
\newcommand{\xApproxX}[1][]{\tilde{x}}

\newcommand{\yV}{y}
\newcommand{\xV}{x}

\newcommand{\MismatchSet}[1][]{\Set{M}}

\newcommand{\MismatchSetf}[1][]{\Set{M}_{\f}}

\newcommand{\h}[1][]{h_{\Mat{\AM}\IfNE{#1}{;#1}}}
\newcommand{\hf}[3][]{h_{\f;\Mat{\AM}\IfNE{#1}{;#1}}\IfNE{#2#3}{( #2, #3 )}}
\newcommand{\hff}[1][]{h_{\f,\f;\Mat{\AM}\IfNE{#1}{;#1}}}
\newcommand{\hZ}[1][]{h_{\Mat{\AM}\IfNE{#1}{;#1}}}
\newcommand{\hZf}[3][]{h_{\f;\Mat{\AM}\IfNE{#1}{;#1}}\IfNE{#2#3}{( #2, #3 )}}
\newcommand{\hZff}[1][]{h_{\f,\f;\Set{\ZX}\IfNE{#1}{;#1}}}

\newcommand{\fAx}{\f( \Vec{\x} )}

\newcommand{\fZix}{\f( \Vec{\x} )_{\iIx}}

\newcommand{\fZiu}{\f( \Vec{\uV} )_{\iIx}}

\NewDocumentCommand{\WI}{s O{\IX}}{\IfBooleanTF{#1}{W}{\mathbf{W}}\Sub{\!#2}}

\NewDocumentCommand{\Ru}{s O{} O{} O{\Vec{\uX}}}{\IfBooleanTF{#1}{R}{\mathbf{R}}_{#4\IfNE{#2}{;#2}\IfNE{#3}{;#3}}}
\NewDocumentCommand{\Rux}{s O{} O{\Vec{\uX}}}{\IfBooleanTF{#1}{\hat{R}}{\mathbf{\hat{R}}}_{#3\IfNE{#2}{;#2}}}

\NewDocumentCommand{\Xu}{s t' O{}}{\IfBooleanTF{#2}{X'_{#3}}{X_{#3}}}
\NewDocumentCommand{\barXu}{s t' O{} D(){}}{\IfBooleanTF{#2}{\bar{X}'_{#3}}{\bar{X}_{#3\IfNE{#4}{,#4}}}}
\NewDocumentCommand{\Yu}{s t' O{\Vec{\uX}}}{\IfBooleanTF{#2}{Y'_{#3}}{Y_{#3}}}
\NewDocumentCommand{\barYu}{s t' t' O{} O{\Vec{\uX}}}{\IfBooleanTF{#1}{\IfBooleanTF{#2}{\IfBooleanTF{#3}{\bar{Y}''_{#5\IfNE{#4}{;#4}}}{\bar{Y}'_{#5\IfNE{#4}{;#4}}}}{\bar{Y}_{#5\IfNE{#4}{;#4}}}}{\IfBooleanTF{#2}{\IfBooleanTF{#3}{\Vec{\bar{Y}}''_{#5\IfNE{#4}{;#4}}}{\Vec{\bar{Y}}'_{#5\IfNE{#4}{;#4}}}}{\Vec{\bar{Y}}_{#5\IfNE{#4}{;#4}}}}}
\NewDocumentCommand{\Yuv}{s t' t' O{} O{\Vec{\uX},\Vec{\vX}}}{\IfBooleanTF{#2}{\IfBooleanTF{#3}{Y''_{#5\IfNE{#4}{;#4}}}{Y'_{#5\IfNE{#4}{;#4}}}}{Y_{#5\IfNE{#4}{;#4}}}}
\NewDocumentCommand{\barYuv}{s t' t' O{} O{\Vec{\uX},\Vec{\vX}}}{\IfBooleanTF{#2}{\IfBooleanTF{#3}{\bar{Y}''_{#5\IfNE{#4}{;#4}}}{\bar{Y}'_{#5\IfNE{#4}{;#4}}}}{\bar{Y}_{#5\IfNE{#4}{;#4}}}}
\NewDocumentCommand{\barUu}{s t' O{\Vec{\uX}}}{\IfBooleanTF{#2}{\bar{U}'_{#3}}{\bar{U}_{#3}}}
\NewDocumentCommand{\Uuv}{s t' t' O{\jIx} O{\Vec{\uX},\Vec{\vX}}}{\IfBooleanTF{#2}{\IfBooleanTF{#3}{U''_{#5;#4}}{U'_{#5;#4}}}{U_{#5;#4}}}
\NewDocumentCommand{\barUuv}{s t' t' O{\Vec{\uX},\Vec{\vX}}}{\IfBooleanTF{#2}{\IfBooleanTF{#3}{\bar{U}''_{#4}}{\bar{U}'_{#4}}}{\bar{U}_{#4}}}
\NewDocumentCommand{\Xux}{s O{\Vec{\uX}}}{\hat{X}_{#2}}
\NewDocumentCommand{\Xuv}{s O{\Vec{\uX},\Vec{\vX}}}{X_{#2}}
\NewDocumentCommand{\Xuvx}{s O{\Vec{\uXx},\Vec{\vXx}}}{\hat{X}_{#2}}
\NewDocumentCommand{\Wu}{s O{} O{\Vec{\uX}}}{\IfBooleanTF{#1}{W}{\mathbf{W}}_{#3\IfNE{#2}{;#2}}}
\NewDocumentCommand{\ru}{s O{} O{\Vec{\uX}}}{\IfBooleanTF{#1}{r}{\mathbf{r}\Sub{#2}}}
\NewDocumentCommand{\yu}{s O{} O{\Vec{\uX}}}{\IfBooleanTF{#1}{y}{\mathbf{y}\Sub{#2}}}
\newcommand{\lu}[1][\Vec{\uX}]{\lX}
\newcommand{\tu}[1][\Vec{\uX}]{t}

\NewDocumentCommand{\Yx}{s O{} O{\Vec{\x}}}{\IfBooleanTF{#1}{Y}{\mathbf{Y}}_{#3\IfNE{#2}{;#2}}}
\NewDocumentCommand{\Wx}{s O{} O{\Vec{\x}}}{\IfBooleanTF{#1}{W}{\mathbf{W}}_{#3\IfNE{#2}{;#2}}}
\NewDocumentCommand{\YxO}{s O{} O{\Vec{\x}}}{\IfBooleanTF{#1}{Y}{\mathbf{Y}}_{\ast#3\IfNE{#2}{;#2}}}
\NewDocumentCommand{\Rx}{s O{} O{\Vec{\x}}}{\IfBooleanTF{#1}{R}{\mathbf{R}}_{#3\IfNE{#2}{;#2}}}

\newcommand{\ksparserealunit}{\( k \)-sparse, real-valued unit\xspace}

\newcommand{\topk}[1][\(  \k  \)]{top-#1\xspace}
\newcommand{\Topk}[1][\(  \k  \)]{Top-#1\xspace}
\newcommand{\ksubset}[1][\(  \k  \)]{subset\xspace}
\newcommand{\kSubset}[1][\(  \k  \)]{Subset\xspace}
\newcommand{\hardthresholding}{hard thresholding\xspace}

\newcommand{\topkhardthresholding}{top-\( k \) hard thresholding\xspace}

\newcommand{\pdf}[1][\Xu]{f_{#1}}

\newcommand{\Basis}{\Set{V}}
\newcommand{\basis}{v}

\newcommand{\errorO}{\epsilonX_{0}}

\newcommand{\URV}{U}

\newcommand{\ux}{u}
\newcommand{\vx}{v}
\newcommand{\wx}{w}
\newcommand{\wOx}{w_{0}}
\newcommand{\fx}[1]{f\Sub{#1}}

\NewDocumentCommand{\bO}{s m}{b_{0,#2}\IfBooleanF{#1}{\,}}
\NewDocumentCommand{\bX}{s t' t' O{} m}{b\IfBooleanT{#2}{\IfBooleanTF{#3}{''}{'}}_{\IfNE{#4}{#4,}#5}\IfBooleanF{#1}{\<}}
\NewDocumentCommand{\bI}{s t' t' m}{b\IfBooleanT{#2}{\IfBooleanTF{#3}{''}{'}}_{1,#4}\IfBooleanF{#1}{\,}}
\NewDocumentCommand{\bII}{s m}{b_{2,#2}\IfBooleanF{#1}{\,}}
\NewDocumentCommand{\bIII}{s m}{b_{3,#2}\IfBooleanF{#1}{\,}}

\NewDocumentCommand{\aX}{s t' t' O{} m}{a\IfBooleanT{#2}{\IfBooleanTF{#3}{''}{'}}_{\IfNE{#4}{#4,}#5}\IfBooleanF{#1}{\<}}

\NewDocumentCommand{\RAIC}{s s}{\IfBooleanTF{#1}{\IfBooleanTF{#2}{R}{r}estricted approximate invertibiity condition}{RAIC}\xspace}
\NewDocumentCommand{\RAICadv}{s s}{\IfBooleanTF{#1}{\IfBooleanTF{#2}{R}{r}estricted approximate invertibiity condition under adversarial noise}{RAIC under adversarial noise}\xspace}

\newcommand{\Tfrac}[2]{\frac{#1}{#2}}

\newcommand{\aValue}{16}
\newcommand{\bCValue}{379.1038}
\newcommand{\cXOneValue}{\sqrt{\frac{3\pi}{\bC}}\left( 1 + \frac{16\sqrt{2}}{3} \right)}
\newcommand{\cXOneValueX}{1.3469}
\newcommand{\cXOneValueXX}{1.3470}
\newcommand{\cXTwoValue}{\frac{3}{\bC}\left( 1 + \frac{4\pi}{3} + \frac{8\sqrt{3\pi}}{3} + 8\sqrt{6\pi} \right)}
\newcommand{\cXTwoValueX}{0.3806}
\newcommand{\cXTwoValueXX}{0.3807}
\newcommand{\cXThreeValue}{\frac{13\sqrt{\pi}}{\sqrt{\bC}}}
\newcommand{\cXThreeValueX}{1.1834}
\newcommand{\cXThreeValueXX}{1.1835}
\newcommand{\cXFourValue}{2+4\sqrt{\pi}}
\newcommand{\cXFourValueX}{9.0898}
\newcommand{\cXFourValueXX}{9.0899}
\newcommand{\cValue}{31.9999}
\newcommand{\cValueXX}{32}

\def\PREIMG{( \SparseSphereSubspace{\k}{\n} ) \times ( \SparseSphereSubspace{\k}{\n} )}

\newcommand{\DXXuUB}{%
  13\sqrt{\frac{\pi \epsilonRAICadv \tauX}{\bC}} + \left( 2 + 4\sqrt{\pi} \right) \tauX \sqrt{\Log( \Tfrac{2e}{\tauX} )}
}
\newcommand{\DXXuUBconst}{%
  \cX{3} \sqrt{\epsilonRAICadv \tauX} + \cX{4} \tauX \sqrt{\Log( \Tfrac{2e}{\tauX} )}
}

\newcommand{\DXXuUBconstX}{%
  \cX{3} \sqrt{\epsilonX \tauX} + \cX{4} \tauX \sqrt{\Log( \Tfrac{2e}{\tauX} )}
}

\newcommand{\rXValueX}{%
  \cX{3} \sqrt{\epsilonX \tauX} + \cX{4} \tauX \sqrt{\Log( \Tfrac{2e}{\tauX} )}
}

\newcommand{\rXValue}{%
  \frac{\cX{}}{\cX{2}} \left( \rXValueX \right)
}
\definecolor{LinkColor}{HTML}{667699}
\definecolor{CiteColor}{HTML}{507395}
\definecolor{URLColor}{rgb}{0.3,0.5,0.7}
\hypersetup{
    colorlinks=true,
    linkcolor=LinkColor,
    filecolor=magenta,
    urlcolor=URLColor, 
    citecolor=CiteColor,
    anchorcolor=blue
}

\definecolor{MathInlineColor}{HTML}{000000}

\everymath{\color{MathInlineColor}}

\def\lightgraycolor{0.9}
\def\mediumgraycolor{0.55}
\def\darkgraycolor{0.4}
\definecolor{light-gray}{gray}{\lightgraycolor}
\definecolor{medium-gray}{gray}{\mediumgraycolor}
\definecolor{dark-gray}{gray}{0.4}

\definecolor{cmtc}{rgb}{0.1,0.1,0.9}

\NewDocumentCommand{\StartComment}{s}{%
\IfBooleanTF{#1}{%
\color{light-gray}\everymath{\color{light-gray}}%
\def\graycolor{\lightgraycolor}
\definecolor{todocolor}{gray}{\graycolor}%
\definecolor{cmtcolor}{gray}{\graycolor}%
\definecolor{cmtqcolor}{gray}{\graycolor}%
\definecolor{MathInlineColor}{gray}{\graycolor}%
\definecolor{LinkColor}{gray}{\graycolor}%
\definecolor{CiteColor}{gray}{\graycolor}%
\definecolor{URLColor}{gray}{\graycolor}%
}{%
\color{medium-gray}\everymath{\color{medium-gray}}%
\def\graycolor{\mediumgraycolor}%
\definecolor{todocolor}{gray}{\graycolor}%
\definecolor{cmtcolor}{gray}{\graycolor}%
\definecolor{cmtqcolor}{gray}{\graycolor}%
\definecolor{MathInlineColor}{gray}{\graycolor}%
\definecolor{LinkColor}{gray}{\graycolor}%
\definecolor{CiteColor}{gray}{\graycolor}%
\definecolor{URLColor}{gray}{\graycolor}%
}}
\newcommand{\EndComment}{\color{black}%
\definecolor{todocolor}{rgb}{1,0,0}%
\definecolor{cmtcolor}{rgb}{0.1,0.1,0.9}%
\definecolor{cmtqcolor}{HTML}{0959aa}%
\definecolor{MathInlineColor}{HTML}{000000}%
\everymath{\color{MathInlineColor}}%
\definecolor{LinkColor}{rgb}{0.30,0.30,0.30}%
\definecolor{CiteColor}{rgb}{0.25,0.40,0.6}%
\definecolor{URLColor}{rgb}{0.3,0.5,0.7}%
}

\NewDocumentCommand{\StartRemark}{s}{%
\IfBooleanTF{#1}{%
\color{dark-gray}\everymath{\color{dark-gray}}%
\def\graycolor{\darkgraycolor}
\definecolor{todocolor}{gray}{\graycolor}%
\definecolor{cmtcolor}{gray}{\graycolor}%
\definecolor{cmtqcolor}{gray}{\graycolor}%
\definecolor{MathInlineColor}{gray}{\graycolor}%
\definecolor{LinkColor}{gray}{\graycolor}%
\definecolor{CiteColor}{gray}{\graycolor}%
\definecolor{URLColor}{gray}{\graycolor}%
}{%
\color{dark-gray}\everymath{\color{dark-gray}}%
\def\graycolor{\darkgraycolor}
\definecolor{todocolor}{gray}{\graycolor}%
\definecolor{cmtcolor}{gray}{\graycolor}%
\definecolor{cmtqcolor}{gray}{\graycolor}%
\definecolor{MathInlineColor}{gray}{\graycolor}%
\definecolor{LinkColor}{gray}{\graycolor}%
\definecolor{CiteColor}{gray}{\graycolor}%
\definecolor{URLColor}{gray}{\graycolor}%
}}

\newcommand{\xToDo}[1]{\normalfont\textcolor{todocolor}{$\textcolor{todocolor}{\ldelim}$\textbf{To\,do}\IfENE{#1}{\textbf{:}\;Fill this in.}{\textbf{:}\;{#1}}$\textcolor{todocolor}{\rdelim}$}}
\newcommand{\ToDo}[1]{%
  \colorlet{origcolor}{.}\everymath{\color{todocolor}}%
  \ifthenelse{\boolean{showcomments}}%
  {\relax\ifmmode\text{\xToDo{#1}}\else\xToDo{#1}\fi}%
  {}%
  \everymath{\color{MathInlineColor}}%
}%

\newcommand{\xCOMMENT}[1]{\normalfont\textcolor{cmtcolor}{$\textcolor{cmtcolor}{\ldelim}$\textbf{Note}\IfNE{#1}{\textbf{:}\;{#1}}$\textcolor{cmtcolor}{\rdelim}$}}
\newcommand{\COMMENT}[1]{%
  \colorlet{origcolor}{.}\everymath{\color{cmtcolor}}%
  \ifthenelse{\boolean{showcomments}}%
  {\relax\ifmmode\text{\xCOMMENT{#1}}\else\xCOMMENT{#1}\fi}%
  {}%
  \everymath{\color{MathInlineColor}}%
}%

\NewDocumentCommand{\?}{s +m}{%
  \colorlet{origcolor}{.}\everymath{\color{cmtqcolor}}%
  \ifthenelse{\boolean{showcomments}}%
  {\normalfont\textcolor{cmtqcolor}{{\IfBooleanF{#1}{$\textcolor{cmtqcolor}{\ldelim}$}#2$\textcolor{cmtqcolor}{\rdelim}$}}}%
  {}%
  \everymath{\color{MathInlineColor}}%
}%
\newcommand\blfootnote[1]{%
  \begingroup
  \renewcommand\thefootnote{}\footnote{#1}%
  \addtocounter{footnote}{-1}%
  \endgroup
}

\title{Robust 1-bit Compressed Sensing with  Iterative Hard Thresholding}
\date{}
 \author{Namiko Matsumoto \\
 \and Arya Mazumdar\blfootnote{The authors are with the University of California San Diego. This work is supported in part by NSF awards 2133484 and 2217058. Emails: \texttt{\{nmatsumoto,arya\}@ucsd.edu}.}}

\begin{document}
\maketitle

\begin{abstract}
In 1-bit compressed sensing, the aim is to estimate a $k$-sparse unit vector $x\in \Sphere{n}$ within an $\epsilon$ error (in $\ell_2$) from minimal number of linear measurements that are quantized to just their signs, i.e., from measurements of the form $y = \Sign(\langle a, x\rangle).$ In this paper, we study a noisy version where a fraction  of the measurements can be flipped, potentially by an adversary. In particular, we analyze the Binary Iterative Hard Thresholding (BIHT) algorithm, a proximal gradient descent on a properly defined loss function used for 1-bit  compressed sensing, in this noisy setting.   It is known from recent results that, with $\tilde{O}(\frac{k}{\epsilon})$ noiseless measurements, BIHT  provides an estimate within $\epsilon$ error. This result is optimal and universal, meaning one set of measurements work for all sparse vectors. In this paper, we show that BIHT also provides better results than all known methods for the noisy setting. We show that when up to  $\tau$-fraction of the sign measurements are incorrect (adversarial error), with the same number of measurements as before, BIHT agnostically provides an estimate of $x$ within an $\tilde{O}(\epsilon+\tau)$ error, maintaining the universality of measurements. This establishes stability of iterative hard thresholding in the presence of measurement error. To obtain the result, we use the restricted approximate invertibility of Gaussian matrices, as well as a tight analysis of the high-dimensional geometry of the adversarially corrupted measurements.
\end{abstract}
\section{Introduction}

Compressed sensing is a  framework in signal processing that exploits the inherent sparsity or compressibility of signals to efficiently acquire and reconstruct them with a sampling rate that is significantly lower than the dimensionality  of the signal~\cite{candes2006robust,donoho2006compressed}.  By using a small number of non-adaptive measurements, often obtained through random projections, compressed sensing enables the recovery of the original signal with high accuracy.

In real-world signal acquisitions and storage, signals are often digitized. This led to introduction to 1-bit compressed sensing (1bCS) by~\cite{DBLP:conf/ciss/BoufounosB08}. In this model, a unit-norm sparse signal $\Vec{x} \in S^{n-1}, \|\Vec{x}\|_0 \le k,$ is acquired through the  operation $\Vec{y} = \Sign(\Mat{\AM}\Vec{x}),$ where $\Mat{\AM}$ is an $m \times n$ real matrix and $\Vec{y}\in \{1,-1\}^m$ is a binary vector containing the coordinate-wise signs of $\Mat{\AM}\Vec{x}$. The primary objective is to design a measurement matrix $\Mat{\AM}$ with minimal number of rows $m,$ such that for any $\Vec{x} \in S^{n-1}, \|\Vec{x}\|_0 \le k$, an estimate $\hat{\Vec{x}}$ from $\Vec{y}$ and $\Mat{\AM}$ via an efficient algorithm can be provided such that $\|\Vec{x}-\hat{\Vec{x}}\|_2 \le \epsilon$,  for a given $0 \le \epsilon \le 1.$ We will refer to $\epsilon$ as the {\em parameter error}.

It is known that $m = \Omega(\frac{k}{\epsilon})$ measurements are necessary~\cite{jacques2013robust} for this. Also, if the entries of the matrix $\Mat{\AM}$ is chosen to be standard normal random random variables then recovery is possible with high probability for all $k$-sparse unit norm vectors with $m = O(\frac{k}{\epsilon} \log\frac{n}{\epsilon})$ measurements~\cite{jacques2013robust}. Hereafter, there has been a series of work that tries to achieve this baseline number of measurements $\tilde{O}(\frac{k}{\epsilon})$ with a computationally tractable algorithm, such as convex relaxations~\cite{plan2013one,boufounos2015quantization, plan2017high}. In particular, the {\em linear estimator} of \cite{plan2017high} shows that $\tilde{O}(\frac{k}{\epsilon^2})$ measurements are sufficient, which is suboptimal in its dependency on the parameter error. 

In this paper, we study a very natural iterative estimation method proposed in \cite{jacques2013robust}, called {\em binary iterative hard thresholding} (BIHT). Iterative hard thresholding is a well-known algorithm for compressed sensing, where estimations of $\Vec{x}$ are projected back to the ``top-$k$'' coordinates in each step to maintain sparsity of the solution~\cite{blumensath2009iterative}.
The description of BIHT is provided in Algorithm~\ref{alg:biht:normalized}, and will be formally discussed  later. In short, it is a proximal gradient descent algorithm where an estimate of  $\Vec{x}$ is updated iteratively followed by the aforementioned projection. BIHT was empirically observed to have excellent performance which was analyzed in several papers such as \cite{jacques2013quantized,boufounos2015quantization,liu2019one,friedlander2021nbiht}. Ultimately, in \cite{matsumoto2022binary}, it was shown that $\tilde{O}(\frac{k}{\epsilon})$ measurements are sufficient for BIHT to produce an estimate with at most $\epsilon$ error, i.e., the optimal dependence on sparsity and error.

In this paper, we show that iterative hard thresholding is in fact even more powerful: it is robust to adversarial noise. Noisy one-bit compressed sensing has also been quite well-studied in the last few years~\cite{awasthi2016learning,dai2016noisy,awasthi2017power,huang2018robust,chinot2022adaboost}. In particular, we assume that any up to $\tau m, 0 \le \tau \le 1,$ coordinates of the measurement vector $\Sign(\Mat{\AM}\Vec{x})$ are flipped by an adversary. In this model, \cite{plan2012robust} showed that their {\em linear estimator} can provide  an estimate $\hat{\Vec{x}}$ of $\Vec{x}$ such that $\|\hat{\Vec{x}}- \Vec{x}\|_2 \le \epsilon$ with $O(\frac{k}{\epsilon^4} (1/2-\tau)^{-2}\log\frac{2n}{k})$ measurements.
In the same model, \cite{awasthi2016learning} provided an algorithm that returns an estimate $\hat{\Vec{x}}$ of $\Vec{x}$ such that $\|\hat{\Vec{x}}- \Vec{x}\|_2 \le \epsilon+ c\tau,$ $c>0$ being a constant, with $\tilde{O}(\frac{k}{\epsilon^3})$ measurements.

To mitigate such sign-flips in measurements, an algorithm called adaptive outlier pursuit was proposed in \cite{yan2012robust} that shows superior performance over BIHT empirically. However the algorithm requires precise knowledge of $\tau$, and  performance deteriorates rapidly without this knowledge. On the other hand, BIHT is agnostic to the number of sign-flips. Another algorithm based on MAP estimation was proposed in \cite{dai2016noisy},  relying on a stable embedding property of the measurement matrix which is known to take at least $\Omega(\frac{k}{\epsilon^2})$ rows. With $\Omega(\frac{k}{\epsilon^2})$ measurements a least-square decoding algorithm was also shown to be effective in~\cite{huang2018robust} in the presence of sign-flips. More recently, the noisy 1-bit compressed sensing problem was also studied from both the perspective of parameter error, and prediction error; in particular the performance of the AdaBoost~(\cite{freund1997decision}) algorithm was analyzed in~\cite{chinot2022adaboost}. The number of required measurements here scales as $\tilde{O}(\frac{k}{\epsilon^6}).$ We have omitted the dependence on $\tau$ in the last few results for the sake of clarity, and also to point out suboptimal dependence on parameter error even in the absence of adversarial sign-flips.

\subsection{Our Contributions}
Under the adversarial sign-flip model described above, we show that BIHT still produces a good estimate of the sparse vector $\Vec{x}$ with the same number of $\tilde{O}(\frac{k}{\epsilon})$ measurements. BIHT is also agnostic to the number of sign-flips: indeed, as long as there is sufficient number of measurements, a good estimate with small parameter error is produced. To be precise, we show that with $m= O(\frac{k}{\epsilon}\log\frac{n}{k\epsilon})$ measurements of which up to $\tau m$ can be corrupted, BIHT converges to an estimate $\hat{\Vec{x}}$ of $\Vec{x}$, such that $$\|\hat{\Vec{x}}-\Vec{x}\|\le \epsilon + O(\sqrt{\epsilon \tau} + \tau \sqrt{\log\frac1\tau}) \asymp \max\{\epsilon, \tau \sqrt{\log \frac1\tau}\}. $$

With only $\tilde{O}(\frac{k}{\epsilon})$ measurements, this result provides the best sample complexity guarantee, i.e., a number of measurements with better dependence on parameter error than \cite{plan2012robust,awasthi2016learning,dai2016noisy,huang2018robust,chinot2022adaboost} mentioned above.

While our work builds on \cite{matsumoto2022binary}, our analysis requires  new insights as well as new technical tools. One of the key steps in \cite{matsumoto2022binary} is to establish a property of Gaussian matrices called {\em restricted approximate invertibility condition} (\RAIC). 
This condition ensures that the estimation error remains controlled throughout the iterations of BIHT by approximately preserving the discrepancy between two vectors and the average of the measurements (rows of matrix $\Mat{\AM}$) that yield distinct outcomes for those vectors.


In this paper, we aim to prove a similar condition but account for the possibility of flipped measurements. To achieve this, we introduce a new definition of RAIC with measurement error. Our main technical achievement is demonstrating that Gaussian matrices possess this property. For this, in addition to validating the results obtained by \cite{matsumoto2022binary} for Gaussian measurements, we also need to establish a (roughly) linear relationship between the expected norm of the sum of up to \(  \tauX \m  \)-many measurements and the expected error resulting from adversarial corruption of up to \(  \tauX \m  \)-many responses.
Consequently, given that the norm of the sum of any set of up to \(  \tauX \m  \)-many measurements can be consistently bounded and not exceed a certain threshold with a high probability, we can  establish an upper bound on the error introduced by adversarial noise with a high probability.
%
%
With the goal of upper bounding the norm of the sum of the up to \(  \tauX \m  \)-many measurements, the vector sum is orthogonally decomposed into two components: (a) its projection onto a particular vector \(  \Vec{\uV}  \) (determined later), and (b) its projection into the kernel of \(  \Vec{\uV}  \), each of which will be bounded separately.
The norm of the two components can be recombined via triangle inequality.
Repeating this over a collection of vectors, \(  \Vec{\uV}  \), leads to a uniform result.
Crucially, it turns out that the number of vectors, \(  \Vec{\uV}  \), which need to be considered in this collection is finite and quantifiable: it does not exceed then number of ways to choose up to \( 
 \tauX \m  \)-many responses to corrupt.
This is related to the tracking of ``mismatches,'' which was a key element in 
previous analysis.



\subsection{Other Related Works}

Without the sparsity constraint, the problem we consider is closely related to the noisy half-space learning problem. 
However, most of the time the focus of such works is to provide guarantee on prediction error, rather than parameter error~\cite{frei2021agnostic,ji2022agnostic}. The objective of this line of work is to come up with distribution-agnostic efficient algorithms, and then to provide guarantees on their zero-one loss (probability of mismatch). 
This problem is also studied with different noise models, for example, Massart noise~\cite{diakonikolas2019distribution}, instance-dependent noise~\cite{menon2018learning,cheng2020learning}, random sign-flip noise~\cite{diakonikolas2022learning}. In particular, the later work shows that with standard Gaussian covariates, and with probability of sign-flip being $\eta,$ one can come up with a classifier to guarantee probability of mismatch $O(\eta)$, with $\frac{n}{\eta^2}$ samples. Since for Gaussian covariates, parameter error and prediction error could be related - this will lead to a suboptimal sample complexity with respect to the error rate, if a ``sparse'' version could be made available. Active learning under this model was also considered in \cite{yan2017revisiting}.

Interestingly, learning $k$-sparse half-spaces where labels can be corrupted has been considered in \cite{zhang2018efficient,shen2021attribute}, with guarantee on the prediction error. Note that these papers study the problem in an ``active PAC learning'' setting, which is different from even the adaptive version of 1-bit compressed sensing. We point the reader to \cite{zhang2018efficient} for a detailed discussion on this difference. Furthermore, if the covariates/measurements were Gaussian, prediction error could be related to parameter error, but that is not the case in general. While in the active learning set up the number of label queries are small, the total sample complexity in \cite{shen2021attribute} scales quadratically with $k$, which is suboptimal in 1-bit compressed sensing.

In \cite{plan2012robust}  a more general sparse signal recovery problem was studied where the binary observations  $y_i \in \{+1,-1\}$ are random: i.e., $y_i=1$ with probability $f(\langle a_i, x \rangle), i =1, \dots, m$, where $f$ is a potentially nonlinear function, such as the logistic function. 



The support recovery problem in 1-bit compressed sensing and constructions of structured measurement matrices are well-studied, though not directly related to this work, e.g.,~\cite{GNJN13,ABK17,flodin2019superset,mazumdar2021support}. Other generalizations of 1-bit compressed sensing has also been recently studied, for example, with generative priors~\cite{liu2019one}.


\paragraph{Organization.} The rest of the paper is organized as follows. In  \SECTION \ref{section:|>prelim} we introduced the notations used in the paper, and also the BIHT algorithm. \SECTION \ref{section:main-results} contains the main result (Theorem~\ref{thm:adversarial}) and a technical overview of the proofs. Subsequently, proofs of the main results appear in \SECTION \ref{section:|>pf-main-thm} and \ref{section:|>pf-main-technical}. Intermediate, and longer proofs are delegated to the appendix.
%
\section{Preliminaries}  
\label{section:|>prelim} 

\subsection{Notations}              
\label{section:|>prelim|>notations} 



Throughout this work, the parameters \(  \k, \n \in \Z_{+}  \) are taken to satisfy
\(  \n \geq 2\k  \),
where \(  \k  \) denotes sparsity (i.e., the maximum number of nonzero entries in a vector),
and where \(  \n  \) is the dimension of the signal vectors and measurements.
The number of measurements (and rows in the measurement matrix) is denoted by \(  \m \in \Z_{+}  \).
For notational simplicity, the parameter \(  \tauX \in (0,1]  \)%
---the fraction of responses that can be corrupted---%
is assumed to satisfy
\(  \tauX \m \in \Z_{+}  \).
%
This does not forgo generality since \(  \tauX \m  \) can be replaced by
\(  \lceil \tauX \m \rceil  \) throughout the analysis in this manuscript.
Note that this work does not consider \(  \tauX = 0  \) since \cite{matsumoto2022binary}
already established the result under noiseless conditions.
%
\par 
%
For the purposes of this discussion, let \(  \Ell, \dX \in \Z_{+}  \),
where \(  \dX \in \Z_{+}  \) specifies an arbitrary dimension.
%
\par 
%
Let \(  \Distr{D}  \) be an arbitrary distribution.
Then, \(  \RV{X} \sim \Distr{D}  \) denotes a random variable
which follows the distribution \(  \Distr{D}  \).
Similarly, let \(  \Set{S}  \) be an arbitrary set.
Then, \(  \RV{X} \sim \Set{S}  \) denotes a random variable
which follows the uniform distribution over \(  \Set{S}  \).
The univariate normal distribution with
mean \(  \mu \in \R  \) and variance \(  \sigma^{2} \in \R_{\geq 0}  \)
is denoted by \(  \N( \mu, \sigma^{2} )  \),
while the \(  \dX  \)-variate normal distribution with
mean vector \(  \Vec{\mu} \in \R^{\dX}  \) and
covariance matrix \(  \Mat{\Sigma} \in \R^{\dX \times \dX}  \)
is denoted by \(  \N( \Vec{\mu}, \Mat{\Sigma} )  \).
%
\par 
%
The \(  \dX  \)-dimensional identity matrix is denoted by \(  \Id{\dX} \in \R^{\dX \times \dX}  \).
More generally, matrices are written as capital letters in boldface, upright typeface, e.g.,
\(  \Mat{A} \in \R^{\Ell \times \dX}  \),
with the \(  i\Th  \) rows denoted by, e.g.,
\(  \Mat{A}\VL{i} \in \R^{\dX}  \), \(  i \in [\Ell]  \), such that
\(  \Mat{A} = ( \Mat{A}\VL{1} \cdots \Mat{A}\VL{\Ell} )^{\T}  \),
and with the \(  (i,j)  \)-entries written in italic typeface, e.g.,
\(  \Mat*{A}_{i,j} \in \R  \).
%
Nonrandom vectors are written as lowercase letters in boldface, upright typeface, e.g.,
\(  \Vec{u} \in \R^{\dX}  \)
with the \(  j\Th  \) entries, \(  j \in [\dX]  \), written in italic typeface, e.g.,
\(  \Vec*{u}_{j} \in \R  \), such that
\(  \Vec{u} = ( \Vec*{u}_{1}, \dots, \Vec*{u}_{\dX} )  \).
%
Random vectors follow the same convention as nonrandom vectors but with uppercase letters, e.g.,
\(  \Vec{Z} = ( \Vec*{Z}_{1}, \dots, \Vec*{Z}_{\dX} ) \sim \N( \Vec{0}, \Id{\dX} )  \).
%
For \(  \JX \subseteq [\dX]  \), the restriction of \(  \Vec{u} \in \R^{\dX}  \) to the entries
indexed by \(  \JX  \) is denoted by
\(  \Restriction{\Vec{u}}{\JX} \in \R^{| \JX |}  \).
%
The support of a vector, \(  \Vec{u} \in \R^{\dX}  \), is denoted by
\(  \Supp( \Vec{u} ) \defeq \{ j \in [\dX] : \Vec*{u}_{j} \neq 0 \} \subseteq [\dX]  \),
and the number of nonzero entries in \(  \Vec{u}  \)%
---the \(  \lnorm{0}  \)-``norm'' of \(  \Vec{u}  \)---%
is written as
\(  \| \Vec{u} \|_{0} \defeq | \Supp( \Vec{u} ) |  \).
%
\par 
%
The \(  \lnorm{2}  \)-unit sphere in \(  \R^{\dX}  \) is denoted by
\(  \Sphere{\dX} \defeq \{ \Vec{u} \in \R^{\dX} : \| \Vec{u} \|_{2} = 1 \}  \),
and the set of \(  \k  \)-sparse \(  \dX  \)-dimensional vectors is written as
\(  \SparseSubspace{\k}{\dX} \defeq \{ \Vec{u} \in \R^{\dX} : \| \Vec{u} \|_{0} \leq \k \}  \).
%
Hence, the set of all \(  \dX  \)-dimensional, \ksparserealunit vectors is denote by
\(  \SparseSphereSubspace{\k}{\dX} \defeq
    \{ \Vec{u} \in \Sphere{\dX} : \| \Vec{u} \|_{0} \leq \k \}  \).
%
The distance between two points projected onto the \(  \lnorm{2}  \)-unit sphere is specified by
the function
\(  \DistSOp{\dX} : \R^{\dX} \times \R^{\dX} \to \R_{\geq 0}  \),
where
\begin{gather*}
  \DistS[\dX]{\Vec{u}}{\Vec{v}}
  =
  \begin{cases}
  \left\| \frac{\Vec{u}}{\| \Vec{u} \|_{2}} - \frac{\Vec{v}}{\| \Vec{v} \|_{2}} \right\|_{2},
     &\cIf \Vec{u}, \Vec{v} \neq \Vec{0}, \\
  0, &\cIf \Vec{u} = \Vec{v} = \Vec{0}, \\
  1, &\cOtherwise,
  \end{cases}
\end{gather*}
for \(  \Vec{u}, \Vec{v} \in \R^{\dX}  \).
The sign function, \(  \Sign. : \R \to \BinarySet  \), follows the convention:
\begin{gather*}
  \Sign( a )
  =
  \begin{cases}
   -1, &\cIf a    < 0, \\
  \+1, &\cIf a \geq 0,
  \end{cases}
\end{gather*}
where \(  a \in \R  \).
This notation extends to vectors as \(  \Sign. : \R^{\dX} \to \BinarySet^{\dX}  \)
by taking the \(  \pm  \)-signs of each entry of a \(  \dX  \)-dimensional vector.
%
\par 
%

\subsection{Hard Thresholding and the BIHT Algorithm}

This work considers two notions of hard thresholding as means to project points into
the subspace of \(  \Ell  \)-sparse vectors, \(  \SparseSubspace{\Ell}{\dX}  \):
\topk[\(  \Ell  \)] and \ksubset[\(  \Ell  \)] \hardthresholding.
These are formalized in the following definitions.
%
\begin{definition}[{\Topk[\(  \Ell  \)] \hardthresholding}]
\label{def:top-k-hard-thresholding}
%
%
The \emph{\topk[\(  \Ell  \)] \hardthresholding operation}, denoted by
\(  \Threshold{\Ell} : \R^{\dX} \to \R^{\dX}  \),
projects a vector \(  \Vec{u} \in \R^{\dX}  \) into \(  \SparseSubspace{\Ell}{\dX}  \)
by retaining only the \(  \Ell  \) largest (in absolute value) entries in \(  \Vec{u}  \)
and setting all other entries to \(  0  \).
Note that ``ties'' can be broken arbitrarily.
More formally, writing
\(  \Set{U}_{\Ell} = \{ \Vec{u'} \in \R^{\dX} : \| \Vec{u'} \|_{0} = \Ell, \ \Vec*{u'}_{j} \in \{ \Vec*{u}_{j}, 0 \} \ \forall j \in [\dX] \}  \),
the \topk[\(  \Ell  \)] \hardthresholding operation maps:
\(
  \Vec{u} \mapsto
  \Threshold{\Ell}( \Vec{u} )
  \in \arg\max_{\Vec{u'} \in \Set{U}_{\Ell}} \| \Vec{u'} \|_{1}
\).
\end{definition}
%
\begin{definition}[{\kSubset[\(  \Ell  \)] \hardthresholding}]
\label{def:k-subset-hard-thresholding}
%
%
The \emph{\ksubset \hardthresholding operation}
associated with a coordinate subset \(  J \subseteq [\dX]  \),
denoted by
\(  \ThresholdSet'{J} : \R^{\dX} \to \R^{\dX}  \),
takes a vector, \(  \Vec{u} \in \R^{\dX}  \), into \(  \SparseSubspace{|J|}{\dX}  \)
by setting all entries in \(  \Vec{u}  \) indexed by \(  [\dX] \setminus J  \) to \(  0  \).
More formally, \(  \ThresholdSet'{J}( \Vec{u} )  \) is the vector whose \(  j\Th  \) entries,
\(  j \in [\dX]  \), are given by
\(  \ThresholdSet*'{J}( \Vec{u} )_{j} = \Vec*{u}_{j} \cdot \I( j \in J )  \).
\end{definition}
%
\par 
%
The measurement matrix is denoted by
\(  \Mat{\AM} \in \R^{\m \times \n}  \),
with the measurements, i.e., its rows, written as
\(  \Vec{\AV}\VL{1}, \dots, \Vec{\AV}\VL{\m} \in \R^{\n}  \), such that
\(  \Mat{\AM} = ( \Vec{\AV}\VL{1} \;\cdots\; \Vec{\AV}\VL{\m} )^{\T}  \).
%
Suppose, 
\(  \Vec{\yV} \in \BinarySet^{\m}  \)
denotes an arbitrary vector that satisfies
\(  \DistH{\Vec{\yV} }{\Sign( \Mat{\AM} \Vec{\Vec{x}} )} \leq \tauX \m  \).
%
This vector, \(  \Vec{\yV}\), can be viewed as introducing adversarial noise into the true responses.
%
At this point we can formally define the normalized BIHT algorithm. It is given as Algorithm~\ref{alg:biht:normalized} below.

\begin{algorithm}
\label{alg:biht:normalized}
\caption{Binary iterative hard thresholding (BIHT) algorithm: Input $\Vec{y},\Mat{\AM}$}
Set \( \Eta = \sqrt{2\pi} \)\;
%
%
%
\(
  \Vec{\hat{x}}^{(0)}
  \sim
  \SparseSphereSubspace{k}{n}
\)\;
\For
{\( \Iter = 1, 2, 3, \dots \)}
{
  \(
    \Vec{\tilde{x}}^{(\Iter)}
    \gets
    \Vec{\hat{x}}^{(\Iter-1)}
    +
    \frac{\eta}{m}
    \Mat{\AM}^{\T}
    \cdot
    \frac{1}{2}
    \left(
      \Vec{y} - \Sign \big( \Mat{\AM} \Vec{\hat{x}}^{(\Iter-1)} \big)
    \right)
  \)\;
  \(
    \Vec{\hat{x}}^{(\Iter)}
    \gets
    \frac
    {\Threshold*{k}(\Vec{\tilde{x}}^{(\Iter)})}
    {\| \Threshold*{k}(\Vec{\tilde{x}}^{(\Iter)}) \|_{2}}
  \)\;
}
\end{algorithm}

%
\subsection{Some Universal Constants}
None of the universal constants appearing in this work are very large.
These constants,
\(  \aC, \bC, \cX{1}, \cX{2}, \cX{3}, \cX{4}, \cX{} > 0  \),
appear throughout the results and analysis in this work.
These universal constants are fixed as follows:
\begin{subequations}
\label{eqn:universal-constants}
\begin{align}
  &
  \aC = \aValue,
  &&
  \bC \lesssim \bCValue,
  \\
  &
  \cX{1} = \cXOneValue \in ( \cXOneValueX, \cXOneValueXX ),
  &&
  \cX{2} = \cXTwoValue \in ( \cXTwoValueX, \cXTwoValueXX ),
  \\
  &
  \cX{3} = \cXThreeValue \in ( \cXThreeValueX, \cXThreeValueXX ),
  &&
  \cX{4} = \cXFourValue \in ( \cXFourValueX, \cXFourValueXX ),
  \\
  &
  \cX{} = 4 \left( \cX{1} + \sqrt{\cX{1}^{2} + \cX{2}} \right)^{2}
        \in ( \cValue, \cValueXX ).
\end{align}
\end{subequations}


\section{Main Result and Technical Overview}        
\label{section:main-results} 

\THEOREM \ref{thm:adversarial}, below,
states the main result of this work, which establishes the convergence of \BIHT
when an arbitrary but bounded number of responses are corrupted.
Note that it is a universal result in the sense that the measurement matrix, \(  \Mat{\AM}  \),
is fixed across the recovery of all \ksparserealunit vectors.

\begin{theorem}
\label{thm:adversarial}
Let
\(  \epsilonX, \errorO, \tauX, \rhoX \in ( 0,1 ]  \),
\(  \rX > 0  \),
\(  \k, \m, \n \in \Z_{+}  \),
where
\begin{gather}
  \label{eqn:thm:adversarial:r}
  \rX \defeq
  \rXValue
  ,\\ \label{eqn:thm:adversarial:epsilon_0}
  \epsilonO \defeq
  \epsilonX  + \rX
,\end{gather}
and where
\begin{gather}
  \label{eqn:thm:adversarial:m_0}
  \m
  \geq
  \frac{4\bC\cX{}}{\epsilonX}
  \Log(
    \binom{\n}{\k}^{2}
    \binom{\n}{2\k}
    \left( \frac{12 \bC}{\epsilonX} \right)^{2\k}
    \left( \frac{3 \aC}{\rhoX} \right)
  )
  =
  \BigO(
    \frac{\k}{\epsilonX} \Log( \frac{\n}{\epsilonX \k} )
    +
    \frac{1}{\epsilonX} \Log( \frac{1}{\rhoX} )
  )
.\end{gather}
%
Fix an \(  \m \times \n  \) measurement matrix,
\(  \Mat{\AM} \in \R^{\m \times \n}  \),
whose rows,
\(  \Vec{\AV}\VL{1}, \dots, \Vec{\AV}\VL{\m} \sim \N( \Vec{0}, \Id{\n} )  \),
are \iid Gaussian random vectors. 
%
Uniformly with probability at least
\(  1 - \rhoX  \),
for all \ksparserealunit vectors,
\(  \Vec{\x} \in \SparseSphereSubspace{\k}{\n}  \),
when given \(  \m  \) noisy responses,
\(  \Vec{\yV}  \in \BinarySet^{\m}  \)
(i.e., with any choice of up to \(  \tauX \m  \) corrupted),
\begin{gather}
\label{eqn:thm:adversarial:3}
  \DistH{\Vec{\yV}}{\Sign( \Mat{\AM} \Vec{\x} )} \leq \tauX \m,
\end{gather}
the sequence of approximations,
\(  \{ \Vec{\xApprox}\IL{\Iter} \in \SparseSphereSubspace{\k}{\n} \}_{\Iter \in \Z_{\geq 0}}  \),
produced by the \NBIHT algorithm converges as
\begin{subequations}
\label{eqn:thm:adversarial:4}
\begin{gather}
\label{eqn:thm:adversarial:4:1}
  \DistS{\Vec{\x}}{\Vec{\xApprox}\IL{\Iter}}
  \leq
  2^{2^{-\Iter}}
  \epsilonO^{1-2^{-\Iter}}
\end{gather}
with an approximation error asymptotically bounded from above by
\begin{gather}
\label{eqn:thm:adversarial:4:2}
  \lim_{\Iter \to \infty} \DistS{\Vec{\x}}{\Vec{\xApprox}\IL{\Iter}}
  \leq
  \epsilonO
.\end{gather}
\end{subequations}
\end{theorem}
\subsection{The Restricted Approximate Invertibility Condition (\RAIC) under Adversarial Noise} 
\label{section:|>prelim|>raic}                                          
As we have discussed in the introduction, the key step of proving our result is to establish a property of Gaussian matrices called restricted approximate invertibility in the presence of adversarial sign-filps. Before we give a technical overview of our proof, here we define the notion of \RAIC and formally present the result regarding Gaussian matrices.

Fixing the measurement matrix,
\(  \Mat{\AM} \in \R^{\m \times \n}  \),
let
\(  \f : \R^{\n} \to \BinarySet^{\m}  \)
denote an arbitrary map that satisfies
\(  \DistH{\f( \Vec{u} )}{\Sign( \Mat{\AM} \Vec{u} )} \leq \tauX \m  \)
for all \(  \Vec{u} \in \R^{\n}  \).
This map, \(  \f  \), can be viewed as introducing one particular adversarial error pattern into the true responses,
\(  \Sign( \Mat{\AM} \Vec{u} ) \mapsto \f( \Vec{u} )  \).
The set of all such functions, \(  \f  \), is denoted by \(  \Set{F}_{\Mat{\AM}} \defeq \{ \f : \R^{\n} \to \BinarySet^{\m} ~:~ \DistH{\f( \Vec{u} )}{\Sign( \Mat{\AM} \Vec{u} )} \leq \tauX \m \ \forall \Vec{u} \in \R^{\n} \}  \).
This is essentially the set of all possible ways to adversarially corrupt the true responses.

Additionally, define the functions
\(  \h, \hf{}{} : \R^{\n} \times \R^{\n} \to \R^{\n}  \),
at arbitrary ordered pairs points,
\(  ( \Vec{u}, \Vec{v} ) \in \R^{\n} \times \R^{\n}  \),
by
\begin{subequations}
\label{eqn:h:def}
\begin{gather}
  \label{eqn:h:def:h_A}
  \h( \Vec{u}, \Vec{v} )
  =
  \frac{\sqrt{2\pi}}{\m}
  \Mat{\AM}^{\T}
  \cdot \frac{1}{2}
  \bigl( \Sign( \Mat{\AM} \Vec{u} ) - \Sign( \Mat{\AM} \Vec{v} ) \bigr)
  =
  \frac{\sqrt{2\pi}}{\m}
  \sum_{\iIx=1}^{\m}
  \Vec{\AV}\VL{\iIx}
  \cdot \frac{1}{2}
  \bigl( \Sign( \langle \Vec{\AV}\VL{\iIx}, \Vec{u} \rangle )
         -
         \Sign( \langle \Vec{\AV}\VL{\iIx}, \Vec{v} \rangle )
  \bigr)
  ,\\ \label{eqn:h:def:h_fA}
  \hf{\Vec{u}}{\Vec{v}}
  =
  \frac{\sqrt{2\pi}}{\m}
  \Mat{\AM}^{\T}
  \cdot \frac{1}{2}
  \bigl( \f( \Vec{u} ) - \Sign( \Mat{\AM} \Vec{v} ) \bigr)
  =
  \frac{\sqrt{2\pi}}{\m}
  \sum_{\iIx=1}^{\m}
  \Vec{\AV}\VL{\iIx}
  \cdot \frac{1}{2}
  \bigl( \f( \Vec{u} )_{i}
         -
         \Sign( \langle \Vec{\AV}\VL{\iIx}, \Vec{v} \rangle )
  \bigr)
,\end{gather}
\end{subequations}
and for \(  \JX \subseteq [\n]  \), let
\(  \h[\JX], \hf[\JX]{}{} : \R^{\n} \times \R^{\n} \to \R^{\n}  \)
denote the functions given at
\(  ( \Vec{u}, \Vec{v} ) \in \R^{\n} \times \R^{\n}  \)
by
\begin{subequations}
\label{eqn:h_J:def}
\begin{gather}
  \label{eqn:h:def:h_AJ}
  \h[\JX]( \Vec{u}, \Vec{v} ) =
  \ThresholdSet'{\Supp( \Vec{u} ) \cup \Supp( \Vec{v} ) \cup \JX}( \h( \Vec{u}, \Vec{v} ) )
  ,\\ \label{eqn:h:def:h_fAJ}
  \hf[\JX]{\Vec{u}}{\Vec{v}} =
  \ThresholdSet'{\Supp( \Vec{u} ) \cup \Supp( \Vec{v} ) \cup \JX}( \hf{\Vec{u}}{\Vec{v}} )
.\end{gather}
\end{subequations}
%
Note that
\begin{gather*}
  \frac{1}{2}
  \bigl( \f( \Vec{u} )_{i}
         -
         \Sign( \langle \Vec{\AV}\VL{\iIx}, \Vec{v} \rangle )
  \bigr)
  =
  -\Sign( \langle \Vec{\AV}\VL{\iIx}, \Vec{v} \rangle )
  \cdot
  \I( \f( \Vec{u} )_{i} \neq \Sign( \langle \Vec{\AV}\VL{\iIx}, \Vec{v} \rangle ) )
.\end{gather*}
and hence, \(  \hf{}{}, \hf[\JX]{}{}  \) are equivalently given by
\begin{subequations}
\label{eqn:notation:h:def-alt}
\begin{gather}
  \label{eqn:notation:hfA:def-alt}
  \hf{\Vec{u}}{\Vec{v}}
  =
  -\frac{\sqrt{2\pi}}{\m}
  \sum_{\iIx=1}^{\m}
  \Vec{\AV}\VL{\iIx}
  \Sign( \langle \Vec{\AV}\VL{\iIx}, \Vec{v} \rangle )
  \cdot \I( \f( \Vec{u} )_{i} \neq \Sign( \langle \Vec{\AV}\VL{\iIx}, \Vec{v} \rangle ) )
  ,\\ \label{eqn:notation:hfAJ:def-alt}
  \hf[\JX]{\Vec{u}}{\Vec{v}}
  =
  -\frac{\sqrt{2\pi}}{\m}
  \sum_{\iIx=1}^{\m}
  \ThresholdSet'{\Supp( \Vec{u} ) \cup \Supp( \Vec{v} ) \cup \JX}( \Vec{\AV}\VL{\iIx} )
  \Sign( \langle \Vec{\AV}\VL{\iIx}, \Vec{v} \rangle )
  \cdot \I( \f( \Vec{u} )_{i} \neq \Sign( \langle \Vec{\AV}\VL{\iIx}, \Vec{v} \rangle ) )
.\end{gather}
\end{subequations}
%
The main technical theorem is stated next.
Its proof is deferred to \SECTION \ref{section:|>pf-main-technical}.

\begin{theorem}[\RAICadv for Gaussian measurements]
\label{thm:raic:adv}
Fix
\(  \rhoX \in (0,1]  \),
\(  \epsilonRAICadv, \tauX \in (0,1]  \),
\(  \k,\m,\n \in \Z_{+}  \),
where 
\begin{gather}
  m
  \geq
  \frac{\bC}{\epsilonRAICadv}
  \Log(
    \binom{\n}{\k}^{2}
    \binom{\n}{2\k}
    \left( \frac{12 \bC}{\epsilonRAICadv} \right)^{2\k}
    \left( \frac{3 \aC}{\rhoX} \right)
  )
  =
  \BigO(
    \frac{\k}{\epsilonRAICadv} \Log( \frac{\n}{\epsilonRAICadv \k} )
    +
    \frac{1}{\epsilonRAICadv} \Log( \frac{1}{\rhoX} )
  )
.\end{gather}
%
Let
\(  \Set{\ZX} = \{ \Vec{\ZX}\VL{1}, \dots, \Vec{\ZX}\VL{\m} \sim \N( \Vec{0}, \Id{\n} ) \}  \)
be a set of \(  \m  \) \iid standard multivariate normal random vectors,
and define the matrix, \(  \Mat{\ZX} \in \R^{\m \times \n}  \), which stacks them up,
\(  \Mat{\ZX} = ( \Vec{\ZX}\VL{1} \;\cdots\; \Vec{\ZX}\VL{\m} )^{\T}  \).
%
%
Then, with probability at least
\(  1 - \rhoX  \),
uniformly for all
\(  \f \in \Set{F}_{\Mat{\AM}}  \),
\(  \Vec{\xV}, \Vec{\yV} \in \SparseSphereSubspace{\k}{\n}  \),
\(  \JX \subseteq [\n]  \), \(  | \JX | \leq \k  \),
\begin{align}
\label{eqn:thm:raic:adv:1}
  \bigl\| ( \Vec{\xV} - \Vec{\yV} ) - \hf[\JX]{\Vec{\xV}}{\Vec{\yV}} \bigr\|_{2}
  &\leq
  \cX{1} \sqrt{\epsilonRAICadv \DistS{\Vec{x}}{\Vec{y}}}
  +
  \cX{2} \epsilonRAICadv
  +
  \DXXuUBconst
.\end{align}
%
\end{theorem}



\subsection{Technical Overview}    
\label{section:|>prelim|>overview} 


The proof of the main theorem, \THEOREM \ref{thm:adversarial},
is broadly divided into three steps, each considered under (bounded) adversarial noise:
\Enum[\label{enum:stochastic}]{\thesection.I} establish a stochastic result for Gaussian measurements,
\Enum[\label{enum:deterministic}]{\thesection.II} establish a deterministic result for the iterative approximation errors of \BIHT with arbitrary measurements, and
\Enum[\label{enum:combine}]{\thesection.III} combine \ref{enum:stochastic} and \ref{enum:deterministic} to characterize the convergence of \BIHT under adversarial noise.
The result obtained in \STEP \ref{enum:stochastic} establishes
the \RAIC for Gaussian measurements under adversarial noise
(\see \THEOREM \ref{thm:raic:adv}) by upper bounding:
\begin{gather}
\label{eqn:enum:stochastic}
  \bigl\| ( \Vec{\xV} - \Vec{\yV} ) - \hf[\JX]{\Vec{\xV}}{\Vec{\yV}} \bigr\|_{2}
  \leq
  \BigO'(
    \sqrt{\epsilonRAICadv \DistS{\Vec{x}}{\Vec{y}}}
    +
    \epsilonRAICadv
    +
    \tauX 
  )
\end{gather}
uniformly with high probability
for all \(  \Vec{\xV}, \Vec{\yV} \in \SparseSphereSubspace{\k}{\n}  \).
The result derived in \STEP \ref{enum:deterministic} upper bounds
the error of the \(  \Iter\Th  \) \BIHT approximations deterministically
by an expression similar to that in the definition of the \RAIC
(\see \LEMMA \ref{lemma:error:deterministic}):
\begin{gather}
\label{eqn:enum:deterministic}
  \DistS{\Vec{\x}}{\Vec{\xApprox}\IL{\Iter}}
  \leq
  \BigO(
    \bigl\| ( \Vec{\x} - \Vec{\xApprox}\IL{\Iter-1} ) - \hf[\JX]{\Vec{\x}}{\Vec{\xApprox}\IL{\Iter-1}} \bigr\|_{2}
  )
.\end{gather}
%
Lastly, for \STEP \ref{enum:combine}, upon the establishment of
\EQUATIONS \eqref{eqn:enum:stochastic} and \eqref{eqn:enum:deterministic},
the two equations taken together will bound the \(  \Iter\Th  \) approximation errors by:
\begin{gather}
\label{eqn:enum:combine}
  \DistS{\Vec{\x}}{\Vec{\xApprox}\IL{\Iter}}
  \leq
  \begin{cases}
  2, &\cIf \Iter = 0, \\
  \BigO*'(
    \sqrt{\epsilonRAICadv \DistS{\Vec{\x}}{\Vec{\xApprox}\IL{\Iter-1}}}
    +
    \epsilonRAICadv
    +
    \tauX 
  ),
  &\cIf \Iter > 0.
  \end{cases}
\end{gather}
%
Note, however, that the above expression states the upper bound on the approximation error
as a recurrence relation, rather than a closed-form result.
Hence, \STEP \ref{enum:combine} will also derive a closed-form expression for
\EQUATION \eqref{eqn:enum:combine} (\see \LEMMA \ref{lemma:error:recurrence}),
where much of the technical work here
has already been accomplished by \cite{matsumoto2022binary}.
%
\par 
%
The majority of the analysis focuses on the stochastic result in \STEP \ref{enum:stochastic},
which is the main technical contribution of this work,
while the analyses for the deterministic bound in \STEP \ref{enum:deterministic}
and the final step, \STEP \ref{enum:combine},
are less involved but allow the \RAIC established in \STEP \ref{enum:stochastic}
to be related to the error of the approximations produced by the \BIHT algorithm
with corrupted responses.
The arguments for \STEP \ref{enum:stochastic} are briefly outlined below.
%
On the other hand, \STEPS \ref{enum:deterministic} and \ref{enum:combine}
are less technically demanding and hence omitted from this overview
(\see \LEMMAS \ref{lemma:error:deterministic} and \ref{lemma:error:recurrence} and their proofs).
%
\paragraph{Overview of the Argument for \STEP \ref{enum:stochastic}.} 
%
The idea behind the approach to \THEOREM \ref{thm:raic:adv} is the following.
There is a (roughly) linear relationship between the expected norm of the sum of up to \(  \tauX \m  \)-many measurements and the expected error from adversarially corrupting up to \(  \tauX \m  \)-many responses. Hence, since the norm of the sum of every choice of up to \(  \tauX \m  \)-many measurements can be uniformly bounded as not ``too large'' with high probability, the error induced by the adversarial noise is similarly upper bounded with high probability.

%
More precisely, the argument for \THEOREM \ref{thm:raic:adv} is broken down into a few steps:
\Enum[\label{enum:step:A:a}]{a}~%
First, applying the triangle inequality, it can be shown that
\begin{gather}
\label{eqn:enum:step:A:a}
  \bigl\| ( \Vec{\xV} - \Vec{\yV} ) - \hf[\JX]{\Vec{\xV}}{\Vec{\yV}} \bigr\|_{2}
  \leq
  \bigl\| ( \Vec{\xV} - \Vec{\yV} ) - \h[\JX]( \Vec{\xV}, \Vec{\yV} ) \bigr\|_{2}
  +
  \bigl\| \hf[\JX]{\Vec{\xV}}{\Vec{\xV}} \bigr\|_{2}
.\end{gather}
%
Then, the focus of the subsequent two steps is upper bounding
the two terms on the \RHS of the above inequality.
Note that \EQUATION \eqref{eqn:enum:step:A:a} gives a roughly linear dependence
of the approximation error on the amount of adversarial noise.
\Enum[\label{enum:step:A:b}]{b}~%
The first term on the \RHS of \EQUATION \eqref{eqn:enum:step:A:a} can be upper bounded
by directly applying \cite[\THEOREM 3.3]{matsumoto2022binary}.
\Enum[\label{enum:step:A:c}]{c}~%
On the other hand, the rightmost term in \EQUATION \eqref{eqn:enum:step:A:a}---%
which (roughly) quantifies the amount of error caused by adversarial noise---%
requires new analysis.
As the first step towards bounding this term, it will be argued that it suffices to
bound each element in the image of \(  \hf[\JX]{}{}  \), where
\(  \hf[\JX]{}{} [ \PREIMG ]  \) has a finite and easily quantifiable size.
Note that this approach will lead to a uniform bound on
the norm of the image of \(  \hf[\JX]{}{}  \) at every real-valued point,
\(  ( \Vec{\xV}, \Vec{\xV} )  \), \(  \Vec{\xV} \in \R^{\n}  \).
\Enum[\label{enum:step:A:d}]{d}~%
Finally, such a uniform bound is obtained by bounding the norm of
the image of \(  \hf[\JX]{}{}  \) at an arbitrary point \(  ( \Vec{\uV}, \Vec{\uV} )  \),
and subsequently union bounding over a specifically constructed set of such points.
This step will orthogonally decompose \(  \hf[\JX]{\Vec{\uV}}{\Vec{\uV}}  \) into two components,
\(  \langle \Vec{\uV}, \hf[\JX]{\Vec{\uV}}{\Vec{\uV}} \rangle \Vec{\uV}  \) and
\(  \hf[\JX]{\Vec{\uV}}{\Vec{\uV}} - \langle \Vec{\uV}, \hf[\JX]{\Vec{\uV}}{\Vec{\uV}} \rangle \Vec{\uV}  \),
such that
\begin{gather*}
  \hf[\JX]{\Vec{\uV}}{\Vec{\uV}}
  =
  \langle \Vec{\uV}, \hf[\JX]{\Vec{\uV}}{\Vec{\uV}} \rangle \Vec{\uV}
  +
  \Bigl( \hf[\JX]{\Vec{\uV}}{\Vec{\uV}}
         - \langle \Vec{\uV}, \hf[\JX]{\Vec{\uV}}{\Vec{\uV}} \rangle \Vec{\uV} \Bigr)
.\end{gather*}
%
The norm of each of the two components will be individually upper bounded
using concentration inequalities for functions of Gaussians,
and subsequently, these bounds will be combined via the triangle inequality,
\begin{gather}
  \left\| \hf[\JX]{\Vec{\uV}}{\Vec{\uV}} \right\|_{2}
  \leq
  \left| \langle \Vec{\uV}, \hf[\JX]{\Vec{\uV}}{\Vec{\uV}} \rangle \right|
  +
  \left\| \hf[\JX]{\Vec{\uV}}{\Vec{\uV}}
          - \langle \Vec{\uV}, \hf[\JX]{\Vec{\uV}}{\Vec{\uV}} \rangle \Vec{\uV} \right\|_{2}
,\end{gather}
and a union bound.
\section{Proof of \THEOREM \ref{thm:adversarial}} 
\label{section:|>pf-main-thm}                     


As discussed in the technical overview (\see \SECTION \ref{section:|>prelim|>overview}),
the proof of the main theorem, \THEOREM \ref{thm:adversarial}, follows largely from
three intermediate results, which are formalized as
\THEOREM \ref{thm:raic:adv} 
---the stochastic result sought in \STEP \ref{enum:stochastic}---and as
\LEMMAS \ref{lemma:error:deterministic} and \ref{lemma:error:recurrence} in
\SECTION \ref{section:|>pf-main-thm|>additional}%
---the deterministic results sought in \STEPS \ref{enum:deterministic} and \ref{enum:combine},
respectively.
Recall that \THEOREM \ref{thm:raic:adv}, the main technical contribution,
establishes that with high probability Gaussian measurements satisfy the \RAICadv,
while \LEMMAS \ref{lemma:error:deterministic} and \ref{lemma:error:recurrence}
provide a means to relate the \RAICadv to a contraction inequality
for the sequence of \BIHT approximation errors, first as a recurrence relation
and subsequently in closed-form.
%
\par 
%

\subsection{Intermediate Results} 
\label{section:|>pf-main-thm|>additional}    

As already discussed, the following lemmas will facilitate the proof of
\THEOREM \ref{thm:adversarial}.
The proof of \LEMMAS \ref{lemma:error:deterministic} and \LEMMA \ref{lemma:error:recurrence},
can be found in \SECTION \ref{section:|>pf-deterministic}.
respectively.

\begin{lemma}
\label{lemma:error:deterministic}
For all
\(  \Vec{\x} \in \SparseSphereSubspace{\k}{\n}  \) and
\(  \Iter \in \Z_{+}  \),
the error of the \(  \Iter\Th  \) BIHT approximation,
\(  \Vec{\xApprox}\IL{\Iter} \in \SparseSphereSubspace{\k}{\n}  \),
is bounded from above by
\begin{align}
  \label{eqn:lemma:error:deterministic:1}
  \DistS{\Vec{\x}}{\Vec{\xApprox}\IL{\Iter}}
  &\leq
  4
  \bigl\|
    ( \Vec{\x} - \Vec{\xApprox}\IL{\Iter-1} )
    -
    \hf[\Supp( \Vec{\xApprox}\IL{\Iter} )]{\Vec{\x}}{\Vec{\xApprox}\IL{\Iter-1}}
  \bigr\|_{2}
.\end{align}
\end{lemma}

\begin{lemma}[\LemmaLabel{\cf{\cite[\LEMMA 4.2]{matsumoto2022binary}}}]
\label{lemma:error:recurrence}
Let
\(  \CC, \cC, \cX{1}, \cX{2}, \cX{3} > 0  \)
be defined as in \EQUATION \eqref{eqn:universal-constants},
and fix
\(  \tauX \in [0,1]  \).
%
Let
\(  \gammaX \in (0,1]  \),
and define the function
\({  \Err{} : \Z_{\geq 0} \to \R  }\)
by the recurrence relation
\begin{subequations}
\label{eqn:lemma:error:recurrence:1}
\begin{gather}
  \label{eqn:lemma:error:recurrence:1:t=0}
  \Err{0} = 2
  ,\\ \label{eqn:lemma:error:recurrence:1:t>0}
  \Err{\Iter}
  = 4\cX{1} \sqrt{\frac{\gammaX}{\cC} \Err{\Iter-1}}
  + \frac{4 \cX{2} \gammaX}{\cC}
  ,\qquad
  \Iter \in \Z_{+}
.\end{gather}
\end{subequations}
%
Then,
\begin{align}
  \label{eqn:lemma:error:recurrence:2:1}
  \lim_{\Iter \to \infty} \Err{\Iter}
  \leq \gammaX
.\end{align}
%
Moreover, the sequence
\(  \{ \Err{\Iter} \}_{\Iter \in \Z_{\geq 0}}  \)
is \pointwise upper bounded by the sequence
\begin{gather}
  \left\{
    2^{2^{-\Iter}}
    \gammaX^{1-2^{-\Iter}}
  \right\}_{\Iter \in \Z_{\geq 0}}
.\end{gather}
\end{lemma}
\subsection{Proof of \THEOREM \ref{thm:adversarial}} 
\label{section:|>pf-main|>pf-main-thm}               

\begin{proof}
{\THEOREM \ref{thm:adversarial}}
%
%
The theorem will follow from an argument analogous to that which appeared in
\cite[proof of \THEOREM 3.1 and \COROLLARY 3.2]{matsumoto2022binary}.
By \LEMMA \ref{lemma:error:deterministic}, followed by \THEOREM \ref{thm:raic:adv},
if \(  \m \geq 4\cX{} \mO  \), then with probability at least \(  1 - \rhoX  \),
for each
\(  \Vec{\x} \in \SparseSphereSubspace{\k}{\n}  \) and
\(  \Iter \in \Z_{\geq 0}  \),
the error of the \(  \Iter\Th  \) BIHT approximation of \(  \Vec{\x}  \) is bounded from above by:
\begin{align*}
  &
  \DistS{\Vec{\x}}{\Vec{\xApprox}\IL{\Iter}}
  \\
  &\leq
  4
  \bigl\|
    ( \Vec{\x} - \Vec{\xApprox}\IL{\Iter-1} )
    -
    \hf[\Supp( \Vec{\xApprox}\IL{\Iter} )]{\Vec{\x}}{\Vec{\xApprox}\IL{\Iter-1}}
  \bigr\|_{2}
  \\
  &\dCmt{by \LEMMA \ref{lemma:error:deterministic}}
  \\
  &\leq
  4 \left(
    \cX{1} \sqrt{\frac{\epsilonX}{\cX{}} \DistS{\Vec{x}}{\Vec{\xApprox}\IL{\Iter-1}}}
    +
    \frac{\cX{2} \epsilonX}{\cX{}}
    +
    \DXXuUBconstX
  \right)
  \\
  &=
  4 \left(
    \cX{1} \sqrt{\frac{\epsilonX}{\cX{}} \DistS{\Vec{x}}{\Vec{\xApprox}\IL{\Iter-1}}}
    +
    \frac{\cX{2} \epsilonX}{\cX{}}
    +
    \frac{\cX{2} \rX}{\cX{}}
  \right)
  \\
  &\dCmt{by the choice of \(  \rX = \frac{\cX{}}{\cX{2}} \left( \DXXuUBconstX \right)  \)}
  \\
  &=
  4 \left(
    \cX{1} \sqrt{\frac{\epsilonX}{\cX{}} \DistS{\Vec{x}}{\Vec{\xApprox}\IL{\Iter-1}}}
    +
    \frac{\cX{2} ( \epsilonX+\rX-\rX )}{\cX{}}
    +
    \frac{\cX{2} \rX}{\cX{}}
  \right)
  \\
  &=
  4 \left(
    \cX{1} \sqrt{\frac{\epsilonX}{\cX{}} \DistS{\Vec{x}}{\Vec{\xApprox}\IL{\Iter-1}}}
    +
    \frac{\cX{2} ( \epsilonO-\rX )}{\cX{}}
    +
    \frac{\cX{2} \rX}{\cX{}}
  \right)
  \\
  &=
  4 \left(
    \cX{1} \sqrt{\frac{\epsilonX}{\cX{}} \DistS{\Vec{x}}{\Vec{\xApprox}\IL{\Iter-1}}}
    +
    \frac{\cX{2} \epsilonO}{\cX{}}
    -
    \frac{\cX{2} \rX}{\cX{}}
    +
    \frac{\cX{2} \rX}{\cX{}}
  \right)
  \\
  &=
  4 \left(
    \cX{1} \sqrt{\frac{\epsilonX}{\cX{}} \DistS{\Vec{x}}{\Vec{\xApprox}\IL{\Iter-1}}}
    +
    \frac{\cX{2} \epsilonO}{\cX{}}
  \right)
  \\
  &\leq
  4 \left(
    \cX{1} \sqrt{\frac{\epsilonO}{\cX{}} \DistS{\Vec{x}}{\Vec{\xApprox}\IL{\Iter-1}}}
    +
    \frac{\cX{2} \epsilonO}{\cX{}}
  \right)
  \\
  &\leq
  4 \cX{1} \sqrt{\frac{\epsilonO}{\cX{}} \DistS{\Vec{x}}{\Vec{\xApprox}\IL{\Iter-1}}}
  +
  \frac{4 \cX{2} \epsilonO}{\cX{}}
.\end{align*}
%
In summary, with probability at least \(  1 - \rhoX  \),
uniformly for all \ksparserealunit vectors,
\(  \Vec{\x} \in \SparseSphereSubspace{\k}{\\n}  \),
the following holds for all \(  \Iter \in \Z_{+}  \):
\begin{align}
\label{eqn:pf:thm:adversarial:1a}
  \DistS{\Vec{\x}}{\Vec{\xApprox}\IL{\Iter}}
  \leq
  4 \cX{1} \sqrt{\frac{\epsilonO}{\cX{}} \DistS{\Vec{x}}{\Vec{\xApprox}\IL{\Iter-1}}}
  +
  \frac{4 \cX{2} \epsilonO}{\cX{}}
\end{align}
%
Additionally, trivially, \(  \DistS{\Vec{\x}}{\Vec{\xApprox}\IL{0}} \leq 2  \) since
\begin{gather}
\label{eqn:pf:thm:adversarial:1b}
  \DistS{\Vec{\x}}{\Vec{\xApprox}\IL{0}} \leq \DistS{\Vec{\x}}{-\Vec{\x}} = 2
.\end{gather}
%
\par 
%
Next, arbitrarily fixing \(  \Vec{\x} \in \SparseSphereSubspace{\k}{\n}  \),
it will be shown by induction that whenever
\EQUATIONS \eqref{eqn:pf:thm:adversarial:1a} and \eqref{eqn:pf:thm:adversarial:1b} hold,
the sequence of values
\(  \{ \Err{\Iter} \}_{\Iter \in \Z_{\geq 0}}  \)
\pointwise upper bounds the sequence
\(  \{ \DistS{\Vec{\x}}{\Vec{\xApprox}\IL{\Iter}} \}_{\Iter \in \Z_{\geq 0}}  \).\,
where \(  \Err{} : \Z_{\geq 0} \to \R  \) is defined as in \LEMMA \ref{lemma:error:recurrence}.
%
\par 
%
The base case, when \(  \Iter = 0  \), is trivial since
\(  \sup_{\Vec{u},\Vec{v} \in \R^{\n}} \DistS{\Vec{u}}{\Vec{v}} = 2 = \Err{0}  \).
%
Now, arbitrarily fixing \(  \Iter \in \Z_{+}  \),
suppose each \(  \IterX\Th  \) \BIHT approximation, \(  \IterX < \Iter  \), satisfies
\(  \DistS{\Vec{\x}}{\Vec{\xApprox}\IL{\IterX}} \leq \Err{\IterX}  \).
%
Then, the aim is to show that
\(  \DistS{\Vec{\x}}{\Vec{\xApprox}\IL{\Iter}} \leq \Err{\Iter}  \),
where
\(  \Err{\Iter}
     = 4 \cX{1} \sqrt{\frac{\gammaX}{\cX{}} \Err{\Iter-1}}
     + \frac{4 \cX{2} \gammaX}{\cX{}}  \)
with the fixing of \(  \gammaX = \epsilonO  \).
Observe:
\begin{align*}
  \DistS{\Vec{\x}}{\Vec{\xApprox}\IL{\Iter}}
  &\leq
  4 \cX{1} \sqrt{\frac{\epsilonO}{\cX{}} \DistS{\Vec{x}}{\Vec{\xApprox}\IL{\Iter-1}}}
  +
  \frac{4 \cX{2} \epsilonO}{\cX{}}
  \\
  &\dCmt{by \EQUATION \eqref{eqn:pf:thm:adversarial:1a}}
  \\
  &\leq
  4 \cX{1} \sqrt{\frac{\epsilonO}{\cX{}} \Err{\Iter-1}}
  +
  \frac{4 \cX{2} \epsilonO}{\cX{}}
  \\
  &\dCmt{by the inductive assumption}
  \\
  &=
  \Err{\Iter}
\end{align*}
%
Said briefly,
\(  \DistS{\Vec{\x}}{\Vec{\xApprox}\IL{\Iter}} \leq \Err{\Iter}  \),
as claimed.
By induction, it follows that for all \(  \Iter \in \Z_{\geq 0}  \),
the error of the \(  \Iter\Th  \) \BIHT approximation is bounded from above by
\(  \DistS{\Vec{\x}}{\Vec{\xApprox}\IL{\Iter}} \leq \Err{\Iter}  \).
%
Extending this to all other \(  \Vec{\x} \in \SparseSphereSubspace{\k}{\n}  \)
via the earlier discussion, every such sequence of \BIHT approximations is
\pointwise upper bounded by \(  \Err{}  \) with high probability.
%
\par 
%
Having verified the above, the theorems follow immediately from
\LEMMA \ref{lemma:error:recurrence}:
\begin{gather*}
  \DistS{\Vec{\x}}{\Vec{\xApprox}\IL{\Iter}}
  \leq
  \Err{\Iter}
  \leq
  2^{2^{-\Iter}} \epsilonO^{1-2^{-\Iter}}
  ,\\
  \lim_{\Iter \to \infty} \DistS{\Vec{\x}}{\Vec{\xApprox}\IL{\Iter}}
  \leq
  \epsilonO
.\end{gather*}
%
This completes the proof of the main theorem,
\THEOREM \ref{thm:adversarial}.
\end{proof}
\let\oldXu\Xu
\let\oldYu\Yu
\let\oldyu\yu
\renewcommand{\Xu}{\RV{X}_{\Vec{\uV}}}
\renewcommand{\Yu}{\Vec{Y}_{\Vec{\uV}}}
\renewcommand{\yu}{\Vec{y}}

\section{Proof of the Main Technical Theorem (\THEOREM \ref{thm:raic:adv})}
\label{section:|>pf-main-technical}


\subsection{Discussion and Preliminaries}   
\label{section:|>pf-main-technical|>prelim} 

We begin by introducing and verifying some results that will set us up for
proving the main technical theorem, \THEOREM \ref{thm:raic:adv}.
Throughout \SECTIONS \ref{section:|>pf-main-technical}-\ref{section:|>D2},
the function
\(  \f : \R^{\n} \to \BinarySet^{\m}  \)
is taken to be any function which upholds:
\(  \DistH{\f( \Vec{w} )}{\Sign( \Mat{\ZX} \Vec{w} )} \leq \tauX \m  \)
at every point, \(  \Vec{w} \in \R^{\n}  \).
Specification of this condition will be henceforth omitted to avoid redundancy.
%
\par 
First off, the \LHS of \EQUATION \eqref{eqn:thm:raic:adv:1} in \THEOREM \ref{thm:raic:adv},
\(  \left\| ( \Vec{\xV} - \Vec{\yV} ) - \hZf[\JX]{\Vec{\xV}}{\Vec{\yV}} \right\|_{2}  \),
is split into two components (with bounding).

\begin{claim}
\label{claim:raic:adv:1}
For all
\(  \Vec{\xV}, \Vec{\yV} \in \R^{\n}  \) and
\(  \JX \subseteq [\n]  \),
the following inequality (deterministically) holds:
\begin{align}
  \left\| ( \Vec{\xV} - \Vec{\yV} ) - \hZf[\JX]{\Vec{\xV}}{\Vec{\yV}} \right\|_{2}
  &\leq
  \left\| ( \Vec{\xV} - \Vec{\yV} ) - \hZ[\JX]( \Vec{\xV}, \Vec{\yV} )  \right\|_{2}
  +
  \left\| \hZf[\JX]{\Vec{\xV}}{\Vec{\xV}} \right\|_{2}
\label{eqn:claim:raic:adv:1:1}
.\end{align}
\end{claim}

\begin{proof}
{\CLAIM \ref{claim:raic:adv:1}}
Fix
\(  \Vec{\xV}, \Vec{\yV} \in \R^{\n}  \) and
\(  \JX \subseteq [\n]  \),
arbitrarily.
The (random) vector
\(
  {( \Vec{\xV} - \Vec{\yV} ) - \hZf[\JX]{\Vec{\xV}}{\Vec{\yV}}}
\)
can be rewritten as follows:
\begin{align*}
  ( \Vec{\xV} - \Vec{\yV} ) - \hZf[\JX]{\Vec{\xV}}{\Vec{\yV}}
  =
  \bigl( ( \Vec{\xV} - \Vec{\yV} ) - \hZ[\JX]( \Vec{\xV}, \Vec{\yV} ) \bigr)
  -
  \bigl( \hZf[\JX]{\Vec{\xV}}{\Vec{\yV}} - \hZ[\JX]( \Vec{\xV}, \Vec{\yV} ) \bigr)
,\end{align*}
where the second and fourth terms on the \RHS cancel.
Additionally, observe:
\begin{align*}
  &
  \hZf[\JX]{\Vec{\xV}}{\Vec{\yV}} - \hZ[\JX]( \Vec{\xV}, \Vec{\yV} )
  \\
  &=
  \frac{\sqrt{2\pi}}{\m}
  \sum_{\iIx=1}^{\m}
  \ThresholdSet[{\Supp( \Vec{\xV} ) \cup \Supp( \Vec{\yV} ) \cup}]{\JX}( \Vec{\ZX}\VL{\iIx} )
  \cdot
  \frac{1}{2}
  \bigl(
    \fZix
    -
    \Sign.( \langle \Vec{\yV}, \Vec{\ZX}\VL{\iIx} \rangle )
  \bigr)
  \\ &\AlignSp
  -
  \frac{\sqrt{2\pi}}{\m}
  \sum_{\iIx=1}^{\m}
  \ThresholdSet[{\Supp( \Vec{\xV} ) \cup \Supp( \Vec{\yV} ) \cup}]{\JX}( \Vec{\ZX}\VL{\iIx} )
  \cdot
  \frac{1}{2}
  \bigl(
    \Sign.( \langle \Vec{\xV}, \Vec{\ZX}\VL{\iIx} \rangle )
    -
    \Sign.( \langle \Vec{\yV}, \Vec{\ZX}\VL{\iIx} \rangle )
  \bigr)
  \\
  &=
  \frac{\sqrt{2\pi}}{\m}
  \sum_{\iIx=1}^{\m}
  \ThresholdSet[{\Supp( \Vec{\xV} ) \cup \Supp( \Vec{\yV} ) \cup}]{\JX}( \Vec{\ZX}\VL{\iIx} )
  \cdot
  \frac{1}{2}
  \Bigl(
  \bigl(
    \fZix
    -
    \Sign.( \langle \Vec{\yV}, \Vec{\ZX}\VL{\iIx} \rangle )
  \bigr)
  -
  \bigl(
    \Sign.( \langle \Vec{\xV}, \Vec{\ZX}\VL{\iIx} \rangle )
    -
    \Sign.( \langle \Vec{\yV}, \Vec{\ZX}\VL{\iIx} \rangle )
  \bigr)
  \Bigr)
  \\
  &=
  \frac{\sqrt{2\pi}}{\m}
  \sum_{\iIx=1}^{\m}
  \ThresholdSet[{\Supp( \Vec{\xV} ) \cup \Supp( \Vec{\yV} ) \cup}]{\JX}( \Vec{\ZX}\VL{\iIx} )
  \cdot
  \frac{1}{2}
  \bigl(
    \fZix
    -
    \Sign.( \langle \Vec{\xV}, \Vec{\ZX}\VL{\iIx} \rangle )
  \bigr)
  \\
  &=
  \hZf[\JX]{\Vec{\xV}}{\Vec{\xV}}
\end{align*}
%
Thus, combining above work:
\begin{align*}
  ( \Vec{\xV} - \Vec{\yV} ) - \hZf[\JX]{\Vec{\xV}}{\Vec{\yV}}
  &=
  \bigl( ( \Vec{\xV} - \Vec{\yV} ) - \hZ[\JX]( \Vec{\xV}, \Vec{\yV} ) \bigr)
  -
  \bigl( \hZf[\JX]{\Vec{\xV}}{\Vec{\yV}} - \hZ[\JX]( \Vec{\xV}, \Vec{\yV} ) \bigr)
  \\
  &=
  \bigl( ( \Vec{\xV} - \Vec{\yV} ) - \hZ[\JX]( \Vec{\xV}, \Vec{\yV} ) \bigr)
  - \hZf[\JX]{\Vec{\xV}}{\Vec{\xV}}
.\end{align*}
%
Then, the norm is upper bounded as follows:
\begin{align*}
  \left\| ( \Vec{\xV} - \Vec{\yV} ) - \hZf[\JX]{\Vec{\xV}}{\Vec{\yV}} \right\|_{2}
  &=
  \left\|
    \bigl( ( \Vec{\xV} - \Vec{\yV} ) - \hZ[\JX]( \Vec{\xV}, \Vec{\yV} ) \bigr)
    - \hZf[\JX]{\Vec{\xV}}{\Vec{\xV}}
  \right\|_{2}
  \\
  &\leq
  \left\| ( \Vec{\xV} - \Vec{\yV} ) - \hZ[\JX]( \Vec{\xV}, \Vec{\yV} ) \right\|_{2}
  +
  \left\| \hZf[\JX]{\Vec{\xV}}{\Vec{\xV}} \right\|_{2}
\end{align*}
where the bottom line applies the triangle inequality.
\end{proof}

\EQUATION* \eqref{eqn:claim:raic:adv:1:1} of \CLAIM \ref{claim:raic:adv:1}
decomposes (with bounding) the random variable of interest,
\(  \left\| ( \Vec{\xV} - \Vec{\yV} ) - \hZf[\JX]{\Vec{\xV}}{\Vec{\yV}} \right\|_{2}  \),
into two terms which are individually easier to control.
The majority of the argument for \THEOREM \ref{thm:raic:adv} is towards
a uniform upper bound on the latter term of this decomposition,
\(  \DXX{\Vec{\xV}}{\Vec{\xV}} \defeq
    \left\| \hZf[\JX]{\Vec{\xV}}{\Vec{\xV}} \right\|_{2}  \).
%
This second term, \(  \DXX{\Vec{\xV}}{\Vec{\xV}}  \), requires new analysis,
which takes up \SECTIONS \ref{section:|>pf-main-technical|>prelim|>D2} and \ref{section:|>D2}.
On the other hand, the first term,
\(  \DX{\Vec{\xV}}{\Vec{\xV}} \defeq
    \left\| ( \Vec{\xV} - \Vec{\yV} ) - \hZ[\JX]( \Vec{\xV}, \Vec{\yV} ) \right\|_{2}  \),
is immediately upper bounded with high probability for all
\(  \Vec{\xV}, \Vec{\yV} \in \SparseSphereSubspace{\k}{\n}  \)
via \cite[\THEOREM 3.3]{matsumoto2022binary}, stated below.

\begin{lemma}[{\cite[\THEOREM 3.3]{matsumoto2022binary}}]
\label{lemma:raic-for-gaussians}
Fix
\(  \epsilonRAIC, \rhoRAIC \in (0,1)  \),
\(  \k, \m, \n \in \Z_{+}  \),
such that
\begin{gather}
\label{eqn:lemma:raic-for-gaussians:1}
  \m \geq
  \frac{\bC}{\epsilonRAIC}
  \Log(
    \binom{\n}{\k}^{2}
    \binom{\n}{2\k}
    \left( \frac{12 \bC}{\epsilonRAIC} \right)^{2\k}
    \left( \frac{\aC}{\rhoRAIC} \right)
  )
.\end{gather}
%
Then, uniformly with probability at least
\(  1 - \rhoRAIC  \),
the Gaussian measurement matrix \(  \Mat{\AM} \in \R^{\m \times \n}  \) satisfies the
\(  ( \k, \n, \epsilonRAIC, \cX{1}, \cX{2} )  \)-RAIC:
\begin{gather}
\label{eqn:lemma:raic-for-gaussians:2}
  \bigl\| ( \Vec{x} - \Vec{y} ) - \h[\JX]( \Vec{x}, \Vec{y} ) \bigr\|_{2}
  \leq
  \cX{1} \sqrt{\epsilonRAIC \DistS{\Vec{x}}{\Vec{y}}} + \cX{2} \epsilonRAIC
\end{gather}
for all
\(  \Vec{x}, \Vec{y} \in \SparseSphereSubspace{\k}{\n}  \) and all
\(  \JX \subseteq [\n]  \), \(  | \JX | \leq \k  \).
\end{lemma}


\subsubsection{Discussion Regarding \texorpdfstring{\( \DXX{\Vec{\xV}}{\Vec{\xV}} \)}{D2;J(x,x)}} 
\label{section:|>pf-main-technical|>prelim|>D2}                                                   

The derivation of an upper bound on the random variable,
\(  \DXX{\Vec{\xV}}{\Vec{\xV}} = \| \hZf[\JX]{\Vec{\xV}}{\Vec{\xV}} \|_{2}  \),
will entirely ignore the specification of the vector \(  \Vec{\xV}  \).
Rather, recall the definition of \(  \hZf[\JX]{}{}  \):
\begin{align*}
  \hZf[\JX]{\Vec{w}}{\Vec{w'}}
  &=
  \frac{\sqrt{2\pi}}{\m}
  \sum_{\iIx=1}^{\m}
  \TJZi
  \cdot \frac{1}{2}
  \left( \Big.
    \f( \Vec{w} )_{\iIx}
    -
    \Sign( \langle \Vec{w'}, \Vec{\ZX}\VL{\iIx} \rangle )
  \right)
  \\
  &=
  \frac{\sqrt{2\pi}}{\m}
  \sum_{\iIx=1}^{\m}
  \TJZi
  \cdot \frac{1}{2}
  \left( \Big.
    \f( \Vec{w} )
    -
    \Sign( \Mat{\ZX} \Vec{w'} )
  \right)_{\iIx}
\end{align*}
where
\(  \Mat{\ZX} \in \R^{\m \times \n}  \),
\(  \Mat{\ZX} = ( \Vec{\ZX}\VL{1} \;\cdots\; \Vec{\ZX}\VL{\m} )^{\T}  \), and where
\(  \Vec{w}, \Vec{w'} \in \R^{\n}  \).
%
The only dependence of \(  \hZf[\JX]{\Vec{w}}{\Vec{w'}}  \) on the preimage,
\(  ( \Vec{w}, \Vec{w'} )  \), is captured in the expression:
\(  \f( \Vec{w} ) - \Sign( \Mat{\ZX} \Vec{w'} )  \),
where the images of \(  \f  \) and \(  \Sign  \) over \(  \R^{\n}  \) and \(  \R^{\m}  \),
respectively, are subsets of the set
\(  \{ \Vec{z} \in \{ -1,0,1 \}^{\m} : \| \Vec{z} \|_{0} \leq \tauX \m \}  \).
%
All else is fixed across all \(  ( \Vec{w}, \Vec{w'} ) \in \R^{\m} \times \R^{\m}  \).
Thus, upon fixing the set of Gaussian vectors, \(  \Mat{\AV}\VL{1}, \dots, \Mat{\AV}\VL{\m}  \),
for each
\(  \JX \subseteq [\n]  \), \(  | \JX | \leq \kX  \),
the image of the function \(  \hZf[\JX]{}{}  \) has finite cardinality no more than:
\begin{gather*}
  \left|
    \hZf[\JX]{}{} \left[ \PREIMG \right]
  \right|
  \leq
  \sum_{\Ell=1}^{\tauX \m} \binom{\m}{\Ell} 2^{\Ell}
.\end{gather*}
%
It then suffices to enumerate each of the up to
\(  \sum_{\Ell=1}^{\tauX \m} \binom{\m}{\Ell} 2^{\Ell}  \)-many
vectors comprising
\(  \hZf[\JX]{}{} \left[ \PREIMG \right]  \)
and bound their norms for each choice of \(  \JX \subseteq [\n]  \), \(  | \JX | \leq \kX  \).
With this motivation, construct a collection of sets,
\(  \US[\JX] \subseteq \SparseSphereSubspace{\k}{\n}  \), \(  \JX \subseteq [\n]  \),
by inserting precisely one vector,
\(  \Vec{\uV} \in \SparseSphereSubspace{\k}{\n}  \),
into \(  \US[\JX]  \) for each vector \(  \Vec{z} \in \hZf[\JX]{}{} \left[ \PREIMG \right]  \)
such that
\(  \hZf[\JX]{\Vec{\uV}}{\Vec{\uV}} = \Vec{z}  \).
%
%
The above discussion is formalized and verified in the following claim and its proof.

\begin{claim}
\label{claim:raic:adv:2}
Fix
\(  \gamma \geq 0  \).
%
Suppose for all
\(  \JX \subseteq [\n]  \), \(  | \JX | \leq \kX  \), and
\(  \Vec{\uV} \in \US[\JX]  \), 
the norm of \(  \hZf[\JX]{\Vec{\uV}}{\Vec{\uV}}  \) is bounded from above by:
\(  \| \hZf[\JX]{\Vec{\uV}}{\Vec{\uV}} \|_{2} \leq \gamma  \).
%
Then, uniformly for all \(  \Vec{\xV} \in \SparseSphereSubspace{\k}{\n}  \) and
for all \(  \JX \subseteq [\n]  \), \(  | \JX | \leq \kX  \),
the same bound holds at \(  ( \Vec{\xV}, \Vec{\xV} )  \):
\(  \| \hZf[\JX]{\Vec{\xV}}{\Vec{\xV}} \|_{2} \leq \gamma  \).
\end{claim}

\begin{proof}
{\CLAIM \ref{claim:raic:adv:2}}
Suppose for the sake of contradiction that
\(  \| \hZf[\JX]{\Vec{\uV}}{\Vec{\uV}} \|_{2} \leq \gamma  \)
for each
\(  \JX \subseteq [\n]  \), \(  | \JX | \leq \kX  \), and
\(  \Vec{\uV} \in \US[\JX]  \), 
but there exists
\(  \Jx \subseteq [\n]  \), \(  | \Jx | \leq \kX  \), and
\(  \Vec{\xV} \in \SparseSphereSubspace{\k}{\n}  \), 
for which
\(  \| \hZf[\Jx]{\Vec{\xV}}{\Vec{\xV}} \|_{2} > \gamma  \).
%
Denote the image of \(  ( \Vec{\xV}, \Vec{\xV} )  \) under \(  \hZf[\Jx]{}{}  \) by
\(  \Vec{z} = \hZf[\Jx]{\Vec{\xV}}{\Vec{\xV}}  \),
where by assumption,
\(  \| \Vec{z} \|_{2} > \gamma  \),
and let
\(
  \Set{\WX} =
  \{ \Vec{\wV} \in \SparseSphereSubspace{\k}{\n} : \hZf[\Jx]{\Vec{\wV}}{\Vec{\wV}} = \Vec{\zV} \}
\).
%
Then, by the construction of the set \(  \US[\Jx]  \), it must be that
\(  | \US[\Jx] \cap \Set{\WX} | = 1  \),
which implies that there exists
\(  \Vec{\uV} \in \US[\Jx]  \) for which
\(  \| \hZf[\JX]{\Vec{\uV}}{\Vec{\uV}} \|_{2} = \| \Vec{z} \|_{2} > \gamma  \)%
---a contradiction.
By this contradiction, the claim holds.
\end{proof}

Due \CLAIM \ref{claim:raic:adv:2},
the proof of \THEOREM \ref{thm:raic:adv} will seek to bound
\(  \| \hZf[\JX]{\Vec{\uV}}{\Vec{\uV}} \|_{2}  \) from above for all
\(  \JX \subseteq [\n]  \), \(  | \JX | \leq \kX  \), and \(  \Vec{\uV} \in \US[\JX]  \).
%
Specifically, \LEMMA \ref{lemma:D2} controls this random variable,
\(  \DXX{\Vec{\uV}}{\Vec{\uV}} = \| \hZf[\JX]{\Vec{\uV}}{\Vec{\uV}} \|_{2}  \),
uniformly for every
\(  \JX \subseteq [\n]  \), \(  | \JX | \leq \kX  \), and \(  \Vec{\uV} \in \US[\JX]  \).

\begin{lemma}
\label{lemma:D2}
Let \(  \m \in \Z_{+}  \) satisfy
\begin{gather*}
  m
  \geq
  \frac{\bC}{\epsilonRAICadv}
  \Log(
    \binom{\n}{\k}^{2}
    \binom{\n}{2\k}
    \left( \frac{12 \bC}{\epsilonRAICadv} \right)^{2\k}
    \left( \frac{3 \aC}{\rhoX} \right)
  )
  =
  \BigO(
    \frac{\k}{\epsilonRAICadv} \Log( \frac{\n}{\epsilonRAICadv \k} )
    +
    \frac{1}{\epsilonRAICadv} \Log( \frac{1}{\rhoX} )
  )
.\end{gather*}
%
With probability at least
\(  1 - \frac{2\rhoX}{3}  \),
uniformly for all
\(  \JX \subseteq [\n]  \), \(  | \JX | \leq \k  \), and \(  \Vec{\uV} \in \US[\JX]  \),
\begin{gather}
  \DXXu \leq \DXXuUB = \DXXuUBconst
.\end{gather}
\end{lemma}

The proof of the lemma is deferred to \APPENDIX \ref{section:|>appendix|>D2:ub}.
Next, \THEOREM \ref{thm:raic:adv} is proved,
contingent on the proof of \LEMMA \ref{lemma:D2}.

\begin{proof}
{\THEOREM \ref{thm:raic:adv}}
Fix
\(
  m
  =
  \frac{\bC}{\epsilonRAICadv}
  \Log*(
    \binom{\n}{\k}^{2}
    \binom{\n}{2\k}
    ( \frac{12 \bC}{\epsilonRAICadv} )^{2\k}
    ( \frac{3 \aC}{\rhoX} )
  )
\).
%
Due to \CLAIM \ref{claim:raic:adv:1}, for every
\(  \Vec{\xV}, \Vec{\yV} \in \R^{\n}  \)
and every
\(  \JX \subseteq [\n]  \), 
\begin{align*}
  \left\| ( \Vec{\xV} - \Vec{\yV} ) - \hZf[\JX]{\Vec{\xV}}{\Vec{\yV}} \right\|_{2}
  &\leq
  \left\| ( \Vec{\xV} - \Vec{\yV} ) - \hZ[\JX]( \Vec{\xV}, \Vec{\yV} )  \right\|_{2}
  +
  \left\| \hZf[\JX]{\Vec{\xV}}{\Vec{\xV}} \right\|_{2}
\end{align*}
%
By \LEMMA \ref{lemma:raic-for-gaussians} (\cite[\THEOREM 3.3]{matsumoto2022binary}),
with probability at least
\(  1 - \frac{\rhoX}{3}  \),
uniformly for all
\(  \Vec{\xV}, \Vec{\yV} \in \SparseSphereSubspace{\k}{\n}  \) and
\(  \JX \subseteq [\n]  \), \(  | \JX | \leq \k  \),
\begin{align*}
  \left\| ( \Vec{\xV} - \Vec{\yV} ) - \hZ[\JX]( \Vec{\xV}, \Vec{\yV} )  \right\|_{2}
  \leq
  \cX{1} \sqrt{\epsilonRAICadv \DistS{\Vec{x}}{\Vec{y}}}
  +
  \cX{2} \epsilonRAICadv
,\end{align*}
where \(  \cX{1}, \cX{2} > 0  \) are universal constants as defined in
\EQUATION \eqref{eqn:universal-constants}.
Additionally, by \LEMMA \ref{lemma:D2},
with probability at least
\(  1 - \frac{2\rhoX}{3}  \),
uniformly for all
\(  \JX \subseteq [\n]  \), \(  | \JX | \leq \k  \), and \(  \Vec{\uV} \in \US[\JX]  \),
\begin{align*}
  \left\| \hZf[\JX]{\Vec{\uV}}{\Vec{\uV}} \right\|_{2} \leq \DXXuUB = \DXXuUBconst
.\end{align*}
%
Recalling \CLAIM \ref{claim:raic:adv:2}, it follows that the same bound on
\(  \left\| \hZf[\JX]{\Vec{\xV}}{\Vec{\xV}} \right\|_{2}  \)
holds uniformly over all
\(  \Vec{\xV} \in \SparseSphereSubspace{\k}{\n}  \),
\(  \JX \subseteq [\n]  \), \(  | \JX | \leq \kX  \),
with the same probability.
Combining the above bounds on
\(  \left\| ( \Vec{\xV} - \Vec{\yV} ) - \hZ[\JX]( \Vec{\xV}, \Vec{\yV} )  \right\|_{2}  \) and
\(  \left\| \hZf[\JX]{\Vec{\uV}}{\Vec{\uV}} \right\|_{2}  \)
via a union bound, and applying \CLAIM \ref{claim:raic:adv:1},
the desired upper bound follows:
with probability at least
\(  1 - \frac{\rhoX}{3} - \frac{2\rhoX}{3} = 1 - \rhoX  \),
uniformly for all
\(  \Vec{\xV}, \Vec{\yV} \in \SparseSphereSubspace{\k}{\n}  \),
\(  \JX \subseteq [\n]  \), \(  | \JX | \leq \kX  \):
\begin{align*}
  \left\| ( \Vec{\xV} - \Vec{\yV} ) - \hZf[\JX]{\Vec{\xV}}{\Vec{\yV}} \right\|_{2}
  &\leq
  \left\| ( \Vec{\xV} - \Vec{\yV} ) - \hZ[\JX]( \Vec{\xV}, \Vec{\yV} )  \right\|_{2}
  +
  \left\| \hZf[\JX]{\Vec{\xV}}{\Vec{\xV}} \right\|_{2}
  \\
  &\leq
  \cX{1} \sqrt{\epsilonRAICadv \DistS{\Vec{x}}{\Vec{y}}} + \cX{2} \epsilonRAICadv
  +
  \DXXuUB
  \\
  &=
  \cX{1} \sqrt{\epsilonRAICadv \DistS{\Vec{x}}{\Vec{y}}} + \cX{2} \epsilonRAICadv
  +
  \DXXuUBconst
\end{align*}
where \(  \cX{1}, \cX{2}, \cX{3}, \cX{4} > 0  \) are universal constants specified in
\EQUATION \eqref{eqn:universal-constants}.
%
\end{proof}


\bibliography{refs}

\let\SECTION\APPENDIX
\let\SECTIONS\APPENDICES

\begin{appendix}

\section{Analysis for \texorpdfstring{\(  \DXX{\Vec{\uV}}{\Vec{\uV}}  \)}{D2;J(u,u)}} 
\label{section:|>D2}                                                                  

\renewcommand{\Dx}[3][\JX]{\DRV'_{1;#1}( #2,#3 )}
\renewcommand{\Dxx}[3][\JX]{\DRV'_{2;#1}( #2,#3 )}

\subsection{An Orthogonal Decomposition} 
\label{section:|>D2|>ortho-decomp}       

To control the random variable,
\(  \DXX{\Vec{\uV}}{\Vec{\uV}} \defeq \| \hZf[\JX]{\Vec{\uV}}{\Vec{\uV}} \|_{2}  \),
the random vector, \(  \hZf[\JX]{\Vec{\uV}}{\Vec{\uV}}  \), is orthogonally decomposed into
two components:
\(  \langle \Vec{\uV}, \hZf[\JX]{\Vec{\uV}}{\Vec{\uV}} \rangle \Vec{\uV}  \) and
\(  \hZf[\JX]{\Vec{\uV}}{\Vec{\uV}}
    - \langle \Vec{\uV}, \hZf[\JX]{\Vec{\uV}}{\Vec{\uV}} \rangle \Vec{\uV}  \),
where
\begin{gather}
  \hZf[\JX]{\Vec{\uV}}{\Vec{\uV}}
  =
  \langle \Vec{\uV}, \hZf[\JX]{\Vec{\uV}}{\Vec{\uV}} \rangle \Vec{\uV}
  +
  \bigl( \hZf[\JX]{\Vec{\uV}}{\Vec{\uV}}
         - \langle \Vec{\uV}, \hZf[\JX]{\Vec{\uV}}{\Vec{\uV}} \rangle \Vec{\uV} \bigr)
\label{eqn:raic:adv:1}
.\end{gather}
%
As such, define the random variables, \(  \Dxu  \) and \(  \Dxxu  \), by
\begin{gather}
  \Dxu = \left| \langle \Vec{\uV}, \hZf[\JX]{\Vec{\uV}}{\Vec{\uV}} \rangle \right|
  ,\\
  \Dxxu
  =
  \left\|
    \hZf[\JX]{\Vec{\uV}}{\Vec{\uV}}
    - \langle \Vec{\uV}, \hZf[\JX]{\Vec{\uV}}{\Vec{\uV}} \rangle \Vec{\uV}
  \right\|_{2}
\label{eqn:raic:adv:2}
.\end{gather}

We make use of this decomposition and these random variables in the following claim.

\begin{claim}
\label{claim:raic:adv:3}
For any \(  \JX \subseteq [\n]  \) and \(  \Vec{\uV} \in \R^{\n}  \),
\begin{gather}
  \DXXu \leq \Dxu + \Dxxu
\label{eqn:raic:adv:3}
.\end{gather}
\end{claim}

\noindent%
Thus, \(  \DXX{\Vec{\uV}}{\Vec{\uV}}  \) can be upper bounded by bounding
\(  \Dxu  \) and \(  \Dxxu  \), the latter two of which are simpler to handle
than directly characterizing \(  \DXX{\Vec{\uV}}{\Vec{\uV}}  \).
Such bounds are obtained in \SECTIONS \ref{section:|>D'1} and \ref{section:|>D'2},
respectively.
The proof of \LEMMA \ref{lemma:D2}, which upper bounds \(  \DXX{\Vec{\uV}}{\Vec{\uV}}  \),
is deferred to \APPENDIX \ref{section:|>appendix|>D2:ub},
while \CLAIM \ref{claim:raic:adv:3} is proved next.


\begin{proof}
{\CLAIM \ref{claim:raic:adv:3}}
The claim directly follows from the orthogonal decomposition discussed above
and the triangle inequality:
\begin{align*}
  \DXXu
  &=
  \left\| \hZf[\JX]{\Vec{\uV}}{\Vec{\uV}} \right\|_{2}
  \\
  &\dCmt{by the definition of the random variable \(  \DXXu  \)}
  \\
  &=
  \left\|
    \langle \Vec{\uV}, \hZf[\JX]{\Vec{\uV}}{\Vec{\uV}} \rangle \Vec{\uV}
    +
    \bigl( \hZf[\JX]{\Vec{\uV}}{\Vec{\uV}}
           - \langle \Vec{\uV}, \hZf[\JX]{\Vec{\uV}}{\Vec{\uV}} \rangle \Vec{\uV} \bigr)
  \right\|_{2}
  \\
  &\dCmt{by \EQUATION \eqref{eqn:raic:adv:1}}
  \\
  &\leq
  \left\|
    \langle \Vec{\uV}, \hZf[\JX]{\Vec{\uV}}{\Vec{\uV}} \rangle \Vec{\uV}
  \right\|_{2}
  +
  \left\|
    \hZf[\JX]{\Vec{\uV}}{\Vec{\uV}}
    - \langle \Vec{\uV}, \hZf[\JX]{\Vec{\uV}}{\Vec{\uV}} \rangle \Vec{\uV}
  \right\|_{2}
  \\
  &\dCmt{by the triangle inequality}
  \\
  &\leq
  \left|
    \langle \Vec{\uV}, \hZf[\JX]{\Vec{\uV}}{\Vec{\uV}} \rangle
  \right|
  \| \Vec{\uV} \|_{2}
  +
  \left\|
    \hZf[\JX]{\Vec{\uV}}{\Vec{\uV}}
    - \langle \Vec{\uV}, \hZf[\JX]{\Vec{\uV}}{\Vec{\uV}} \rangle \Vec{\uV}
  \right\|_{2}
  \\
  &\dCmt{due to the homogeneity of norms}
  \\
  &\leq
  \left|
    \langle \Vec{\uV}, \hZf[\JX]{\Vec{\uV}}{\Vec{\uV}} \rangle
  \right|
  +
  \left\|
    \hZf[\JX]{\Vec{\uV}}{\Vec{\uV}}
    - \langle \Vec{\uV}, \hZf[\JX]{\Vec{\uV}}{\Vec{\uV}} \rangle \Vec{\uV}
  \right\|_{2}
  \\
  &\dCmt{\(  \because \| \Vec{\uV} \|_{2} = 1  \)}
  \\
  &=
  \Dxu + \Dxxu
  \\
  &\dCmt{by the definitions of the random variables \(  \Dxu, \Dxxu  \)}
\end{align*}
\end{proof}

\subsection{Concentration Inequalities for the Orthogonal Decomposition} 
\label{section:|>D2|>concentration-ineq}                                 

\renewcommand{\ZX}{Z}

\let\newXu\Xu
\let\Xu\oldXu
\let\newYu\Yu
\let\Yu\oldYu


Before the random variables, \(  \Dxu  \) and \(  \Dxxu  \), are bounded,
two concentration inequalities are stated below as
\LEMMAS \ref{lemma:concentration-ineq:1} and \ref{lemma:concentration-ineq:3}
to facilitate the analysis.
The proofs of these lemmas are deferred to \SECTION \ref{section:|>ortho-decomp}.

\begin{lemma}
\label{lemma:concentration-ineq:1}
Fix
\(  \lu, \tu > 0  \).
%
Let
\(  \Set{\ZX} = \{ \Vec{\ZX}\VL{1}, \dots, \Vec{\ZX}\VL{\lu} \sim \N( \Vec{0}, \Id{\n} ) \}  \)
be a collection of \(  \lu  \) \iid Gaussian vectors, and fix a \ksparserealunit vector,
\(  \Vec{\uX} \in \SparseSphereSubspace{\k}{\n}  \),
and a coordinate subset, \(  \JX \subseteq [\n]  \), \(  | \JX | \leq \kX  \).
%
Define the random variables
\begin{gather*}
  \Xu[\iIx] \defeq
  \left\langle \Vec{\uX}, \TJZi \right\rangle \Sign( \langle \Vec{\uV}, \TJZi \rangle )
,\end{gather*}
for \(  \iIx \in [\lu]  \), and
\begin{gather*}
  \barXu
  \defeq
  \left\langle
    \Vec{\uX},
    \sum_{\iIx=1}^{\lu} \TJZi \Sign( \langle \Vec{\uV}, \TJZi \rangle )
  \right\rangle
  =
  \sum_{\iIx=1}^{\lu} \Xu[\iIx]
.\end{gather*}
%
The concentration of the random variable \(  \barXu  \) is such that
\begin{gather}
\label{eqn:lemma:concentration-ineq:1:1}
  \Pr \left(
    \barXu \geq \left( \sqrt{\frac{2}{\pi}} + \tu \right) \lu
  \right)
  \leq
  e^{-\frac{1}{2} \lu \tu^{2}}
.\end{gather}
%
Additionally, for each \(  \iIx \in [\Ell]  \), the is an equivalence:
\(  | \Xu[\iIx] | = \Xu[\iIx]  \).
\end{lemma}

\begin{lemma}
\label{lemma:concentration-ineq:3}
Fix
\(  \lu, \tu > 0  \).
%
Let
\(  \Set{\ZX} = \{ \Vec{\ZX}\VL{1}, \dots, \Vec{\ZX}\VL{\lu} \sim \N( \Vec{0}, \Id{\n} ) \}  \)
be a collection of \(  \lu  \) \iid Gaussian vectors, and fix a \ksparserealunit vector,
\(  \Vec{\uX} \in \SparseSphereSubspace{\k}{\n}  \),
and a coordinate subset, \(  \JX \subseteq [\n]  \), \(  | \JX | \leq \kX  \).
%
Define the random vector
\begin{gather*}
  \barYu =
  \sum_{\iIx=1}^{\lu}
  \Big(
    \TJZi \Sign( \langle \Vec{\uV}, \TJZi \rangle )
    -
    \langle \Vec{\uV}, \TJZi \rangle \Sign( \langle \Vec{\uV}, \TJZi \rangle ) \Vec{\uV}
  \Big)
.\end{gather*}
%
The concentration of the random variable representing its norm, \(  \| \barYu \|_{2}  \),
is such that
\begin{gather}
\label{eqn:lemma:concentration-ineq:3:1}
  \Pr \left(
    \left\| \barYu \right\|_{2} > \sqrt{\frac{( \kX-1 ) \lu}{2}} + \lu \tu
  \right)
  \leq
  e^{-\frac{1}{2} \lu \tu^{2}}
.\end{gather}
\end{lemma}


\renewcommand{\ZX}{A}

\let\oldXu\Xu
\let\Xu\newXu
\let\oldYu\Yu
\let\Yu\newYu


\let\Xu\oldXu
\let\Yu\oldYu
\let\yu\oldyu

\renewcommand{\Dx}{\DRV\RVL{1}}
\renewcommand{\Dxx}{\DRV\RVL{2}}
\subsection{Bounding \texorpdfstring{\(  \Dxu \defeq | \langle \Vec{\uV}, \hZf[\JX]{\Vec{\uV}}{\Vec{\uV}} \rangle |  \)}{D'1;J(u,u) = |<u,hf;A;J(u,u)>|}} %
\label{section:|>D'1} 

Having introduced the concentration inequalities in \SECTION \ref{section:|>D2|>concentration-ineq},
we are ready to bound the random variables \(  \Dxu  \) and \(  \Dxxu  \).
To start off, the random variable,
\(  \Dxu \defeq | \langle \Vec{\uV}, \hZf[\JX]{\Vec{\uV}}{\Vec{\uV}} \rangle |  \),
is bounded from above per the following lemma.

\begin{lemma}
\label{lemma:D'1}
Fix
\(  \rhoX \in (0,1]  \).
%
Suppose
\begin{gather*}
  m
  \geq
  \frac{\bC}{\epsilonRAICadv}
  \Log(
    \binom{\n}{\k}^{2}
    \binom{\n}{2\k}
    \left( \frac{12 \bC}{\epsilonRAICadv} \right)^{2\k}
    \left( \frac{3 \aC}{\rhoX} \right)
  )
  =
  \BigO(
    \frac{\k}{\epsilonRAICadv} \Log( \frac{\n}{\epsilonRAICadv \k} )
    +
    \frac{1}{\epsilonRAICadv} \Log( \frac{1}{\rhoX} )
  )
.\end{gather*}
%
Then, with probability at least
\(  1 - \frac{\rhoX}{3}  \),
uniformly for all
\(  \JX \subseteq [\n]  \), \(  | \JX | \leq \kX  \), and \(  \Vec{\uV} \in \US[\JX]  \),
\begin{gather}
  \Dxu
  \leq
  \frac{2 \lu}{\m} + \frac{\sqrt{2\pi} \lu \tu}{\m}
%
\end{gather}
where
\begin{gather}
  \lu \leq \tauX \m
  ,\\
  \tu =
  \sqrt{\frac{2}{\lu} \Log( 2 \cdot 2^{\lu} \binom{\m}{\lu} \binom{\n}{\kX} \frac{3 \tauX \m}{\rhoX} )}
.\end{gather}
\end{lemma}

\begin{proof}
{\LEMMA \ref{lemma:D'1}}
First, expanding out and rewriting the expression for \(  \hZf[\JX]{}{}  \) yields:
\begin{align*}
  \hZf[\JX]{\Vec{\uV}}{\Vec{\uV}}
  &=
  \frac{\sqrt{2\pi}}{\m}
  \sum_{\iIx=1}^{\m}
  \TJZi \cdot \frac{1}{2}
  \Bigl( \fZiu - \Sign.( \langle \Vec{\uV}, \TJZi \rangle ) \Bigr)
  \\
  &\dCmt{by \EQUATION \eqref{eqn:h:def:h_fA} in \SECTION \ref{eqn:h:def}}
  \\
  &=
  -
  \frac{\sqrt{2\pi}}{\m}
  \sum_{\iIx=1}^{\m}
  \TJZi \Sign( \langle \Vec{\uV}, \TJZi \rangle )
  \cdot
  \I\Bigl( \f( \Vec{\uV} )_{\iIx} \neq \Sign.(\langle \Vec{\uV}, \TJZi \rangle) \Bigr)
  \\
  &\dCmt{by \EQUATION \eqref{eqn:notation:hfA:def-alt} in \SECTION \ref{section:|>prelim}}
  \\
  &=
  -
  \frac{\sqrt{2\pi}}{\m}
  \sum_{\iIx \in \IX}
  \TJZi \Sign( \langle \Vec{\uV}, \TJZi \rangle )
  \\
  &\dCmt{per the remark below}
\end{align*}
where
\(  \IX \subseteq [\m]  \) indexes the sign-mismatches:
\begin{align*}
  \IX
  &\defeq
  \{ \iIx \in [\m] :
     \fZiu \neq \Sign.( \langle \Vec{\uV}, \TJZi \rangle ) \}
  \equiv
  \left\{ \iIx \in [\m] : \I( \f( \Vec{\uV} )_{\iIx} \neq \Sign.(\langle \Vec{\uV}, \TJZi \rangle) ) \neq 0 \right\}
.\end{align*}
%
%
Note that the assumption on \(  \f  \) stated at the beginning of
\SECTION \ref{section:|>pf-main-technical}%
---that the number of corrupted responses is bounded---%
ensures that
\(  | \IX | \leq \tauX \m  \).
%
Fix
\(  \lu = | \IX | \leq \tauX \m  \),
and without loss of generality, assume
\(  \IX = [\lu]  \).
%
Under this assumption, the above derivation implies:
\begin{align*}
  \hZf[\JX]{\Vec{\uV}}{\Vec{\uV}}
  &=
  -
  \frac{\sqrt{2\pi}}{\m}
  \sum_{\iIx=1}^{\lu}
  \TJZi \Sign( \langle \Vec{\uV}, \TJZi \rangle )
,\end{align*}
or equivalently,
\begin{align*}
  -\hZf[\JX]{\Vec{\uV}}{\Vec{\uV}}
  &=
  \frac{\sqrt{2\pi}}{\m}
  \sum_{\iIx=1}^{\lu}
  \TJZi \Sign( \langle \Vec{\uV}, \TJZi \rangle )
.\end{align*}
%
Now, define the random variables,
\begin{gather*}
  \Xu[\iIx] \defeq
  \left\langle \Vec{\uV}, \TJZi \right\rangle \Sign( \left\langle \Vec{\uV}, \TJZi \right\rangle )
\end{gather*}
for \(  \iIx \in [\lu]  \), and let
\(
  \barXu \defeq \sum_{\iIx=1}^{\lu} \Xu[\iIx]
\).
%
Note that by \LEMMA \ref{lemma:concentration-ineq:1},
\(  | \Xu[\iIx] | = \Xu[\iIx]  \) for each \(  \iIx \in [\lu]  \),
and thus,
\(  | \barXu | = \barXu  \),
as shown in the following derivation:
\begin{align*}
  | \barXu |
  &=
  \left| \sum_{\iIx=1}^{\lu} \Xu[\iIx] \right|
  =
  \left| \sum_{\iIx=1}^{\lu} | \Xu[\iIx] | \right|
  =
  \sum_{\iIx=1}^{\lu} | \Xu[\iIx] |
  =
  \sum_{\iIx=1}^{\lu} \Xu[\iIx]
  =
  \barXu
.\end{align*}
%
Since \(  | \barXu | = \barXu  \), any bound that holds for \(  \barXu  \)
must also hold for \(  | \barXu |  \).
Hence, this proof will focus on upper bounding the value taken by \(  \barXu  \),
rather that directly characterizing \(  | \barXu |  \).
Using this notation,
\begin{align*}
  \left\langle \Vec{\uV}, -\hZf[\JX]{\Vec{\uV}}{\Vec{\uV}} \right\rangle
  &=
  \left\langle
    \Vec{\uV},
    \frac{\sqrt{2\pi}}{\m}
    \sum_{\iIx=1}^{\lu}
    \TJZi \Sign( \langle \Vec{\uV}, \TJZi \rangle )
  \right\rangle
  \\
  &=
  \frac{\sqrt{2\pi}}{\m}
  \sum_{\iIx=1}^{\lu}
  \left\langle \Vec{\uV}, \TJZi \right\rangle \Sign( \left\langle \Vec{\uV}, \TJZi \right\rangle )
  \\
  &=
  \frac{\sqrt{2\pi}}{\m}
  \sum_{\iIx=1}^{\lu}
  \Xu[\iIx]
  \\
  &=
  \frac{\sqrt{2\pi}}{\m}
  \barXu
.\end{align*}
%
Note that by the above observations,
\begin{align*}
  \left\langle \Vec{\uV}, -\hZf[\JX]{\Vec{\uV}}{\Vec{\uV}} \right\rangle
  =
  \frac{\sqrt{2\pi}}{\m} \barXu
  =
  \frac{\sqrt{2\pi}}{\m} \left| \barXu \right|
  =
  \left| \frac{\sqrt{2\pi}}{\m} \barXu \right|
  \geq 0
,\end{align*}
and therefore,
\(
  \left| \left\langle \Vec{\uV}, -\hZf[\JX]{\Vec{\uV}}{\Vec{\uV}} \right\rangle \right|
  =
  \left\langle \Vec{\uV}, -\hZf[\JX]{\Vec{\uV}}{\Vec{\uV}} \right\rangle
\).
%
\def\ECONST{\sqrt{\frac{2}{\pi}}}%
Due to \LEMMA \ref{lemma:concentration-ineq:1},
the random variable \(  \barXu  \) is bounded from above by
\begin{align*}
  \barXu
  &\leq
  \E[\Xu] + \lu \tu
  \leq
  \ECONST \lu + \lu \tu
  =
  \left( \ECONST + \tu \right) \lu
\end{align*}
with probability at least
\(  1 - e^{-\frac{1}{2} \lu \tu^{2}}  \).
%
Take
\begin{align*}
  \tu
  &=
  \sqrt{\frac{2}{\lu} \Log( 2 \cdot 2^{\lu} \binom{\m}{\lu} \binom{\n}{\kX} \frac{3 \tauX \m}{\rhoX} )}
.\end{align*}
%
Then, the desired bound follows:
\begin{align*}
  \left| \left\langle \Vec{\uV}, \hZf[\JX]{\Vec{\uV}}{\Vec{\uV}} \right\rangle \right|
  &=
  \left| -\left\langle \Vec{\uV}, \hZf[\JX]{\Vec{\uV}}{\Vec{\uV}} \right\rangle \right|
  \\
  &=
  \left| \left\langle \Vec{\uV}, -\hZf[\JX]{\Vec{\uV}}{\Vec{\uV}} \right\rangle \right|
  \\
  &=
  \left\langle \Vec{\uV}, -\hZf[\JX]{\Vec{\uV}}{\Vec{\uV}} \right\rangle
  \\
  &=
  \frac{\sqrt{2\pi}}{\m}
  \barXu
  \\
  &\leq
  \frac{\sqrt{2\pi}}{\m}
  \left( \ECONST + \tu \right) \lu
  \\
  &=
  \frac{2 \lu}{\m}
  +
  \frac{\sqrt{2\pi} \lu \tu}{\m}
.\end{align*}
%
This inequality holds for any \emph{single} choice of
\(  \JX \subseteq [\n]  \), \(  | \JX | \leq \kX  \), and \(  \Vec{\uV} \in \US[\JX]  \)
with probability at least
\begin{align*}
  1 - e^{-\frac{1}{2} \lu \tu^{2}}
  &=
  1 - \frac{\frac{\rhoX}{3 \tauX \m}}{2 \cdot 2^{\lu} \binom{\m}{\lu} \binom{\n}{\kX}}
,\end{align*}
and by a union bound, the above inequality holds uniformly for every
\(  \JX \subseteq [\n]  \), \(  | \JX | \leq \kX  \),
\(  \Vec{\uV} \in \US[\JX]  \), and
\(  \lu \in [\tauX \m]  \)
with probability at least
\begin{align*}
  1 -
  \sum_{\lu=1}^{\tauX \m}
  2 \cdot 2^{\lu} \binom{\m}{\lu} \binom{\n}{\kX} e^{-\frac{1}{2} \lu \tu^{2}}
  &=
  1 -
  \sum_{\lu=1}^{\tauX \m}
  2 \cdot 2^{\lu} \binom{\m}{\lu} \binom{\n}{\kX}
  \frac{\frac{\rhoX}{3 \tauX \m}}{2 \cdot 2^{\lu} \binom{\m}{\lu} \binom{\n}{\kX}}
  \\
  &=
  1 - \sum_{\lu=1}^{\tauX \m} \frac{\rhoX}{3 \tauX \m}
  \\
  &=
  1 - \tauX \m \cdot \frac{\rhoX}{3 \tauX \m}
  \\
  &=
  1 - \frac{\rhoX}{3}
.\end{align*}
\end{proof}

\subsection{Bounding \texorpdfstring{\(  \Dxxu \defeq \| \hZf[\JX]{\Vec{\uV}}{\Vec{\uV}} - \langle \Vec{\uV}, \hZf[\JX]{\Vec{\uV}}{\Vec{\uV}} \rangle \Vec{\uV} \|_{2}  \)}{D'2;J(u,u) = ||hf;A;J(u,u) - <u,hf;A;J(u,u)>u||2}
} %
\label{section:|>D'2} 

Next, the second random variable in the orthogonal decomposition,
\(  \Dxxu \defeq \| \hZf[\JX]{\Vec{\uV}}{\Vec{\uV}} - \langle \Vec{\uV}, \hZf[\JX]{\Vec{\uV}}{\Vec{\uV}} \rangle \Vec{\uV} \|_{2}  \),
is upper bounded in \LEMMA \ref{lemma:D'2}, laid out below.

\begin{lemma}
\label{lemma:D'2}
Fix
\(  \rhoX \in (0,1]  \).
%
Suppose 
\begin{gather*}
  m
  \geq
  \frac{\bC}{\epsilonRAICadv}
  \Log(
    \binom{\n}{\k}^{2}
    \binom{\n}{2\k}
    \left( \frac{12 \bC}{\epsilonRAICadv} \right)^{2\k}
    \left( \frac{3 \aC}{\rhoX} \right)
  )
  =
  \BigO(
    \frac{\k}{\epsilonRAICadv} \Log( \frac{\n}{\epsilonRAICadv \k} )
    +
    \frac{1}{\epsilonRAICadv} \Log( \frac{1}{\rhoX} )
  )
.\end{gather*}
%
Then, with probability at least
\(  1 - \frac{\rhoX}{3}  \),
uniformly for all
\(  \JX \subseteq [\n]  \), \(  | \JX | \leq \kX  \), and \(  \Vec{\uV} \in \US[\JX]  \),
\begin{gather}
  \Dxxu
  \leq
  \frac{\sqrt{2\pi ( \kX-1 ) \lu}}{\m}
  +
  \frac{\sqrt{2\pi} \lu \tu}{\m}
\end{gather}
where
\begin{gather}
  \lu \leq \tauX \m
  ,\\
  \tu =
  \sqrt{\frac{2}{\lu} \Log( 2 \cdot 2^{\lu} \binom{\m}{\lu} \binom{\n}{\kX} \frac{3 \tauX \m}{\rhoX} )}
.\end{gather}
\end{lemma}

\begin{proof}
{\LEMMA \ref{lemma:D'2}}
Define the random variables
\(  \Xu[\iIx]  \), \(  \iIx \in [\lu]  \), and
\(  \barXu = \sum_{\iIx=1}^{\lu} \Xu[\iIx]  \)
as in the proof of \LEMMA \ref{lemma:D'1}.
As before, write
\(  \IX \subseteq [\m]  \),
\(  \IX \defeq
    \{ \iIx \in [\m] :
       \fZiu
       \neq \Sign.( \langle \Vec{\uV}, \TJZi \rangle ) \}  \),
and let
\(  \lu \defeq | \IX |  \), where
\(  \lu = | \IX | \leq \tauX \m  \).
%
Without loss of generality, take
\(  \IX = [\lu]  \).
%
From the proof of \LEMMA \ref{lemma:D'1}, we have:
\begin{align}
  \nonumber
  \hZf[\JX]{\Vec{\uV}}{\Vec{\uV}}
  &=
  -
  \frac{\sqrt{2\pi}}{\m}
  \sum_{\iIx \in \IX}
  \TJZi \Sign( \langle \Vec{\uV}, \TJZi \rangle )
  \\ \label{eqn:pf:lemma:D'2:1}
  &=
  -
  \frac{\sqrt{2\pi}}{\m}
  \sum_{\iIx=1}^{\lu}
  \TJZi \Sign( \langle \Vec{\uV}, \TJZi \rangle )
\end{align}
where the second equality uses the assumption that \(  \IX = [\lu]  \).
Recall that
\begin{align}
  \nonumber
  \left\langle \Vec{\uV}, \hZf[\JX]{\Vec{\uV}}{\Vec{\uV}} \right\rangle
  &=
  \left\langle
    \Vec{\uV},
    -
    \frac{\sqrt{2\pi}}{\m}
    \sum_{\iIx=1}^{\lu}
    \TJZi \Sign( \langle \Vec{\uV}, \TJZi \rangle )
  \right\rangle
  \\ \label{eqn:pf:lemma:D'2:2}
  &=
  -
  \frac{\sqrt{2\pi}}{\m}
  \sum_{\iIx=1}^{\lu}
  \left\langle \Big.
    \Vec{\uV},
    \TJZi \Sign( \langle \Vec{\uV}, \TJZi \rangle )
  \right\rangle
.\end{align}
%
Thus,
\begin{align}
  \nonumber
  &
  \hZf[\JX]{\Vec{\uV}}{\Vec{\uV}}
  -
  \langle \Vec{\uV}, \hZf[\JX]{\Vec{\uV}}{\Vec{\uV}} \rangle \Vec{\uV}
  \\ \nonumber
  &=
  -
  \frac{\sqrt{2\pi}}{\m}
  \sum_{\iIx=1}^{\lu}
  \TJZi \Sign( \langle \Vec{\uV}, \TJZi \rangle )
  +
  \frac{\sqrt{2\pi}}{\m}
  \sum_{\iIx=1}^{\lu}
  \left\langle \Big.
    \Vec{\uV},
    \TJZi \Sign( \langle \Vec{\uV}, \TJZi \rangle )
  \right\rangle
  \Vec{\uV}
  \\ \nonumber
  &\dCmt{by expanding the terms via \EQUATIONS \eqref{eqn:pf:lemma:D'2:1} and \eqref{eqn:pf:lemma:D'2:2}}
  \\ \label{eqn:pf:lemma:D'2:3}
  &=
  -
  \frac{\sqrt{2\pi}}{\m}
  \sum_{\iIx=1}^{\lu}
  \left( \Big.
    \TJZi \Sign( \langle \Vec{\uV}, \TJZi \rangle )
    -
    \left\langle \big.
      \Vec{\uV},
      \TJZi \Sign( \langle \Vec{\uV}, \TJZi \rangle )
    \right\rangle
    \Vec{\uV}
  \right)
  \\ \nonumber
  &\dCmt{by combining the summations and factoring out the \(  \tfrac{\sqrt{2\pi}}{\m}  \) term.}
\end{align}
%
Taking the norm of the above expression yields the following:
\begin{align}
  \nonumber
  &
  \left\|
    \hZf[\JX]{\Vec{\uV}}{\Vec{\uV}}
    -
    \langle \Vec{\uV}, \hZf[\JX]{\Vec{\uV}}{\Vec{\uV}} \rangle \Vec{\uV}
  \right\|_{2}
  \\ \nonumber
  &=
  \left\|
  -
  \frac{\sqrt{2\pi}}{\m}
  \sum_{\iIx=1}^{\lu}
  \left( \Big.
    \TJZi \Sign( \langle \Vec{\uV}, \TJZi \rangle )
    -
    \left\langle \big.
      \Vec{\uV},
      \TJZi \Sign( \langle \Vec{\uV}, \TJZi \rangle )
    \right\rangle
    \Vec{\uV}
  \right)
  \right\|_{2}
  \\ \nonumber
  &\dCmt{by \EQUATION \eqref{eqn:pf:lemma:D'2:3}}
  \\ \nonumber
  &=
  \frac{\sqrt{2\pi}}{\m}
  \left\|
  \sum_{\iIx=1}^{\lu}
  \left( \Big.
    \TJZi \Sign( \langle \Vec{\uV}, \TJZi \rangle )
    -
    \left\langle \big.
      \Vec{\uV},
      \TJZi \Sign( \langle \Vec{\uV}, \TJZi \rangle )
    \right\rangle
    \Vec{\uV}
  \right)
  \right\|_{2}
  \\ \nonumber
  &\dCmt{due to the homogeneity of norms}
  \\ \nonumber
  &\leq
  \frac{\sqrt{2\pi}}{\m} \left( \big. \sqrt{\frac{( \kX-1 ) \lu}{2}} + \lu \tu \right)
  \\ \nonumber
  &\dCmt{by \LEMMA \ref{lemma:concentration-ineq:3}}
  \\ \nonumber
  &=
  \frac{\sqrt{\pi ( \kX-1 ) \lu}}{\m}
  +
  \frac{\sqrt{2\pi} \lu \tu}{\m}
  \\ \nonumber
  &\dCmt{by distributing the \(  \tfrac{\sqrt{2\pi}}{\m}  \) term}
\end{align}
where the second to last expression (the inequality) 
holds with probability at least
\(  1 - e^{-\frac{1}{2} \lu \tu^{2}} \geq 1 - \frac{\rhoX}{3}  \).
\end{proof}

\let\Xu\oldXu
\let\Yu\oldYu
\let\yu\oldyu
\section{Controlling \texorpdfstring{\(  \DXX{\Vec{\uV}}{\Vec{\uV}} \defeq \| \hZf[\JX]{\Vec{\uV}}{\Vec{\uV}} \|_{2}  \)}{D2;J(u,u) = ||hf;A;J(u,u)||2}} %
\label{section:|>appendix|>D2:ub} 

The analysis in \SECTIONS \ref{section:|>D'1} and \ref{section:|>D'2}
makes possible a straightforward derivation of the upper bound on
\(  \DXX{\Vec{\uV}}{\Vec{\uV}} \defeq \| \hZf[\JX]{\Vec{\uV}}{\Vec{\uV}} \|_{2}  \)
as stated in \LEMMA \ref{lemma:D2}, which is proved next.

\begin{proof}
{\LEMMA \ref{lemma:D2}}
Due to \LEMMAS \ref{lemma:D'1} and \ref{lemma:D'2} and by a union bound over their results,
the following inequalities hold simultaneously for all
\(  \JX \subseteq [\n]  \), \(  | \JX | \leq \kX  \), 
and \(  \Vec{\uV} \in \US[\JX]  \)
with probability at least
\(  1 - \frac{2 \rhoX}{3}  \)
when
\(
  m
  \geq
  \frac{\bC}{\epsilonRAICadv}
  \Log*(
    \binom{\n}{\k}^{2}
    \binom{\n}{2\k}
    \left( \frac{12 \bC}{\epsilonRAICadv} \right)^{2\k}
    \left( \frac{3 \aC}{\rhoX} \right)
  )
\):
\begin{gather*}
  \Dxu \leq \frac{2 \lu}{\m} + \frac{\sqrt{2\pi} \lu \tu}{\m}
  ,\\
  \Dxxu \leq \frac{\sqrt{2\pi} \lu \tu}{\m} + \frac{\sqrt{\pi ( \kX-1 ) \lu}}{\m}
,\end{gather*}
where
\begin{gather*}
  \lu \leq \tauX \m
  ,\\
  \tu =
  \sqrt{\frac{2}{\lu} \Log( 2 \cdot 2^{\lu} \binom{\m}{\lu} \binom{\n}{\kX} \frac{3 \tauX \m}{\rhoX} )}
.\end{gather*}
%
Observe:
\def\CONST{2}%
\begin{align*}
  \lu \tu
  &=
  \lu
  \sqrt
  {\frac{2}{\lu} \Log( \CONST \cdot 2^{\lu} \binom{\m}{\lu} \binom{\n}{\kX} \frac{3 \tauX \m}{\rhoX} )}
  \\
  &=
  \sqrt{2 \lu \Log( \CONST \cdot 2^{\lu} \binom{\m}{\lu} \binom{\n}{\kX} \frac{3 \tauX \m}{\rhoX} )}
  \\
  &\leq
  \sqrt{
    2 \tauX \m
    \Log( \CONST \cdot 2^{\tauX \m} \binom{\m}{\tauX \m} \binom{\n}{\kX} \frac{3 \tauX \m}{\rhoX} )
  }
  \\
  &\dCmt{\(  \because \lu \leq \tauX \m  \)}
  \\
  &=
  \sqrt{
    2 \tauX \m \Log( 2^{\tauX \m} \binom{\m}{\tauX \m} )
    +
    2 \tauX \m \Log\binom{\n}{\kX}
    +
    2 \tauX \m \Log( \frac{3 \tauX \m}{\rhoX} )
    +
    2 \tauX \m \Log( \CONST )
  }
  \\
  &\leq
  \sqrt{
    2 ( \tauX \m )^{2} \Log( \Tfrac{2e}{\tauX} )
    +
    2 \tauX \m \Log \binom{\n}{\kX}
    +
    2 \tauX \m \Log( \frac{3 \tauX \m}{\rhoX} )
    +
    2 \tauX \m \Log( \CONST )
  }
  \\
  &\dCmt{\(  \because 2^{\tauX \m} \binom{\m}{\tauX \m} \leq 2^{\tauX \m} \left( \frac{e \m}{\tauX \m} \right)^{\tauX \m} = 2^{\tauX \m} \left( \frac{e}{\tauX} \right)^{\tauX \m} = \left( \frac{2e}{\tauX} \right)^{\tauX \m}  \)}
  \\
  &\leq
  \sqrt{2 ( \tauX \m )^{2} \Log( \Tfrac{2e}{\tauX} )}
  +
  \sqrt{2 \tauX \m \Log \binom{\n}{\kX}}
  +
  \sqrt{2 \tauX \m \Log( \frac{3 \tauX \m}{\rhoX} )}
  +
  \sqrt{2 \tauX \m \Log( \CONST )}
  \\
  &\dCmt{by the triangle inequality}
  \\
  &=
  \sqrt{2 ( \tauX \m )^{2} \Log( \Tfrac{2e}{\tauX} )}
  +
  \sqrt{\frac{2 \tauX \m^{2} \Log \binom{\n}{\kX}}{\m}}
  +
  \sqrt{\frac{2 \tauX \m^{2} \Log( \frac{3 \tauX \m}{\rhoX} )}{\m}}
  +
  \sqrt{\frac{2 \tauX \m^{2} \Log( \CONST )}{\m}}
  \\
  &\dCmt{by multiplying each of the last three terms by \(  \sqrt{\frac{\m}{\m}}  \)}
  \\
  &=
  \tauX \m \sqrt{2 \Log( \Tfrac{2e}{\tauX} )}
  +
  \m \sqrt{\frac{2 \tauX \Log \binom{\n}{\kX}}{\m}}
  +
  \m \sqrt{\frac{2 \tauX \Log( \frac{3 \tauX \m}{\rhoX} )}{\m}}
  +
  \m \sqrt{\frac{2 \tauX \Log( \CONST )}{\m}}
  \\
  &=
  \tauX \m \sqrt{2 \Log( \Tfrac{2e}{\tauX} )}
  +
  \m \sqrt{\frac{2 \epsilonRAICadv \tauX \Log \binom{\n}{\kX}}{\epsilonRAICadv \m}}
  +
  \m \sqrt{\frac{2 \epsilonRAICadv \tauX \Log( \frac{3 \tauX \m}{\rhoX} )}{\epsilonRAICadv \m}}
  +
  \m \sqrt{\frac{2 \epsilonRAICadv \tauX \Log( \CONST )}{\epsilonRAICadv \m}}
  \\
  &\dCmt{by multiplying each of the last three terms by \(  \sqrt{\frac{\epsilonRAICadv}{\epsilonRAICadv}}  \)}
  \\
  &\leq
  \tauX \m \sqrt{2 \Log( \Tfrac{2e}{\tauX} )}
  +
  \m\sqrt{\frac{2 \epsilonRAICadv \tauX}{\bC}}
  +
  \m \sqrt{\frac{2 \epsilonRAICadv \tauX}{\bC}}
  +
  \m \sqrt{\frac{2 \epsilonRAICadv \tauX}{\bC}}
  \\
  &\dCmt{\(  \because \epsilonRAICadv \m \geq \bC \max \left\{ \Log \binom{\n}{\kX}, \Log( \frac{3 \tauX \m}{\rhoX} ), \Log( \CONST ) \right\}  \),}
  \\
  &\dCmtx{where \(  \bC > 0  \) is a universal constant specified in \EQUATION \eqref{eqn:universal-constants}}
  \\
  &=
  \tauX \m \sqrt{2 \Log( \Tfrac{2e}{\tauX} )}
  +
  3 \m \sqrt{\frac{2 \epsilonRAICadv \tauX}{\bC}}
  \\
  &=
  \m \left(
  \tauX \sqrt{2 \Log( \Tfrac{2e}{\tauX} )}
  +
  3 \sqrt{\frac{2 \epsilonRAICadv \tauX}{\bC}}
  \right)
.\end{align*}
%
Then, dividing the above expressions by \(  \m  \), it follows that
\begin{align}
\label{eqn:pf:lemma:D2:1}
  \frac{\lu \tu}{\m}
  \leq
  \tauX \sqrt{2 \Log( \Tfrac{2e}{\tauX} )}
  +
  3 \sqrt{\frac{2 \epsilonRAICadv \tauX}{\bC}}
.\end{align}
%
Additionally, note that
\begin{gather}
  \label{eqn:pf:lemma:D2:2}
  \frac{\lu}{\m} \leq \frac{\tauX \m}{\m} = \tauX
  ,\\ \label{eqn:pf:lemma:D2:3}
  \frac{\sqrt{\pi ( \kX-1 ) \lu}}{\m}
  \leq
  \frac{\sqrt{\pi ( \kX-1 ) \tauX \m}}{\m}
  =
  \sqrt{\frac{\pi ( \kX-1 ) \tauX}{\m}}
  \leq
  \sqrt{\frac{\pi ( \kX-1 ) \epsilonRAICadv \tauX}{\bC \kX}}
  \leq
  \sqrt{\frac{\pi \epsilonRAICadv \tauX}{\bC}}
.\end{gather}
%
%
Combining the above results yields the following upper bound:
\begin{align*}
  \frac{2 \lu}{\m} + \frac{2\sqrt{2\pi} \lu \tu}{\m} + \frac{\sqrt{\pi ( \kX-1 ) \lu}}{\m}
  &\leq
  2 \tauX
  +
  2\sqrt{2\pi} \cdot \tauX \sqrt{2 \Log( \Tfrac{2e}{\tauX} )}
  +
  2\sqrt{2\pi} \cdot 3 \sqrt{\frac{2 \epsilonRAICadv \tauX}{\bC}}
  +
  \sqrt{\frac{\pi \epsilonRAICadv \tauX}{\bC}}
  \\
  &\dCmt{due to \EQUATIONS \eqref{eqn:pf:lemma:D2:1}-\eqref{eqn:pf:lemma:D2:3}}
  \\
  &=
  2 \tauX
  +
  4\sqrt{\pi} \cdot \tauX \sqrt{\Log( \Tfrac{2e}{\tauX} )}
  +
  12\sqrt{\frac{\pi \epsilonRAICadv \tauX}{\bC}}
  +
  \sqrt{\frac{\pi \epsilonRAICadv \tauX}{\bC}}
  \\
  &=
  2 \tauX
  +
  4\sqrt{\pi} \cdot \tauX \sqrt{\Log( \Tfrac{2e}{\tauX} )}
  +
  13\sqrt{\frac{\pi \epsilonRAICadv \tauX}{\bC}}
  \\
  &\leq
  2 \tauX \sqrt{\Log( \Tfrac{2e}{\tauX} )}
  +
  4\sqrt{\pi} \cdot \tauX \sqrt{\Log( \Tfrac{2e}{\tauX} )}
  +
  13\sqrt{\frac{\pi \epsilonRAICadv \tauX}{\bC}}
  \\
  &\dCmt{\(  \because \tauX \leq \tauX \sqrt{\Log( \Tfrac{2e}{\tauX} )}  \) for \(  \tauX \in (0,1]  \)}
  \\
  &=
  \left( 2 + 4\sqrt{\pi} \right) \tauX \sqrt{\Log( \Tfrac{2e}{\tauX} )}
  +
  13\sqrt{\frac{\pi \epsilonRAICadv \tauX}{\bC}}
  \\ \TagEqn\label{eqn:pf:lemma:D2:4}
  &=
  13\sqrt{\frac{\pi \epsilonRAICadv \tauX}{\bC}}
  +
  \left( 2 + 4\sqrt{\pi} \right) \tauX \sqrt{\Log( \Tfrac{2e}{\tauX} )}
.\end{align*}
%
Therefore, by \CLAIM \ref{claim:raic:adv:3} and an earlier remark,
with probability at least \(  1 - \frac{2\rhoX}{3}  \),
uniformly for all \(  \JX \subseteq [\n]  \), \(  | \JX | \leq \k  \), and
\(  \Vec{\uV} \in \US[\JX]  \),
\(  \DXXu  \) is bounded from above as follows:
\begin{align*}
  \DXXu
  &\leq
  \Dxu + \Dxxu
  \\
  &\dCmt{by \CLAIM \ref{claim:raic:adv:3}}
  \\
  &\leq
  \left( \Big. \frac{2 \lu}{\m} + \frac{\sqrt{2\pi} \lu \tu}{\m} \right)
  +
  \left( \Big. \frac{\sqrt{2\pi} \lu \tu}{\m} + \frac{\sqrt{\pi ( \kX-1 ) \lu}}{\m} \right)
  \\
  &\dCmt{by \LEMMAS \ref{lemma:D'1} and \ref{lemma:D'2}}
  \\
  &=
  \frac{2 \lu}{\m}
  +
  \frac{2\sqrt{2\pi} \lu \tu}{\m} + \frac{\sqrt{\pi ( \kX-1 ) \lu}}{\m}
  \\
  &\leq
  \DXXuUB
  \\
  &\dCmt{by \EQUATION \eqref{eqn:pf:lemma:D2:4}}
  \\
  &=
  \DXXuUBconst
  \\
  &\dCmt{due to an appropriate choice of the universal constants,}
  \\
  &\dCmtx{\(  \cX{3}, \cX{4} > 0  \), as defined in \EQUATION \eqref{eqn:universal-constants}}
\end{align*}
which completes the proof.
\end{proof}
\renewcommand{\ZX}{Z}

\renewcommand{\Xu}[1][]{\RV{X}_{\Vec{\uV}\IfNE{#1}{;#1}}}
\renewcommand{\Yu}[1][]{\Vec{Y}_{\Vec{\uV}\IfNE{#1}{;#1}}}
\renewcommand{\barXu}[1][]{\RV{\bar{X}}_{\Vec{\uV}\IfNE{#1}{;#1}}}
\RenewDocumentCommand{\barYu}{s O{}}{\IfBooleanTF{#1}{\bar{Y}}{\mathbf{\bar{Y}}}_{\Vec{\uV}\IfNE{#2}{;#2}}}

\section{Proofs of the Concentration Inequalities --                                    
         \LEMMAS \ref{lemma:concentration-ineq:1} and \ref{lemma:concentration-ineq:3}} 
\label{section:|>ortho-decomp}                                                          

\begin{lemma}
\label{lemma:distr:1}
Fix a \ksparserealunit vector,
\(  \Vec{\uV} \in \SparseSphereSubspace{\k}{\n}  \),
and let
\(  \JX \subseteq [\n]  \), \(  | \JX | \leq \k  \). 
%
Let
\(  \Vec{\ZX} \sim \N( \Vec{0}, \Id{\n} )  \)
be a Gaussian vector with \iid entries.
Define the random variable \(  \Xu  \) by
\begin{gather*}
  \Xu = \langle \Vec{\uV}, \TJZ \rangle \Sign( \langle \Vec{\uV}, \TJZ \rangle )
.\end{gather*}
%
Then,
\(  \Xu = | \langle \Vec{\uV}, \Vec{\ZX} \rangle | = | \Xu |  \),
and
\(  \E[ \Xu ] = \sqrt{\frac{2}{\pi}}  \).
%
\end{lemma}

\begin{proof}
{\LEMMA \ref{lemma:distr:1}}
Fix
\(  \Vec{\uV} \in \SparseSphereSubspace{\k}{\n}  \) and
\(  \JX \subseteq [\n]  \), \(  | \JX | \leq \k  \), 
arbitrarily, and taking
\(  \Vec{\ZX} \sim \N( \Vec{0}, \Id{\n} )  \),
let
\(  \Xu = \langle \Vec{\uV}, \TJZ \rangle \Sign( \langle \Vec{\uV}, \TJZ \rangle )  \).
%
Because
\(  \Supp( \Vec{\uV} ) \subseteq \Supp( \Vec{\uV} ) \cup \JX  \), trivially,
there is an equality:
\begin{align*}
  \langle \Vec{\uV}, \TJZ \rangle
  &=
  \sum_{\jIx=1}^{\n}
  \Vec*{\uV}_{\jIx} \ThresholdSet{\JX}( \Vec{\ZX} )_{\jIx}
  \\
  &=
  \sum_{\jIx \in \Supp( \Vec{\uV} ) \cap \Supp( \TJZ )}
  \Vec*{\uV}_{\jIx} \Vec*{\ZX}_{\jIx}
  \\
  &=
  \sum_{\jIx \in \Supp( \Vec{\uV} ) \cap ( \Supp( \Vec{\uV} ) \cup \JX )}
  \Vec*{\uV}_{\jIx} \Vec*{\ZX}_{\jIx}
  \\
  &=
  \sum_{\jIx \in \Supp( \Vec{\uV} )}
  \Vec*{\uV}_{\jIx} \Vec*{\ZX}_{\jIx}
  \\
  &=
  \sum_{\jIx=1}^{\n}
  \Vec*{\uV}_{\jIx} \Vec*{\ZX}_{\jIx}
  \\
  &=
  \langle \Vec{\uV}, \Vec{\ZX} \rangle
\end{align*}
%
Thus, the random variable \(  \Xu  \) is equivalently given by
\begin{align*}
  \Xu
  &= \langle \Vec{\uV}, \TJZ \rangle \Sign( \langle \Vec{\uV}, \TJZ \rangle )
  \\
  &= \langle \Vec{\uV}, \Vec{\ZX} \rangle \Sign( \langle \Vec{\uV}, \Vec{\ZX} \rangle )
\end{align*}
%
Note that for any \(  a \in \R  \),
\begin{align*}
  a \Sign( a )
  &=
  |a|.
\end{align*}
%
Therefore,
\(  \Xu = \langle \Vec{\uV}, \Vec{\ZX} \rangle \Sign( \langle \Vec{\uV}, \Vec{\ZX} \rangle )
        = | \langle \Vec{\uV}, \Vec{\ZX} \rangle |
        = | \Xu |  \),
as claimed.
By a well-known property of Gaussians,
\(  \langle \Vec{\uV}, \Vec{\ZX} \rangle \sim \N(0,1)  \),
and hence,
\(  \Xu = | \langle \Vec{\uV}, \Vec{\ZX} \rangle | \sim | \YRV |  \), where
\(  \YRV \sim \N(0,1)  \)
is a half-normal random variable.
Since \(  \Xu \sim | \YRV |  \), these two random variables are equal in expectation:
\(  \E[ \Xu ] = \E \left[ \big. | \YRV | \right] = \sqrt{\frac{2}{\pi}}  \).
\end{proof}

\begin{lemma}
\label{lemma:distr:2}
Fix
\(  \Vec{\uV}, \Vec{\vV} \in \SparseSphereSubspace{\k}{\n}  \) and
\(  \JX \subseteq [\n]  \), \(  | \JX | \leq \k  \),
\(  \Supp( \Vec{\vV} ) \subseteq \Supp( \Vec{\uV} ) \cup \JX  \).
%
Let
\begin{gather*}
  \Yuv =
  \bigl\langle
    \Vec{\vV},
    \TJZ \Sign( \langle \Vec{\uV}, \TJZ \rangle )
    -
    \langle \Vec{\uV}, \TJZ \rangle \Sign( \langle \Vec{\uV}, \TJZ \rangle ) \Vec{\uV}
  \bigr\rangle
.\end{gather*}
%
Then,
\(  \Yuv \sim \N(0,1)  \).
\end{lemma}

\let\oldkX\kX
\renewcommand{\kX}{k'}

\begin{lemma}
\label{lemma:norm-distr:2}
Fix
\(  \Vec{\uV} \in \SparseSphereSubspace{\k}{\n}  \) and
\(  \JX \subseteq [\n]  \), \(  | \JX | \leq \k  \),
and write \(  \kX = | \Supp( \Vec{\uV} ) \cup \JX | \leq 2\k  \).
Let
\begin{gather*}
  \barYu =
  \sum_{\iIx=1}^{\lu}
  \Big(
    \TJZi \Sign( \langle \Vec{\uV}, \TJZi \rangle )
    -
    \langle \Vec{\uV}, \TJZi \rangle \Sign( \langle \Vec{\uV}, \TJZi \rangle ) \Vec{\uV}
  \Big)
\end{gather*}
and let
\(  \Vec{W} \sim \N( \Vec{0}, \lu \Id{\kX-1} )  \).
%
Then,
\(  \| \Yu \|_{2} \sim \| \Vec{W} \|_{2}  \).
\end{lemma}

\let\kX\oldkX

\newcommand{\dDim}{d}
\begin{lemma}
\label{lemma:concentration-ineq:gaussian}
Let
\(  \Vec{\WRV} \sim \N(0,\sigma^{2} \Id{\dDim} )  \)
be a Gaussian random vector with \iid entries.
Then,
\begin{gather}
  \Pr \left( \| \Vec{W} \|_{2} > \sigma \sqrt{\frac{\dDim}{2}} + \sigma^{2} \tu \right)
  \leq
  e^{-\frac{1}{2} \sigma^{2} \tu^{2}}
.\end{gather}
\end{lemma}

\begin{proof}
{\LEMMA \ref{lemma:distr:2}}
Fix a pair of orthonormal vectors,
\(  \Vec{\uV}, \Vec{\vV} \in \SparseSphereSubspace{\k}{\n}  \),
arbitrarily, and let
\begin{gather*}
  \Yuv =
  \Bigl\langle
    \Vec{\vV},
    \TJZ \Sign( \langle \Vec{\uV}, \TJZ \rangle )
    -
    \langle \Vec{\uV}, \TJZ \rangle \Sign( \langle \Vec{\uV}, \TJZ \rangle ) \Vec{\uV}
  \Bigr\rangle
\end{gather*}
%
Observe:
\begin{align*}
  \Yuv
  &=
  \Bigl\langle
    \Vec{\vV},
    \TJZ \Sign( \langle \Vec{\uV}, \TJZ \rangle )
    -
    \langle \Vec{\uV}, \TJZ \rangle \Sign( \langle \Vec{\uV}, \TJZ \rangle ) \Vec{\uV}
  \Bigr\rangle
  \\
  &=
  \Bigl\langle
    \Vec{\vV},
    \TJZ \Sign( \langle \Vec{\uV}, \TJZ \rangle )
  \Bigr\rangle
  -
  \Bigl\langle
    \Vec{\vV},
    \langle \Vec{\uV}, \TJZ \rangle \Sign( \langle \Vec{\uV}, \TJZ \rangle ) \Vec{\uV}
  \Bigr\rangle
  \\
  &\dCmt{by the linearity of inner products}
  \\
  &=
  \Bigl\langle
    \Vec{\vV},
    \TJZ \Sign( \langle \Vec{\uV}, \TJZ \rangle )
  \Bigr\rangle
  -
  \langle \Vec{\uV}, \TJZ \rangle \Sign( \langle \Vec{\uV}, \TJZ \rangle )
  \bigl\langle
    \Vec{\vV},
    \Vec{\uV}
  \bigr\rangle
  \\
  &\dCmt{by the linearity of inner products}
  \\
  &=
  \Bigl\langle
    \Vec{\vV},
    \TJZ \Sign( \langle \Vec{\uV}, \TJZ \rangle )
  \Bigr\rangle
  - 0
  \\
  &\dCmt{by the orthogonality of \(  \Vec{\uV}  \) and \(  \Vec{\vV}  \)}
  \\
  &=
  \Bigl\langle
    \Vec{\vV},
    \TJZ \Sign( \langle \Vec{\uV}, \TJZ \rangle )
  \Bigr\rangle
  \\
  &=
  \Bigl\langle
    \Vec{\vV},
    \Vec{\ZX} \Sign( \langle \Vec{\uV}, \TJZ \rangle )
  \Bigr\rangle
  \\
  &\dCmt{as analogously argued earlier since \(  \Supp( \Vec{\vV} ) \subseteq \Supp( \Vec{\uV} ) \cup \JX  \)}
  \\
  &=
  \bigl\langle
    \Vec{\vV},
    \Vec{\ZX}
  \bigr\rangle
  \Sign( \langle \Vec{\uV}, \TJZ \rangle )
  \\
  &\dCmt{by the linearity of inner products}
\end{align*}
%
\par 
%
The remaining step is to show that
\(  \langle \Vec{\vV}, \Vec{\ZX} \rangle
    \Sign( \langle \Vec{\uV}, \TJZ \rangle )
    \sim \N(0,1)  \).
%
This can be achieved by a two-step argument:
(a) First, we will argue that \(  \langle \Vec{\vV}, \Vec{\ZX} \rangle  \) and
    \(  \Sign( \langle \Vec{\uV}, \TJZ \rangle )  \) are independent.
(b) Then, by standard facts about Gaussians and due to the independence shown in \STEP (a),
    the claim will follow.
Starting with \STEP (a), note that if
\(  \langle \Vec{\uV}, \Vec{\ZX} \rangle  \) and \(  \langle \Vec{\vV}, \Vec{\ZX} \rangle  \)
are independent, so are
\( \Sign( \langle \Vec{\uV}, \Vec{\ZX} \rangle )  \) and \( \langle \Vec{\vV}, \Vec{\ZX} \rangle \).
%
Therefore, it suffices to establish the independence of
\(  \langle \Vec{\uV}, \Vec{\ZX} \rangle  \) and \(  \langle \Vec{\vV}, \Vec{\ZX} \rangle  \).
%
Write
\(  \URV_{1} = \langle \Vec{\uV}, \Vec{\ZX} \rangle  \) and
\(  \URV_{2} = \langle \Vec{\vV}, \Vec{\ZX} \rangle  \).
%
By a well-known fact about Gaussians,
\(  \URV_{1}, \URV_{2} \sim \N(0,1)  \).
%
Now, consider the joint distribution of
\begin{gather*}
  \left( \begin{array}{c}
    \URV_{1} \\
    \URV_{2}
  \end{array} \right)
  \sim
  \N \left(
  \left( \begin{array}{c}
    0 \\
    0
  \end{array} \right),
  \Mat{\Sigma}
  \right)
\end{gather*}
which is a \(  0  \)-mean bivariate Gaussian with covariance matrix
\(  \Mat{\Sigma} \in \R^{2 \times 2}  \).
%
The goal is to show that
\(  \Mat{\Sigma} = \Id{2}  \).
%
Each \(  \iIx\Th  \) diagonal entry, \(  \iIx \in \{ 1,2 \}  \), is given by:
\begin{align*}
  \Mat*{\Sigma}_{\iIx,\iIx}
  &=
  \Cov( \URV_{\iIx}, \URV_{\iIx} )
  \\
  &=
  \Var( \URV_{\iIx} )
  \\
  &=
  1
\end{align*}
where the last line follows from the earlier observation that
\(  \URV_{1}, \URV_{2} \sim \N(0,1)  \).
%
On the other hand, each off-diagonal \(  ( \iIx, \jIx )  \)-entry,
\(  \iIx, \jIx \in \{ 1,2 \}  \), \(  \iIx \neq \jIx  \),
is obtained as follows.
Assuming without loss of generality that \(  \iIx = 1  \) and \(  \jIx = 2  \),
the corresponding covariance is \(  0  \)-valued due to the next derivation:
\begin{align*}
  \Mat*{\Sigma}_{\jIx,\iIx}
  =
  \Mat*{\Sigma}_{\iIx,\jIx}
  &=
  \Cov( \URV_{\iIx}, \URV_{\jIx} )
  \\
  &=
  \E [ \URV_{\iIx} \URV_{\jIx} ] - \E[ \URV_{\iIx} ] \E[ \URV_{\jIx} ]
  \\
  &=
  \E [ \URV_{\iIx} \URV_{\jIx} ]
  \\
  &=
  \E [ \URV_{1} \URV_{2} ]
  \\
  &=
  \E [ \Vec{\uV}^{\T} \Vec{\ZX} \Vec{\ZX}^{\T} \Vec{\vV} ]
  \\
  &=
  \Vec{\uV}^{\T} \E [ \Vec{\ZX} \Vec{\ZX}^{\T} ] \Vec{\vV}
  \\
  &=
  \Vec{\uV}^{\T} \Cov( \Vec{\ZX}, \Vec{\ZX} ) \Vec{\vV}
  \\
  &=
  \Vec{\uV}^{\T} \Id{2} \Vec{\vV}
  =
  \Vec{\uV}^{\T} \Vec{\vV}
  =
  \cos( \DistAngular{\Vec{\uV}}{\Vec{\vV}} )
  =
  \cos \left( \frac{\pi}{2} \right)
  =
  0
\end{align*}
as said.
From the above work, it follows that
\(  \Mat{\Sigma} = \Id{2}  \) and
\(  ( \URV_{1}, \URV_{2} ) \sim \N( \Vec{0}, \Id{2} )  \).
%
This therefore establishes the desired independence of
\(  \URV_{1} = \langle \Vec{\uV}, \Vec{\ZX} \rangle  \) and
\(  \URV_{2} = \langle \Vec{\vV}, \Vec{\ZX} \rangle  \).
%
The independence of
\(  \langle \Vec{\vV}, \Vec{\ZX} \rangle  \) and
\(  \Sign( \langle \Vec{\uV}, \Vec{\ZX} \rangle )  \)
follows, completing \STEP (a).
%
\par 
%
Proceeding to \STEP (b), recall that the goal of this step is to show that
\(  \langle \Vec{\vV}, \Vec{\ZX} \rangle \Sign( \langle \Vec{\uV}, \Vec{\ZX} \rangle )
    \sim \N(0,1)  \).
%
Note that because
\(  \langle \Vec{\uV}, \Vec{\ZX} \rangle \sim \N(0,1)  \),
the symmetry of the distribution \(  \N(0,1)  \) around its mean (\(  0  \)) leads to
\(  \Pr( \langle \Vec{\uV}, \Vec{\ZX} \rangle < 0 )
    = \Pr( \langle \Vec{\uV}, \Vec{\ZX} \rangle \geq 0 )
    = \frac{1}{2}
\), and hence
\(  \Pr ( \Sign( \langle \Vec{\uV}, \Vec{\ZX} \rangle ) = -1 )
    = \Pr ( \Sign( \langle \Vec{\uV}, \Vec{\ZX} \rangle ) = 1 )
    = \frac{1}{2}  \).
%
This in turn implies that
\(  \langle \Vec{\vV}, \Vec{\ZX} \rangle \Sign( \langle \Vec{\uV}, \Vec{\ZX} \rangle )
    \sim \RV{W} \RV{S}  \), where
\(  \RV{W} \sim \N(0,1)  \) and \(  \RV{S} \sim \{ -1,1 \}  \) are independent.
%
The density function of \(  \RV{W}  \) (a univariate Gaussian) is given at \(  w \in \R  \) by
\(  \pdf[\RV{W}]( w ) = \frac{1}{\sqrt{2\pi}} e^{-\frac{w^{2}}{2}}  \),
while the mass function of \(  \RV{S}  \) is given at \(  s \in \{ -1,1 \}  \) by
\(  \pdf[\RV{S}]( s ) = \frac{1}{2}  \)
and is otherwise \(  0  \)-valued.
%
Additionally,
\(  \pdf[-\RV{W}]( w ) = \frac{1}{\sqrt{2\pi}} e^{-\frac{(-w)^{2}}{2}}  \).
%
Due to the independence of \(  \RV{W}  \) and \(  \RV{S}  \),
their joint density function is simply the product of their individual densities:
\(  \pdf[\RV{W}, \RV{S}]( w,s ) = \pdf[\RV{W}]( w ) \pdf[\RV{S}]( s )  \).
%
Notice that
\begin{align*}
  ( \RV{W} \RV{S} \Mid| \RV{S}=s )
  =
  ( s \RV{W} \Mid| \RV{S}=s )
  =
  s \RV{W}
\end{align*}
where the last equality uses the independence discussed above.
The density of \(  \RV{W} \RV{S}  \) is then given at \(  z \in \R  \) by:
\begin{align*}
  \pdf[\RV{W} \RV{S}]( z )
  &=
  \pdf[\RV{W} \RV{S} | \RV{S}=-1]( z | -1 )
  \pdf[\RV{S}]( -1 )
  +
  \pdf[\RV{W} \RV{S} | \RV{S}=1]( z | 1 )
  \pdf[\RV{S}]( 1 )
  \\
  &\dCmt{by the law of total probability}
  \\
  &=
  \frac{1}{2}
  \pdf[\RV{W} \RV{S} | \RV{S}=-1]( z | -1 )
  +
  \frac{1}{2}
  \pdf[\RV{W} \RV{S} | \RV{S}=1]( z | 1 )
  \\
  &\dCmt{by the definition of \(  \pdf[\RV{S}]  \)}
  \\
  &=
  \frac{1}{2}
  \pdf[-\RV{W}]( z )
  +
  \frac{1}{2}
  \pdf[\RV{W}]( z )
  \\
  &\dCmt{by an earlier remark}
  \\
  &=
  \frac{1}{2}
  \cdot
  \frac{1}{\sqrt{2\pi}} e^{-\frac{(-z)^{2}}{2}}
  +
  \frac{1}{2}
  \frac{1}{\sqrt{2\pi}} e^{-\frac{z^{2}}{2}}
  \\
  &\dCmt{by the definitions of \(  \pdf[-\RV{W}], \pdf[\RV{W}]  \)}
  \\
  &=
  \frac{1}{2}
  \cdot
  \frac{1}{\sqrt{2\pi}} e^{-\frac{z^{2}}{2}}
  +
  \frac{1}{2}
  \cdot
  \frac{1}{\sqrt{2\pi}} e^{-\frac{z^{2}}{2}}
  \\
  &\dCmt{by squaring the negative term in the first exponent}
  \\
  &=
  \frac{1}{\sqrt{2\pi}} e^{-\frac{z^{2}}{2}}
  \\
  &\dCmt{by simplification}
\end{align*}
%
In short,
\(  \pdf[\RV{W} \RV{S}]( z ) = \frac{1}{\sqrt{2\pi}} e^{-\frac{z^{2}}{2}}  \),
which is precisely the density function of a standard normal random variable, \(  \N(0,1)  \).
Therefore,
\(   \langle \Vec{\vV}, \Vec{\ZX} \rangle \Sign( \langle \Vec{\uV}, \Vec{\ZX} \rangle )
     \sim \RV{W} \RV{S} \sim \N(0,1)  \).
%
This completes \STEP (b).
Moreover, combined with an earlier argument, it follows that
\(  \Yuv = \langle \Vec{\vV}, \Vec{\ZX} \rangle \Sign( \langle \Vec{\uV}, \Vec{\ZX} \rangle )
         \sim \N(0,1)  \),
as the lemma claimed.
\end{proof}

\let\oldkX\kX
\renewcommand{\kX}{k'}
\begin{proof}
{\LEMMA \ref{lemma:norm-distr:2}}
Fix
\(  \Vec{\uV} \in \SparseSphereSubspace{\k}{\n}  \) and
\(  \JX \subseteq [\n]  \), \(  | \JX | \leq \k  \), 
arbitrarily, where
\(  \kX = | \Supp( \Vec{\uV} ) \cup \JX | \leq 2\k  \),
and define
\begin{gather*}
  \barYu =
  \sum_{\iIx=1}^{\lu}
  \Big(
    \TJZi \Sign( \langle \Vec{\uV}, \Vec{\ZX}\VL{\iIx} \rangle )
    -
    \langle \Vec{\uV}, \TJZi \rangle \Sign( \langle \Vec{\uV}, \TJZi \rangle ) \Vec{\uV}
  \Big)
\end{gather*}
%
Let
\(  \Vec{W} \sim \N( \Vec{0}, \lu \Id{\kX-1} )  \).
%
The random variable of interest is \(  \| \barYu \|_{2}  \).
%
\begin{align*}
  \left\| \barYu \right\|_{2}
  &=
  \sqrt{\sum_{\jIx=1}^{\n} \barYu*[\jIx]^{2}}
  \\
  &=
  \sqrt{\sum_{\jIx \in \Supp( \Vec{\uV} ) \cup \JX} \barYu*[\jIx]^{2}}
  \\
  &=
  \left\| \Restriction{\barYu}{\Supp( \Vec{\uV} ) \cup \JX} \right\|_{2}
\end{align*}
%
Additionally, by an argument in the proof of \LEMMA \ref{lemma:distr:1},
\(  \langle \Vec{\uV}, \TJZi \rangle \Sign( \langle \Vec{\uV}, \TJZi \rangle )
    =
    \langle \Vec{\uV}, \ThresholdSet'{\Supp( \Vec{\uV} ) \cup \JX}( \Vec{\ZX}\VL{\iIx} ) \rangle \Sign( \langle \Vec{\uV}, \ThresholdSet'{\Supp( \Vec{\uV} ) \cup \JX}( \Vec{\ZX}\VL{\iIx} ) \rangle )
    =
    \langle \Vec{\uV}, \Vec{\ZX}\VL{\iIx} \rangle \Sign( \langle \Vec{\uV}, \Vec{\ZX}\VL{\iIx} \rangle )
\).
%
Hence, 
in order to simplify notations in this proof,
assume without loss of generality that
\(  \n = \kX = | \Supp( \Vec{\uV} ) \cup \JX |  \).
%
Then,
\(  \Vec{\uV}  \), \(  \Vec{\ZX}\VL{\iIx}  \), \(  \iIx \in [\lu]  \), and \(  \barYu  \)
are all \(  \kX  \)-dimensional, and the definition of \(  \barYu  \) can be written as
\begin{gather*}
  \barYu =
  \sum_{\iIx=1}^{\lu}
  \Big(
    \Vec{\ZX}\VL{\iIx} \Sign( \langle \Vec{\uV}, \Vec{\ZX}\VL{\iIx} \rangle )
    -
    \langle \Vec{\uV}, \Vec{\ZX}\VL{\iIx} \rangle
    \Sign( \langle \Vec{\uV}, \Vec{\ZX}\VL{\iIx} \rangle ) \Vec{\uV}
  \Big)
.\end{gather*}
%
Let
\(  \Basis = \{ \Vec{\basis}\VL{1}, \dots, \Vec{\basis}\VL{\kX} \} \subset \R^{\kX}  \)
be an orthonormal basis for \(  \R^{\kX}  \), where
\(  \Vec{\basis}\VL{\kX} = \Vec{\uV}  \).
%
Then,
\begin{align*}
  \barYu
  &=
  \sum_{\iIx=1}^{\lu}
  \Bigl(
    \Vec{\ZX}\VL{\iIx} \Sign( \langle \Vec{\uV}, \Vec{\ZX}\VL{\iIx} \rangle )
    -
    \langle \Vec{\uV}, \Vec{\ZX}\VL{\iIx} \rangle
    \Sign( \langle \Vec{\uV}, \Vec{\ZX}\VL{\iIx} \rangle ) \Vec{\uV}
  \Bigr)
  \\
  &=
  \sum_{\iIx=1}^{\lu}
  \Bigl(
    \Vec{\ZX}\VL{\iIx}
    -
    \langle \Vec{\uV}, \Vec{\ZX}\VL{\iIx} \rangle
    \Vec{\uV}
  \Bigr)
  \Sign( \langle \Vec{\uV}, \Vec{\ZX}\VL{\iIx} \rangle )
  \\
  &\dCmt{by distributivity}
  \\
  &=
  \sum_{\iIx=1}^{\lu}
  \sum_{\jIx=1}^{\kX}
  \Bigl\langle
    \Vec{\basis}\VL{\jIx},
    \bigl(
      \Vec{\ZX}\VL{\iIx}
      -
      \langle \Vec{\uV}, \Vec{\ZX}\VL{\iIx} \rangle
      \Vec{\uV}
    \bigr)
    \Sign( \langle \Vec{\uV}, \Vec{\ZX}\VL{\iIx} \rangle )
  \Bigr\rangle
  \Vec{\basis}\VL{\jIx}
  \\
  &\dCmt{orthogonal decomposition via the basis \(  \Basis  \)}
  \\
  &=
  \sum_{\iIx=1}^{\lu}
  \sum_{\jIx=1}^{\kX}
  \Bigl\langle
    \Vec{\basis}\VL{\jIx},
    \Vec{\ZX}\VL{\iIx}
    -
    \langle \Vec{\uV}, \Vec{\ZX}\VL{\iIx} \rangle
    \Vec{\uV}
  \Bigr\rangle
  \Vec{\basis}\VL{\jIx}
  \Sign( \langle \Vec{\uV}, \Vec{\ZX}\VL{\iIx} \rangle )
  \\
  &\dCmt{by the linearity of inner products}
  \\
  &=
  \sum_{\iIx=1}^{\lu}
  \sum_{\jIx=1}^{\kX}
  \Bigl(
  \langle
    \Vec{\basis}\VL{\jIx},
    \Vec{\ZX}\VL{\iIx}
  \rangle
  -
  \langle
    \Vec{\basis}\VL{\jIx},
    \langle \Vec{\uV}, \Vec{\ZX}\VL{\iIx} \rangle
    \Vec{\uV}
  \rangle
  \Bigr)
  \Vec{\basis}\VL{\jIx}
  \Sign( \langle \Vec{\uV}, \Vec{\ZX}\VL{\iIx} \rangle )
  \\
  &\dCmt{by the linearity of inner products}
  \\
  &=
  \sum_{\iIx=1}^{\lu}
  \sum_{\jIx=1}^{\kX-1}
  \Bigl(
  \langle
    \Vec{\basis}\VL{\jIx},
    \Vec{\ZX}\VL{\iIx}
  \rangle
  -
  \langle
    \Vec{\basis}\VL{\jIx},
    \langle \Vec{\uV}, \Vec{\ZX}\VL{\iIx} \rangle
    \Vec{\uV}
  \rangle
  \Bigr)
  \Vec{\basis}\VL{\jIx}
  \Sign( \langle \Vec{\uV}, \Vec{\ZX}\VL{\iIx} \rangle )
  +
  \Bigl(
  \langle
    \Vec{\uV},
    \Vec{\ZX}\VL{\iIx}
  \rangle
  -
  \langle
    \Vec{\uV},
    \langle \Vec{\uV}, \Vec{\ZX}\VL{\iIx} \rangle
    \Vec{\uV}
  \rangle
  \Bigr)
  \Vec{\basis}\VL{\jIx}
  \Sign( \langle \Vec{\uV}, \Vec{\ZX}\VL{\iIx} \rangle )
  \\
  &\dCmt{by distributivity}
  \\
  &=
  \sum_{\iIx=1}^{\lu}
  \sum_{\jIx=1}^{\kX-1}
  \Bigl(
  \langle
    \Vec{\basis}\VL{\jIx},
    \Vec{\ZX}\VL{\iIx}
  \rangle
  -
  \langle
    \Vec{\basis}\VL{\jIx},
    \langle \Vec{\uV}, \Vec{\ZX}\VL{\iIx} \rangle
    \Vec{\uV}
  \rangle
  \Bigr)
  \Vec{\basis}\VL{\jIx}
  \Sign( \langle \Vec{\uV}, \Vec{\ZX}\VL{\iIx} \rangle )
  +
  \Bigl(
  \langle
    \Vec{\uV},
    \Vec{\ZX}\VL{\iIx}
  \rangle
  -
  \langle
    \Vec{\uV},
    \Vec{\uV}
  \rangle
  \langle \Vec{\uV}, \Vec{\ZX}\VL{\iIx} \rangle
  \Bigr)
  \Vec{\basis}\VL{\jIx}
  \Sign( \langle \Vec{\uV}, \Vec{\ZX}\VL{\iIx} \rangle )
  \\
  &\dCmt{by the linearity of inner products}
  \\
  &=
  \sum_{\iIx=1}^{\lu}
  \sum_{\jIx=1}^{\kX-1}
  \Bigl(
  \langle
    \Vec{\basis}\VL{\jIx},
    \Vec{\ZX}\VL{\iIx}
  \rangle
  -
  \langle
    \Vec{\basis}\VL{\jIx},
    \langle \Vec{\uV}, \Vec{\ZX}\VL{\iIx} \rangle
    \Vec{\uV}
  \rangle
  \Bigr)
  \Vec{\basis}\VL{\jIx}
  \Sign( \langle \Vec{\uV}, \Vec{\ZX}\VL{\iIx} \rangle )
  +
  \Bigl(
  \langle \Vec{\uV}, a\Vec{\ZX}\VL{\iIx} \rangle
  -
  \langle \Vec{\uV}, \Vec{\ZX}\VL{\iIx} \rangle
  \Bigr)
  \Vec{\basis}\VL{\jIx}
  \Sign( \langle \Vec{\uV}, \Vec{\ZX}\VL{\iIx} \rangle )
  \\
  &\dCmt{\(  \because \langle \Vec{\uV}, \Vec{\uV} \rangle = \| \Vec{\uV} \|_{2}^{2} = 1  \)}
  \\
  &=
  \sum_{\iIx=1}^{\lu}
  \sum_{\jIx=1}^{\kX-1}
  \Bigl(
  \langle
    \Vec{\basis}\VL{\jIx},
    \Vec{\ZX}\VL{\iIx}
  \rangle
  -
  \langle
    \Vec{\basis}\VL{\jIx},
    \langle \Vec{\uV}, \Vec{\ZX}\VL{\iIx} \rangle
    \Vec{\uV}
  \rangle
  \Bigr)
  \Vec{\basis}\VL{\jIx}
  \Sign( \langle \Vec{\uV}, \Vec{\ZX}\VL{\iIx} \rangle )
  +
  0 \cdot
  \Vec{\basis}\VL{\jIx}
  \Sign( \langle \Vec{\uV}, \Vec{\ZX}\VL{\iIx} \rangle )
  \\
  &\dCmt{\(  \because \langle \Vec{\uV}, \Vec{\ZX}\VL{\iIx} \rangle - \langle \Vec{\uV}, \Vec{\ZX}\VL{\iIx} \rangle = 0  \)}
  \\
  &=
  \sum_{\iIx=1}^{\lu}
  \sum_{\jIx=1}^{\kX-1}
  \Bigl(
  \langle
    \Vec{\basis}\VL{\jIx},
    \Vec{\ZX}\VL{\iIx}
  \rangle
  -
  \langle
    \Vec{\basis}\VL{\jIx},
    \langle \Vec{\uV}, \Vec{\ZX}\VL{\iIx} \rangle
    \Vec{\uV}
  \rangle
  \Bigr)
  \Vec{\basis}\VL{\jIx}
  \Sign( \langle \Vec{\uV}, \Vec{\ZX}\VL{\iIx} \rangle )
  \\
  &\dCmt{via simplification}
  \\
  &=
  \sum_{\iIx=1}^{\lu}
  \sum_{\jIx=1}^{\kX-1}
  \Bigl(
  \langle
    \Vec{\basis}\VL{\jIx},
    \Vec{\ZX}\VL{\iIx}
  \rangle
  -
  \langle
    \Vec{\basis}\VL{\jIx},
    \Vec{\uV}
  \rangle
  \langle \Vec{\uV}, \Vec{\ZX}\VL{\iIx} \rangle
  \Bigr)
  \Vec{\basis}\VL{\jIx}
  \Sign( \langle \Vec{\uV}, \Vec{\ZX}\VL{\iIx} \rangle )
  \\
  &\dCmt{by the linearity of inner products}
  \\
  &=
  \sum_{\iIx=1}^{\lu}
  \sum_{\jIx=1}^{\kX-1}
  \langle
    \Vec{\basis}\VL{\jIx},
    \Vec{\ZX}\VL{\iIx}
  \rangle
  \Vec{\basis}\VL{\jIx}
  \Sign( \langle \Vec{\uV}, \Vec{\ZX}\VL{\iIx} \rangle )
  \\
  &\dCmt{\(  \because \Vec{\basis}\VL{\jIx} \perp \Vec{\basis}\VL{\kX} =  \Vec{\uV}  \) when \(  \jIx \neq \kX  \)}
  \\
  &=
  \sum_{\jIx=1}^{\kX-1}
  \sum_{\iIx=1}^{\lu}
  \langle \Vec{\basis}\VL{\jIx}, \Vec{\ZX}\VL{\iIx} \rangle
  \Vec{\basis}\VL{\jIx}
  \Sign( \langle \Vec{\uV}, \Vec{\ZX}\VL{\iIx} \rangle )
  \\
  &\dCmt{the summations can be reordered since they do not have dependencies}
  \\
  &=
  \sum_{\jIx=1}^{\kX-1}
  \sum_{\iIx=1}^{\lu}
  \langle
    \Vec{\basis}\VL{\jIx},
    \Vec{\ZX}\VL{\iIx} \Sign( \langle \Vec{\uV}, \Vec{\ZX}\VL{\iIx} \rangle )
  \rangle
  \Vec{\basis}\VL{\jIx}
  \\
  &\dCmt{by the linearity of inner products}
  \\
  &=
  \sum_{\jIx=1}^{\kX-1}
  \Vec{\basis}\VL{\jIx}
  \left\langle \Vec{\basis}\VL{\jIx},
  \sum_{\iIx=1}^{\lu}
  \Vec{\ZX}\VL{\iIx}
  \Sign( \langle \Vec{\uV}, \Vec{\ZX}\VL{\iIx} \rangle )
  \right\rangle
  \\
  &\dCmt{by the linearity of inner products}
\end{align*}
%
Let
\(  \RV{S}_{1}, \dots, \RV{S}_{\lu} \sim \{ -1,1 \}  \)
be \iid Rademacher random variables which are also independent of
\(  \Vec{\ZX}\VL{1}, \dots, \Vec{\ZX}\VL{\lu}  \).
Due to an argument in the proof of \LEMMA \ref{lemma:distr:2},
\(
  \Vec{\ZX}\VL{\iIx} \Sign( \langle \Vec{\uV}, \Vec{\ZX}\VL{\iIx} \rangle )
  \sim
  \Vec{\ZX}\VL{\iIx} \RV{S}_{\iIx}
  \sim
  \Vec{\ZX}\VL{\iIx}
  \sim
  \N(\Vec{0}, \Id{\kX})
\),
and additionally, the random vectors,
\(  \{ \Vec{\ZX}\VL{\iIx} \Sign( \langle \Vec{\uV}, \Vec{\ZX}\VL{\iIx} \rangle ) \}_{\iIx \in [\lu]}  \),
are mutually independent.
Hence,
\begin{gather*}
  \sum_{\iIx=1}^{\lu}
  \Vec{\ZX}\VL{\iIx}
  \Sign( \langle \Vec{\uV}, \Vec{\ZX}\VL{\iIx} \rangle )
  \sim
  \N( \Vec{0}, \lu \Id{\kX-1} )
.\end{gather*}
%
By an argument analogous to that which appeared in the proof of \LEMMA \ref{lemma:distr:2},
the random variables,
\(
  \left\{
  \left\langle \Vec{\basis}\VL{\jIx},
  \sum_{\iIx=1}^{\lu}
  \Vec{\ZX}\VL{\iIx}
  \Sign( \langle \Vec{\uV}, \Vec{\ZX}\VL{\iIx} \rangle )
  \right\rangle
  \right\}_{\jIx \in [\kX-1]}
\),
are mutually independent, and therefore,
\begin{gather*}
  \left\langle \Vec{\basis}\VL{\jIx},
  \sum_{\iIx=1}^{\lu}
  \Vec{\ZX}\VL{\iIx}
  \Sign( \langle \Vec{\uV}, \Vec{\ZX}\VL{\iIx} \rangle )
  \right\rangle
  \sim
  \RV{W}_{\jIx}
  \sim
  \N(0,\sigma^{2}=\lu)
,\end{gather*}
where the random variables, \(  \{ \RV{W}_{\jIx} \}_{\jIx \in [\kX-1]}  \), 
are likewise mutually independent.
Using these random variables, the lemma's result is obtained as follows:
\begin{align*}
  \left\| \barYu \right\|_{2}
  &=
  \left\|
    \sum_{\jIx=1}^{\kX-1}
    \Vec{\basis}\VL{\jIx}
    \left\langle \Vec{\basis}\VL{\jIx},
    \sum_{\iIx=1}^{\lu}
    \Vec{\ZX}\VL{\iIx}
    \Sign( \langle \Vec{\uV}, \Vec{\ZX}\VL{\iIx} \rangle )
    \right\rangle
  \right\|_{2}
  \\
  &=
  \sqrt{
  \sum_{\jIx=1}^{\kX-1}
  \left\langle \Vec{\basis}\VL{\jIx},
  \sum_{\iIx=1}^{\lu}
  \Vec{\ZX}\VL{\iIx}
  \Sign( \langle \Vec{\uV}, \Vec{\ZX}\VL{\iIx} \rangle )
  \right\rangle^{2}
  }
  \\
  &\sim
  \sqrt{
  \sum_{\jIx=1}^{\kX-1}
  \RV{W}_{\jIx}^{2}
  }
  \\
  &=
  \left\| \Vec{W} \right\|_{2}
\end{align*}
%
To summarize, we have now shown that
\(  \| \barYu \|_{2} \sim \| \Vec{W} \|_{2}  \),
where
\(  \Vec{W} \sim \N( \Vec{0}, \lu \Id{\kX-1} )  \),
thus completing the proof of \LEMMA \ref{lemma:norm-distr:2}.
\end{proof}

\let\kX\oldkX

\begin{proof}
{\LEMMA \ref{lemma:concentration-ineq:gaussian}}
Let
\(  \Vec{U} \sim \N( \Vec{0}, \Id{\dDim} )  \) and
\(  \Vec{W} \sim \N( \Vec{0}, \sigma^{2} \Id{\dDim} )  \).
%
Note that
\(  \Vec{W} \sim \sigma \Vec{U}  \) and
\(  \| \Vec{W} \|_{2} \sim \sigma \| \Vec{U} \|_{2}  \),
and hence,
\begin{align*}
  \E \bigl[ \| \Vec{W} \|_{2} \bigr]
  =
  \E \bigl[ \| \sigma \Vec{U} \|_{2} \bigr]
  =
  \E \bigl[ \sigma \| \Vec{U} \|_{2} \bigr]
  =
  \sigma \E \bigl[ \| \Vec{U} \|_{2} \bigr]
.\end{align*}
%
It is well-known that
\(  \| \Vec{U} \|_{2} \sim \chi_{\dDim}  \),
and therefore,
\begin{align*}
  \E \bigl[ \| \Vec{U} \|_{2} \bigr]
  =
  \frac{\Gamma \left( \frac{\dDim+1}{2} \right)}{\Gamma \left( \frac{\dDim}{2} \right)}
  \leq
  \sqrt{\frac{\dDim}{2}}
,\end{align*}
where the inequality on the \RHS can be derived from the Legendre duplication formula
and Stirling's approximation.
Plugging this into the expression for the expectation of \(  \| \Vec{W} \|_{2}  \) yields
\begin{align*}
  \E \bigl[ \| \Vec{W} \|_{2} \bigr]
  =
  \sigma \E \bigl[ \| \Vec{U} \|_{2} \bigr]
  \leq
  \sigma \sqrt{\frac{\dDim}{2}}
  =
  \sqrt{\frac{\sigma^{2} \dDim}{2}}
.\end{align*}
%
By a standard concentration inequality for \(  L  \)-Lipschitz functions on Gaussian vectors,
where here, \(  \| \cdot \|_{2}  \) is \(  ( L=1 )  \)-Lipschitz
(\see e.g., \cite{wainwright2019high}),
\begin{gather*}
  \Pr \left(
    \| \Vec{U} \|_{2} > \sqrt{\frac{\dDim}{2}} + \tu'
  \right)
  \leq
  \Pr \left(
    \| \Vec{U} \|_{2} > \E \bigl[ \| \Vec{U} \|_{2} \bigr] + \tu'
  \right)
  \leq
  e^{-\frac{\tu'^{2}}{2 L^{2}}}
  =
  e^{-\frac{\tu'^{2}}{2}}
.\end{gather*}
%
Setting
\(  \tu' = \sigma \tu  \),
\begin{gather*}
  \Pr \left(
    \| \Vec{U} \|_{2} > \sqrt{\frac{\dDim}{2}} + \sigma \tu
  \right)
  \leq
  e^{-\frac{\sigma^{2} \tu^{2}}{2}}
.\end{gather*}
%
Finally, by the earlier observation that
\(  \Vec{W} \sim \sigma \Vec{U}  \) and
\(  \| \Vec{W} \|_{2} \sim \sigma \| \Vec{U} \|_{2}  \)
the lemma's concentration inequality follows:
\begin{gather*}
  \Pr \left(
    \| \Vec{W} \|_{2} > \sigma \sqrt{\frac{\dDim}{2}} + \sigma^{2} \tu
  \right)
  =
  \Pr \left(
    \sigma \| \Vec{U} \|_{2} > \sigma \sqrt{\frac{\dDim}{2}} + \sigma^{2} \tu
  \right)
  =
  \Pr \left(
    \| \Vec{U} \|_{2} > \sqrt{\frac{\dDim}{2}} + \sigma \tu
  \right)
  \leq
  e^{-\frac{\sigma^{2} \tu^{2}}{2}}
.\end{gather*}
\end{proof}

\begin{proof}
{\LEMMA \ref{lemma:concentration-ineq:1}}
Fix
\(  \Vec{\uV} \in \SparseSphereSubspace{\k}{\n}  \) and
\(  \JX \subseteq [\n]  \), \(  | \JX | \leq \k  \), 
arbitrarily.
Taking
\(  \Vec{\ZX}\VL{1}, \dots, \Vec{\ZX}\VL{\lu} \sim \N( \Vec{0}, \Id{\n} )  \),
let
\(  \Xu[\iIx] = \langle \Vec{\uV}, \TJZi \rangle \Sign( \langle \Vec{\uV}, \TJZi \rangle )  \),
\(  \iIx \in [\lu]  \),
and write
\(  \barXu = \sum_{\iIx=1}^{\lu} \Xu[\iIx]  \).
%
Let
\(  \YRV[1], \dots, \YRV[\lu] \sim \N(0,1)  \)
be independent standard normal random variables, and let
\(  \barYRV = \sum_{\iIx=1}^{\lu} | \YRV[\iIx] |  \).
%
By \LEMMA \ref{lemma:distr:1},
\(  \Xu[\iIx] \sim | \YRV[\iIx] |  \)
for each \(  \iIx \in [\lu]  \), and therefore,
\(  \barXu \sim \barYRV  \).
%
Hence, since \(  \barXu  \) and \(  \barYRV  \) follow the same distribution,
it suffices to bound the concentration the random variable \(  \barYRV  \).
Define the random vector,
\(  \Vec{\YRV} = ( \YRV[1], \dots, \YRV[\lu] )  \),
and note that
\begin{align*}
  \barYRV
  = \sum_{\iIx=1}^{\lu} | \YRV[\iIx] |
  = \| \Vec{\YRV} \|_{1}
.\end{align*}
%
Recall that for any \(  \Vec{w} \in \R^{\lu}  \),
\(  \| \Vec{w} \|_{1} \leq \sqrt{\lu} \| \Vec{w} \|_{2}  \)
since
\begin{align*}
  \| \Vec{w} \|_{1}
  &= \sum_{\iIx=1}^{\lu} | \Vec*{w}_{\iIx} |
  \\
  &= \langle \Vec{1}, ( | \Vec*{w}_{1} |, \dots, | \Vec*{w}_{\lu} | ) \rangle
  \\
  &\leq \| \Vec{1} \|_{2} \| ( | \Vec*{w}_{1} |, \dots, | \Vec*{w}_{\lu} | ) \|_{2}
  \\
  &\dCmt{by the Cauchy-Schwarz inequality}
  \\
  &= \| \Vec{1} \|_{2}
     \sqrt{\sum_{\jIx=1}^{\lu} | \Vec*{w}_{\jIx} |^{2}}
  \\
  &= \| \Vec{1} \|_{2}
     \sqrt{\sum_{\jIx=1}^{\lu} \Vec*{w}_{\jIx}^{2}}
  \\
  &= \| \Vec{1} \|_{2} \| ( \Vec*{w}_{1}, \dots, \Vec*{w}_{\lu} ) \|_{2}
  \\
  &= \| \Vec{1} \|_{2} \| \Vec{w} \|_{2}
  \\
  &= \sqrt{\lu} \| \Vec{w} \|_{2}
\end{align*}
%
Additionally, observe:
\begin{align*}
  \| \Vec{v} \|_{1}
  &=
  \| ( \Vec{v} - \Vec{w} ) + \Vec{w} \|_{1}
  \\
  &\leq
  \| \Vec{v} - \Vec{w} \|_{1} + \|\Vec{w} \|_{1}
  \\
  &\dCmt{by the triangle inequality}
  \\
  &\longrightarrow
  \| \Vec{v} \|_{1} \leq \| \Vec{v} - \Vec{w} \|_{1} + \|\Vec{w} \|_{1}
  \\
  &\longrightarrow
  \| \Vec{v} \|_{1} - \|\Vec{w} \|_{1} \leq \| \Vec{v} - \Vec{w} \|_{1}
  \\
  &\dCmt{rearrangement of terms}
  \\
  &\longrightarrow
  \| \Vec{v} \|_{1} - \|\Vec{w} \|_{1} \leq \sqrt{\lu} \| \Vec{v} - \Vec{w} \|_{2}
  \\
  &\dCmt{as argued earlier}
\end{align*}
and thus, \(  \| \cdot \|_{1}  \) is \(  L  \)-Lipschitz, where
\(  L = \sqrt{\lu}  \).
%
By a standard concentration for Gaussian random vectors under \(  L  \)-Lipschitz functions
(\see e.g., \cite{wainwright2019high}),
\begin{align*}
  \Pr \left( \Big.
    \| \Vec{\YRV} \|_{1} \geq \E \left[ \big. \| \Vec{\YRV} \|_{1} \right] + \lu \tu
  \right)
  \leq
  e^{-\frac{\lu^{2} \tu^{2}}{2L^{2}}}
  =
  e^{-\frac{\lu^{2} \tu^{2}}{2\lu}}
  =
  e^{-\frac{1}{2} \lu \tu^{2}}
,\end{align*}
where
\begin{align*}
  \E \left[ \big. \| \Vec{\YRV} \|_{1} \right]
  &=
  \E \left[ \sum_{\iIx=1}^{\lu} | \YRV[\iIx] | \right]
  \\
  &=
  \sum_{\iIx=1}^{\lu} \E \left[ \big. | \YRV[\iIx] | \right]
  \\
  &\dCmt{by the linearity of expectation}
  \\
  &=
  \sum_{\iIx=1}^{\lu} \sqrt{\frac{2}{\pi}}
  \\
  &\dCmt{the mean of a half-normal random variable (well-known)}
  \\
  &=
  \sqrt{\frac{2}{\pi}} \lu
\end{align*}
%
Combining the last two derivations yields:
\begin{align*}
  \Pr \left(
    \| \Vec{\YRV} \|_{1} \geq \left( \sqrt{\frac{2}{\pi}} + \tu \right) \lu
  \right)
  \leq
  e^{-\frac{1}{2} \lu \tu^{2}}
.\end{align*}
%
From this and the earlier discussion,
since
\(  \barXu \sim \barYRV = \| \Vec{\YRV} \|_{1}  \),
it follows that
\begin{align*}
  \Pr \left(
    \barXu \geq \left( \sqrt{\frac{2}{\pi}} + \tu \right) \lu
  \right)
  \leq
  e^{-\frac{1}{2} \lu \tu^{2}}
\end{align*}
as desired.
\end{proof}

\begin{proof}
{\LEMMA \ref{lemma:concentration-ineq:3}}
Let
\(  \Vec{\uV} \in \SparseSphereSubspace{\k}{\n}  \) and
\(  \JX \subseteq [\n]  \), \(  | \JX | \leq \k  \), 
and let
\(  \Vec{\ZX}\VL{1}, \dots, \Vec{\ZX}\VL{\lu} \sim \N( \Vec{0}, \Id{\n} )  \)
and
\(  \Vec{\WRV} \sim \N( \Vec{0}, \lu \Id{\kX-1} )  \)
be independent Gaussian vectors, each with \iid entries.
Define the random variable
\begin{gather*}
  \barYu =
  \sum_{\iIx=1}^{\lu}
  \Big(
    \TJZi
    -
    \langle \Vec{\uV}, \TJZi \rangle \Sign( \langle \Vec{\uV}, \TJZi \rangle )
    \Vec{\uV}
  \Big)
\end{gather*}
where the random variable of interest is
\(  \| \barYu \|_{2}  \).
%
The majority of the necessary work has already been achieved in
\LEMMAS \ref{lemma:distr:2} and \ref{lemma:concentration-ineq:gaussian}.
%
%
By \LEMMA \ref{lemma:distr:2},
\(  \| \barYu \|_{2} \sim \| \Vec{\WRV} \|_{2}  \),
and thus, by \LEMMA \ref{lemma:concentration-ineq:gaussian},
\begin{align*}
  \Pr \left(
    \| \barYu \|_{2} > \sqrt{( \kX-1 ) \lu} + \lu \tu
  \right)
  \leq
  e^{-\frac{1}{2} \lu \tu^{2}}
\end{align*}
as claimed.
\end{proof}
\section{Proof of the Deterministic Results, \LEMMAS \ref{lemma:error:deterministic} and \ref{lemma:error:recurrence}} %
\label{section:|>pf-deterministic} 

\subsection{Proof of \LEMMA \ref{lemma:error:deterministic}} 
\label{section:|>pf-deterministic|>pf-deterministic}         


\begin{proof}
{\LEMMA \ref{lemma:error:deterministic}}
The proof will focus on verifying a slight generalization of \LEMMA \ref{lemma:error:deterministic},
which is formally stated as the following claim.

\begin{claim}
\label{claim:pf:lemma:error:deterministic:1}
Let
\(  \Vec{\uX}, \Vec{\vX}, \Vec{\zX} \in \SparseSphereSubspace{\k}{\n}  \), and
\(  \Vec{\wX} \in \R^{\n}  \),
where
\begin{gather}
\label{eqn:claim:pf:lemma:error:deterministic:1:1}
  \Vec{\uX}
  =
  \frac{\Threshold{\k}( \Vec{\vX} + \Vec{\wX} )}
       {\| \Threshold{\k}( \Vec{\vX} + \Vec{\wX} ) \|_{2}}
%
,\end{gather}
and where
\(  \| \Vec{\vX} + \Vec{\wX} \|_{0} \geq k  \).
%
Then,
\begin{gather}
\label{eqn:claim:pf:lemma:error:deterministic:1:2}
  \left\| \Vec{\zX} - \Vec{\uX} \right\|_{2}
  \leq
  4 \left\| ( \Vec{\zX} - \Vec{\vX} ) - \ThresholdSet'{\Supp( \Vec{\zX} ) \cup \Supp( \Vec{\uX} ) \cup \Supp( \Vec{\vX} )}( \Vec{\wX} ) \right\|_{2}
.\end{gather}
\end{claim}

Note that
\(  \| \Vec{\xApprox}\IL{\Iter-1} \|_{0} \leq \kX   \),
by design, and since the random vector
\(  \hf[\Supp( \Vec{\xApprox}\IL{\Iter} )]{\Vec{\x}}{\Vec{\xApprox}\IL{\Iter-1}}  \)
follows a continuous distribution,
\(  \| \hf[\Supp( \Vec{\xApprox}\IL{\Iter} )]{\Vec{\x}}{\Vec{\xApprox}\IL{\Iter-1}} \|_{0} = \n  \).
%
Moreover, due to the condition that \(  \n \geq 2\k  \),
\(  \| \Vec{\xApprox}\IL{\Iter-1} + \hf[\Supp( \Vec{\xApprox}\IL{\Iter} )]{\Vec{\x}}{\Vec{\xApprox}\IL{\Iter-1}} \|_{0} \geq \n-\k \geq 2\k-\k = \k  \).
%
Hence, by taking
\(  \Vec{\zX} = \Vec{\x}  \),
\(  \Vec{\uX} = \Vec{\xApprox}\IL{\Iter}  \),
\(  \Vec{\vX} = \Vec{\xApprox}\IL{\Iter-1}  \), and
\(  \Vec{\wX} = \hf[\Supp( \Vec{\xApprox}\IL{\Iter} )]{\Vec{\x}}{\Vec{\xApprox}\IL{\Iter-1}}  \),
where
\(  \| \Vec{\vX} + \Vec{\wX} \|_{0} \geq \k  \)
due to the above discussion,
\CLAIM \ref{claim:pf:lemma:error:deterministic:1} bounds the approximation error,
\(  \DistS{\Vec{\x}}{\Vec{\xApprox}\IL{\Iter}}  \), as follows:
\begin{align*}
  \DistS{\Vec{\x}}{\Vec{\xApprox}\IL{\Iter}}
  &=
  \| \Vec{\x} - \Vec{\xApprox}\IL{\Iter} \|_{2}
  \leq
  4 \| ( \Vec{\x} - \Vec{\xApprox}\IL{\Iter-1} ) - \hf[\Supp( \Vec{\xApprox}\IL{\Iter} )]{\Vec{\x}}{\Vec{\xApprox}\IL{\Iter-1}} \|_{2}
.\end{align*}
%
Hence, the proof of \LEMMA \ref{lemma:error:deterministic} amounts to verifying
\CLAIM \ref{claim:pf:lemma:error:deterministic:1}, as accomplished next.
%
%
\begin{subproof}
{\CLAIM \ref{claim:pf:lemma:error:deterministic:1}}
The following work is nearly identical to the arguments in
\cite[proof of \LEMMA 4.1]{matsumoto2022binary}.
First, note that
\begin{gather*}
  \Vec{\uX}
  =
  \frac{\Threshold{\k}( \Vec{\vX} + \Vec{\wX} )}
       {\| \Threshold{\k}( \Vec{\vX} + \Vec{\wX} ) \|_{2}}
  =
  \frac{\ThresholdSet'{\Supp( \Threshold{\k}( \Vec{\vX} + \Vec{\wX} ) )}( \Vec{\vX} + \Vec{\wX} )}
       {\| \ThresholdSet'{\Supp( \Threshold{\k}( \Vec{\vX} + \Vec{\wX} ) )}
           ( \Vec{\vX} + \Vec{\wX} ) \|_{2}}
  =
  \frac{\ThresholdSet'{\Supp( \Vec{\uX} )}( \Vec{\vX} + \Vec{\wX} )}
       {\| \ThresholdSet'{\Supp( \Vec{\uX} )}( \Vec{\vX} + \Vec{\wX} ) \|_{2}}
.\end{gather*}
%
This will be useful later on.
Next, observe:
\begin{align*}
  &
  \Vec{\zX} - \Vec{\uX}
  \\
  &=
  \Vec{\zX}
  -
  \frac{\ThresholdSet'{\Supp( \Vec{\uX} )}( \Vec{\vX} + \Vec{\wX} )}
       {\| \ThresholdSet'{\Supp( \Vec{\uX} )}( \Vec{\vX} + \Vec{\wX} ) \|_{2}}
  \\
  &=
  \left(
    \Vec{\zX} - \ThresholdSet'{\Supp( \Vec{\zX} ) \cup \Supp( \Vec{\uX} )}( \Vec{\vX} + \Vec{\wX} )
  \right)
  +
  \left(
    \ThresholdSet'{\Supp( \Vec{\zX} ) \cup \Supp( \Vec{\uX} )}( \Vec{\vX} + \Vec{\wX} )
    -
    \ThresholdSet'{\Supp( \Vec{\uX} )}( \Vec{\vX} + \Vec{\wX} )
  \right)
  +
  \left(
    \ThresholdSet'{\Supp( \Vec{\uX} )}( \Vec{\vX} + \Vec{\wX} )
    -
    \frac{\ThresholdSet'{\Supp( \Vec{\uX} )}( \Vec{\vX} + \Vec{\wX} )}
         {\| \ThresholdSet'{\Supp( \Vec{\uX} )}( \Vec{\vX} + \Vec{\wX} ) \|_{2}}
  \right)
\end{align*}
%
Then,
\begin{align}
  \nonumber
  &
  \left\| \Vec{\zX} - \Vec{\uX} \right\|_{2}
  \\ \nonumber
  &=
  \left\|
    \left(
      \Vec{\zX} - \ThresholdSet'{\Supp( \Vec{\zX} ) \cup \Supp( \Vec{\uX} )}( \Vec{\vX} + \Vec{\wX} )
    \right)
    +
    \left(
      \ThresholdSet'{\Supp( \Vec{\zX} ) \cup \Supp( \Vec{\uX} )}( \Vec{\vX} + \Vec{\wX} )
      -
      \ThresholdSet'{\Supp( \Vec{\uX} )}( \Vec{\vX} + \Vec{\wX} )
    \right)
    +
    \left(
      \ThresholdSet'{\Supp( \Vec{\uX} )}( \Vec{\vX} + \Vec{\wX} )
      -
      \frac{\ThresholdSet'{\Supp( \Vec{\uX} )}( \Vec{\vX} + \Vec{\wX} )}
           {\| \ThresholdSet'{\Supp( \Vec{\uX} )}( \Vec{\vX} + \Vec{\wX} ) \|_{2}}
    \right)
  \right\|_{2}
  \\
  &\leq
  \left\|
    \Vec{\zX} - \ThresholdSet'{\Supp( \Vec{\zX} ) \cup \Supp( \Vec{\uX} )}( \Vec{\vX} + \Vec{\wX} )
  \right\|_{2}
  +
  \left\|
    \ThresholdSet'{\Supp( \Vec{\zX} ) \cup \Supp( \Vec{\uX} )}( \Vec{\vX} + \Vec{\wX} )
    -
    \ThresholdSet'{\Supp( \Vec{\uX} )}( \Vec{\vX} + \Vec{\wX} )
  \right\|_{2}
  +
  \left\|
    \ThresholdSet'{\Supp( \Vec{\uX} )}( \Vec{\vX} + \Vec{\wX} )
    -
    \frac{\ThresholdSet'{\Supp( \Vec{\uX} )}( \Vec{\vX} + \Vec{\wX} )}
         {\| \ThresholdSet'{\Supp( \Vec{\uX} )}( \Vec{\vX} + \Vec{\wX} ) \|_{2}}
  \right\|_{2}
\label{eqn:pf:lemma:error:deterministic:1}
\end{align}
where the last inequality applies the triangle inequality.
For clarity, denote the three terms in the last line by
\(  \alphaX[1], \alphaX[2], \alphaX[3] \geq 0  \),
where
\begin{gather*}
  \alphaX[1] \defeq
  \left\|
    \Vec{\zX} - \ThresholdSet'{\Supp( \Vec{\zX} ) \cup \Supp( \Vec{\uX} )}( \Vec{\vX} + \Vec{\wX} )
  \right\|_{2}
  ,\\
  \alphaX[2] \defeq
  \left\|
    \ThresholdSet'{\Supp( \Vec{\zX} ) \cup \Supp( \Vec{\uX} )}( \Vec{\vX} + \Vec{\wX} )
    -
    \ThresholdSet'{\Supp( \Vec{\uX} )}( \Vec{\vX} + \Vec{\wX} )
  \right\|_{2}
  ,\\
  \alphaX[3] \defeq
  \left\|
    \ThresholdSet'{\Supp( \Vec{\uX} )}( \Vec{\vX} + \Vec{\wX} )
    -
    \frac{\ThresholdSet'{\Supp( \Vec{\uX} )}( \Vec{\vX} + \Vec{\wX} )}
         {\| \ThresholdSet'{\Supp( \Vec{\uX} )}( \Vec{\vX} + \Vec{\wX} ) \|_{2}}
  \right\|_{2}
.\end{gather*}
%
\par 
%
The remainder of the proof is carried out in the following three steps.
\Enum[{\label{enum:pf:lemma:error:deterministic:a}}]{a}
First, \eqref{eqn:pf:lemma:error:deterministic:1} will be upper bounded by
\begin{align*}
  &
  \left\|
    \Vec{\zX} - \ThresholdSet'{\Supp( \Vec{\zX} ) \cup \Supp( \Vec{\uX} )}( \Vec{\vX} + \Vec{\wX} )
  \right\|_{2}
  +
  \left\|
    \ThresholdSet'{\Supp( \Vec{\zX} ) \cup \Supp( \Vec{\uX} )}( \Vec{\vX} + \Vec{\wX} )
    -
    \ThresholdSet'{\Supp( \Vec{\uX} )}( \Vec{\vX} + \Vec{\wX} )
  \right\|_{2}
  +
  \left\|
    \ThresholdSet'{\Supp( \Vec{\uX} )}( \Vec{\vX} + \Vec{\wX} )
    -
    \frac{\ThresholdSet'{\Supp( \Vec{\uX} )}( \Vec{\vX} + \Vec{\wX} )}
         {\| \ThresholdSet'{\Supp( \Vec{\uX} )}( \Vec{\vX} + \Vec{\wX} ) \|_{2}}
  \right\|_{2}
  \\
  &\leq
  2 \alphaX[1] + 2 \alphaX[2]
.\end{align*}
%
\Enum[{\label{enum:pf:lemma:error:deterministic:b}}]{b}
Then, simple arguments yield upper bounds on \(  \alphaX[1]  \) and \(  \alphaX[2]  \).
\Enum[{\label{enum:pf:lemma:error:deterministic:c}}]{c}
Lastly, combining the preceding work will provide the desired upper bound on
\(  \| \Vec{\zX} - \Vec{\uX} \|_{2}  \).
%
\paragraph{\STEP \ref{enum:pf:lemma:error:deterministic:a}.} 
%
First, the following derivation establishes the bound:
\(  \alphaX[3] \leq \alphaX[1] + \alphaX[2]  \).
\begin{align*}
  \alphaX[3]
  &=
  \left\|
    \ThresholdSet'{\Supp( \Vec{\uX} )}( \Vec{\vX} + \Vec{\wX} )
    -
    \frac{\ThresholdSet'{\Supp( \Vec{\uX} )}( \Vec{\vX} + \Vec{\wX} )}
         {\| \ThresholdSet'{\Supp( \Vec{\uX} )}( \Vec{\vX} + \Vec{\wX} ) \|_{2}}
  \right\|_{2}
  \\
  &\dCmt{by the definition of \(  \alphaX[3]  \)}
  \\
  &=
  \left\|
    \left( \| \ThresholdSet'{\Supp( \Vec{\uX} )}( \Vec{\vX} + \Vec{\wX} ) \|_{2} - 1 \right)
    \frac{\ThresholdSet'{\Supp( \Vec{\uX} )}( \Vec{\vX} + \Vec{\wX} )}
         {\| \ThresholdSet'{\Supp( \Vec{\uX} )}( \Vec{\vX} + \Vec{\wX} ) \|_{2}}
  \right\|_{2}
  \\
  &\dCmt{by distributivity}
  \\
  &=
  \left| \| \ThresholdSet'{\Supp( \Vec{\uX} )}( \Vec{\vX} + \Vec{\wX} ) \|_{2} - 1 \right|
  \left\|
    \frac{\ThresholdSet'{\Supp( \Vec{\uX} )}( \Vec{\vX} + \Vec{\wX} )}
         {\| \ThresholdSet'{\Supp( \Vec{\uX} )}( \Vec{\vX} + \Vec{\wX} ) \|_{2}}
  \right\|_{2}
  \\
  &\dCmt{by an axiom for metrics}
  \\
  &=
  \left| \| \ThresholdSet'{\Supp( \Vec{\uX} )}( \Vec{\vX} + \Vec{\wX} ) \|_{2} - 1 \right|
  \\
  &\dCmt{\(  \because \left\| \frac{\ThresholdSet'{\Supp( \Vec{\uX} )}( \Vec{\vX} + \Vec{\wX} )}{\| \ThresholdSet'{\Supp( \Vec{\uX} )}( \Vec{\vX} + \Vec{\wX} ) \|_{2}} \right\|_{2} = 1  \)}
  \\
  &=
  \left|
    \left\| \ThresholdSet'{\Supp( \Vec{\uX} )}( \Vec{\vX} + \Vec{\wX} ) \right\|_{2}
    -
    \| \Vec{\zX} \|_{2}
  \right|
  \\
  &\dCmt{\(  \because \| \Vec{\zX} \|_{2} = 1  \)}
  \\
  &\leq
  \left\|
    \ThresholdSet'{\Supp( \Vec{\uX} )}( \Vec{\vX} + \Vec{\wX} )
    -
    \Vec{\zX}
  \right\|_{2}
  \\
  &\dCmt{by the (reverse) triangle inequality (\see \REMARK \ref{remark:reverse-triangle-ineq}, below)}
  \\
  &=
  \left\|
    \left(
    \ThresholdSet'{\Supp( \Vec{\uX} )}( \Vec{\vX} + \Vec{\wX} )
    -
    \ThresholdSet'{\Supp( \Vec{\zX} ) \cup \Supp( \Vec{\uX} )}( \Vec{\vX} + \Vec{\wX} )
    \right)
    +
    \left(
    \ThresholdSet'{\Supp( \Vec{\zX} ) \cup \Supp( \Vec{\uX} )}( \Vec{\vX} + \Vec{\wX} )
    -
    \Vec{\zX}
    \right)
  \right\|_{2}
  \\
  &\dCmt{the inserted \(  \pm \ThresholdSet'{\Supp( \Vec{\zX} ) \cup \Supp( \Vec{\uX} )}( \Vec{\vX} + \Vec{\wX} )  \) terms cancel out each other}
  \\
  &\leq
  \left\|
    \ThresholdSet'{\Supp( \Vec{\uX} )}( \Vec{\vX} + \Vec{\wX} )
    -
    \ThresholdSet'{\Supp( \Vec{\zX} ) \cup \Supp( \Vec{\uX} )}( \Vec{\vX} + \Vec{\wX} )
  \right\|_{2}
  +
  \left\|
    \ThresholdSet'{\Supp( \Vec{\zX} ) \cup \Supp( \Vec{\uX} )}( \Vec{\vX} + \Vec{\wX} )
    -
    \Vec{\zX}
  \right\|_{2}
  \\
  &\dCmt{by the triangle inequality}
  \\
  &=
  \left\|
    \Vec{\zX}
    -
    \ThresholdSet'{\Supp( \Vec{\zX} ) \cup \Supp( \Vec{\uX} )}( \Vec{\vX} + \Vec{\wX} )
  \right\|_{2}
  +
  \left\|
    \ThresholdSet'{\Supp( \Vec{\zX} ) \cup \Supp( \Vec{\uX} )}( \Vec{\vX} + \Vec{\wX} )
    -
    \ThresholdSet'{\Supp( \Vec{\uX} )}( \Vec{\vX} + \Vec{\wX} )
  \right\|_{2}
  \\
  &\dCmt{by rearrangement}
  \\
  &=
  \alphaX[1] + \alphaX[2]
  \\
  &\dCmt{by the definitions of \(  \alphaX[1], \alphaX[2]  \)}
\end{align*}
%
\begin{remark}
\label{remark:reverse-triangle-ineq}
The reverse triangle inequality applied in the above derivation can be established from the
triangle inequality.
To formalize this, fix \(  \Vec{a}, \Vec{b} \in \R^{\n}  \) arbitrarily.
The reverse triangle inequality claims that
\(  \big| \| \Vec{a} \|_{2} - \| \Vec{b} \|_{2} \big| \leq \| \Vec{a} - \Vec{b} \|_{2}  \).
%
To verify this, note that the triangle inequality implies that
\(
  \| \Vec{a} \|_{2}
  =
  \| ( \Vec{a} - \Vec{b} ) + \Vec{b} \|_{2}
  \leq
  \| \Vec{a} - \Vec{b} \|_{2} + \| \Vec{b} \|_{2}
\).
%
Thus, by rearrangement,
\(  \| \Vec{a} \|_{2} - \| \Vec{b} \|_{2} \leq \| \Vec{a} - \Vec{b} \|_{2}  \).
%
Likewise, by swapping the roles of \(  \Vec{a}  \) and \(  \Vec{b}  \) in the above arguments,
it follows that
\(  \| \Vec{b} \|_{2} - \| \Vec{a} \|_{2} \leq \| \Vec{a} - \Vec{b} \|_{2}  \).
%
Combining the two bounds then yields the reverse triangle inequality:
\(  \big| \| \Vec{a} \|_{2} - \| \Vec{b} \|_{2} \big| \leq \| \Vec{a} - \Vec{b} \|_{2}  \).
\end{remark}
%
Now, we have that
\begin{align*}
  \left\| \Vec{\zX} - \Vec{\uX} \right\|_{2}
  \leq
  \alphaX[1] + \alphaX[2] + \alphaX[3]
  \leq
  \alphaX[1] + \alphaX[2] + \alphaX[1] + \alphaX[2]
  =
  2 \alphaX[1] + 2 \alphaX[2]
\end{align*}
which completes \STEP \ref{enum:pf:lemma:error:deterministic:a}.
%
\paragraph{\STEP \ref{enum:pf:lemma:error:deterministic:b}.} 
%
Write
\begin{gather*}
  \alphaX[1]' \defeq
  \left\|
    \Vec{\zX}
    -
    \ThresholdSet'{\Supp( \Vec{\zX} ) \cup \Supp( \Vec{\uX} ) \cup \Supp( \Vec{\vX} )}( \Vec{\vX} + \Vec{\wX} )
  \right\|_{2}
  ,\\
  \alphaX[2]' \defeq
  \left\|
    \ThresholdSet'{\Supp( \Vec{\uX} ) \setminus \Supp( \Vec{\zX} )}( \Vec{\vX} + \Vec{\wX} )
  \right\|_{2}
.\end{gather*}
%
The goal in this step will be to show that
\(  \alphaX[2] \leq \alphaX[1] \leq \alphaX[1]'  \).
%
To bound \(  \alphaX[2]  \), observe:
\begin{align*}
  \alphaX[2]
  &=
  \left\|
    \ThresholdSet'{\Supp( \Vec{\zX} ) \cup \Supp( \Vec{\uX} )}( \Vec{\vX} + \Vec{\wX} )
    -
    \ThresholdSet'{\Supp( \Vec{\uX} )}( \Vec{\vX} + \Vec{\wX} )
  \right\|_{2}
  \\
  &=
  \sqrt{
  \sum_{\jIx=1}^{\n}
  ( \ThresholdSet'{\Supp( \Vec{\zX} ) \cup \Supp( \Vec{\uX} )}( \Vec{\vX} + \Vec{\wX} ) )_{\jIx}^{2}
  -
  \sum_{\jIx=1}^{\n}
  ( \ThresholdSet'{\Supp( \Vec{\uX} )}( \Vec{\vX} + \Vec{\wX} ) )_{\jIx}^{2}
  }
  \\
  &\dCmt{by expanding out the definition of the \(  \lnorm{2}  \)-norm}
  \\
  &=
  \sqrt{
  \sum_{\jIx \in \Supp( \Vec{\zX} ) \cup \Supp( \Vec{\uX} )}
  ( \Vec{\vX} + \Vec{\wX} )_{\jIx}^{2}
  -
  \sum_{\jIx \in \Supp( \Vec{\uX} )}
  ( \Vec{\vX} + \Vec{\wX} )_{\jIx}^{2}
  }
  \\
  &\dCmt{due to the definition of the thresholding operation, \(  \ThresholdOp  \)}
  \\
  &=
  \sqrt{
  \sum_{\jIx \in \Supp( \Vec{\zX} ) \setminus \Supp( \Vec{\uX} )}
  ( \Vec{\vX} + \Vec{\wX} )_{\jIx}^{2}
  +
  \sum_{\jIx \in \Supp( \Vec{\uX} )}
  ( \Vec{\vX} + \Vec{\wX} )_{\jIx}^{2}
  -
  \sum_{\jIx \in \Supp( \Vec{\uX} )}
  ( \Vec{\vX} + \Vec{\wX} )_{\jIx}^{2}
  }
  \\
  &\dCmt{\(  \because ( \Supp( \Vec{\zX} ) \setminus \Supp( \Vec{\uX} ) ) \sqcup \Supp( \Vec{\uX} )  = \Supp( \Vec{\zX} ) \cup \Supp( \Vec{\uX} )  \) is a disjoint partition}
  \\
  &=
  \sqrt{
  \sum_{\jIx \in \Supp( \Vec{\zX} ) \setminus \Supp( \Vec{\uX} )}
  ( \Vec{\vX} + \Vec{\wX} )_{\jIx}^{2}
  }
  \\
  &\dCmt{the leftmost pair of summations in the preceding line cancel out each other}
  \\
  &=
  \sqrt{
  \sum_{\jIx=1}^{\n}
  \ThresholdSet'{\Supp( \Vec{\zX} ) \setminus \Supp( \Vec{\uX} )}( \Vec{\vX} + \Vec{\wX} )_{\jIx}^{2}
  }
  \\
  &\dCmt{due to the definition of the thresholding operation, \(  \ThresholdOp  \)}
  \\
  &=
  \left\|
    \ThresholdSet'{\Supp( \Vec{\zX} ) \setminus \Supp( \Vec{\uX} )}( \Vec{\vX} + \Vec{\wX} )
  \right\|_{2}
  \\
  &\dCmt{by condensing notation via the definition of the \(  \lnorm{2}  \)-norm}
  \\
  &\leq
  \left\|
    \ThresholdSet'{\Supp( \Vec{\uX} ) \setminus \Supp( \Vec{\zX} )}( \Vec{\vX} + \Vec{\wX} )
  \right\|_{2}
  \\
  &\dCmt{\see \REMARK \ref{remark:top-k-ineq} below}
  \\
  &=
  \alphaX[2]'
  \\
  &\dCmt{by the definition of \(  \alphaX[2]'  \)}
\end{align*}
%
\begin{remark}
\label{remark:top-k-ineq}
The above derivation uses, where noted, the inequality:
\begin{gather*}
  \left\|
    \ThresholdSet'{\Supp( \Vec{\zX} ) \setminus \Supp( \Vec{\uX} )}( \Vec{\vX} + \Vec{\wX} )
  \right\|_{2}
  \leq
  \left\|
    \ThresholdSet'{\Supp( \Vec{\uX} ) \setminus \Supp( \Vec{\zX} )}( \Vec{\vX} + \Vec{\wX} )
  \right\|_{2}
.\end{gather*}
%
This inequality is verified as follows.
Recall that
\begin{gather*}
  \Vec{\uX}
  =
  \frac{\Threshold{\k}( \Vec{\vX} + \Vec{\wX} )}
       {\| \Threshold{\k}( \Vec{\vX} + \Vec{\wX} ) \|_{2}}
  =
  \frac{\ThresholdSet'{\Supp( \Threshold{\k}( \Vec{\vX} + \Vec{\wX} ) )}( \Vec{\vX} + \Vec{\wX} )}
       {\| \ThresholdSet'{\Supp( \Threshold{\k}( \Vec{\vX} + \Vec{\wX} ) )}
           ( \Vec{\vX} + \Vec{\wX} ) \|_{2}}
,\end{gather*}
and hence,
\(  \Supp( \Vec{\uX} ) = \Supp( \Threshold{\k}( \Vec{\vX} + \Vec{\wX} ) )  \).
%
For any \(  \jIx \notin \Supp( \Vec{\uX} )  \),
the definition of the \topkhardthresholding operation, \(  \Threshold{\k}  \), enforces:
\(  | \Vec*{\vX}_{\jIx} + \Vec*{\wX}_{\jIx} | \leq \min_{\jIx' \in \Supp( \Vec{\uX} )} | \Vec*{\vX}_{\jIx'} + \Vec*{\wX}_{\jIx'} |  \).
%
Additionally,
\(  \| \Vec{\uX} \|_{0} = \k  \)
since
\(  \| \Vec{\vX} + \Vec{\wX} \|_{0} \geq \k  \)
ensures the top-\(  \k  \) entries in \(  \Vec{\vX} + \Vec{\wX}  \) are all nonzero.
On the other hand, recall that \(  \| \Vec{\zX} \|_{0} \leq \k  \).
As a result,
\(  | \Supp( \Vec{\zX} ) \setminus \Supp( \Vec{\uX} ) | \leq | \Supp( \Vec{\uX} ) \setminus \Supp( \Vec{\zX} ) |  \),
and at the same time, for any
\(  \jIx \in \Supp( \Vec{\zX} ) \setminus \Supp( \Vec{\uX} )  \) and
\(  \jIx' \in \Supp( \Vec{\uX} ) \setminus \Supp( \Vec{\zX} )  \),
\(  | \Vec*{\vX}_{\jIx} + \Vec*{\wX}_{\jIx} | \leq | \Vec*{\vX}_{\jIx'} + \Vec*{\wX}_{\jIx'} |  \).
%
Write
\(  \Ell \defeq | \Supp( \Vec{\zX} ) \setminus \Supp( \Vec{\uX} ) |  \) and
\(  \Ell' \defeq | \Supp( \Vec{\uX} ) \setminus \Supp( \Vec{\zX} ) |  \),
and let
\(  \{ \jIx_{1}, \dots, \jIx_{\Ell} \} = \Supp( \Vec{\zX} ) \setminus \Supp( \Vec{\uX} )  \) and
\(  \{ \jIx'_{1}, \dots, \jIx'_{\Ell'} \} = \Supp( \Vec{\uX} ) \setminus \Supp( \Vec{\zX} )  \).
%
Note that by the above discussion, \(  \Ell \leq \Ell'  \), and
\(  | \Vec*{\vX}_{\jIx_{s}} + \Vec*{\wX}_{\jIx_{s}} | \leq | \Vec*{\vX}_{\jIx'_{s}} + \Vec*{\wX}_{\jIx'_{s}} |  \)
for each \(  s \in [\Ell]  \).
Taken together, the desired inequality follows:
\begin{align*}
  \left\|
    \ThresholdSet'{\Supp( \Vec{\zX} ) \setminus \Supp( \Vec{\uX} )}( \Vec{\vX} + \Vec{\wX} )
  \right\|_{2}
  &=
  \sum_{\jIx=1}^{\n}
  \ThresholdSet'{\Supp( \Vec{\zX} ) \setminus \Supp( \Vec{\uX} )}( \Vec{\vX} + \Vec{\wX} )_{\jIx}^{2}
  \\
  &\dCmt{by expanding out the definition of the \(  \lnorm{2}  \)-norm}
  \\
  &=
  \sum_{\jIx \in \Supp( \Vec{\zX} ) \setminus \Supp( \Vec{\uX} )}
  ( \Vec*{\vX}_{\jIx} + \Vec*{\wX}_{\jIx} )^{2}
  \\
  &\dCmt{due to the definition of the hard thresholding operation, \(  \ThresholdOp  \)}
  \\
  &=
  \sum_{s=1}^{\Ell}
  ( \Vec*{\vX}_{\jIx_{s}} + \Vec*{\wX}_{\jIx_{s}} )^{2}
  \\
  &\dCmt{by reindexing with \(  \{ \jIx_{1}, \dots, \jIx_{\Ell} \} = \Supp( \Vec{\zX} ) \setminus \Supp( \Vec{\uX} )  \)}
  \\
  &\leq
  \sum_{s=1}^{\Ell}
  ( \Vec*{\vX}_{\jIx'_{s}} + \Vec*{\wX}_{\jIx'_{s}} )^{2}
  \\
  &\dCmt{by the earlier observation that \(  | \Vec*{\vX}_{\jIx_{s}} + \Vec*{\wX}_{\jIx_{s}} | \leq | \Vec*{\vX}_{\jIx'_{s}} + \Vec*{\wX}_{\jIx'_{s}} |  \)}
  \\
  &\leq
  \sum_{s=1}^{\Ell'}
  ( \Vec*{\vX}_{\jIx'_{s}} + \Vec*{\wX}_{\jIx'_{s}} )^{2}
  \\
  &\dCmt{since all summands are nonnegative, and \(  \Ell \leq \Ell'  \)}
  \\
  &=
  \sum_{\jIx \in \Supp( \Vec{\uX} ) \setminus \Supp( \Vec{\zX} )}
  ( \Vec*{\vX}_{\jIx} + \Vec*{\wX}_{\jIx} )^{2}
  \\
  &\dCmt{by reindexing with \(  \{ \jIx'_{1}, \dots, \jIx'_{\Ell'} \} = \Supp( \Vec{\uX} ) \setminus \Supp( \Vec{\zX} )  \)}
  \\
  &=
  \sum_{\jIx=1}^{\n}
  \ThresholdSet'{\Supp( \Vec{\uX} ) \setminus \Supp( \Vec{\zX} )}( \Vec{\vX} + \Vec{\wX} )_{\jIx}^{2}
  \\
  &\dCmt{due to the definition of the hard thresholding operation, \(  \ThresholdOp  \)}
  \\
  &=
  \left\|
    \ThresholdSet'{\Supp( \Vec{\uX} ) \setminus \Supp( \Vec{\zX} )}( \Vec{\vX} + \Vec{\wX} )
  \right\|_{2}
  \\
  &\dCmt{by condensing notation via the definition of the \(  \lnorm{2}  \)-norm}
\end{align*}
%
Thus, the desired inequality,
\(
  \|
    \ThresholdSet'{\Supp( \Vec{\zX} ) \setminus \Supp( \Vec{\uX} )}( \Vec{\vX} + \Vec{\wX} )
  \|_{2}
  \leq
  \|
    \ThresholdSet'{\Supp( \Vec{\uX} ) \setminus \Supp( \Vec{\zX} )}( \Vec{\vX} + \Vec{\wX} )
  \|_{2}
\),
has been verified.
\end{remark}
%
The above work has shown that
\(  \alphaX[2] \leq \alphaX[2]'  \).
%
%
Next, to bound \(  \alphaX[2]'  \), observe:
\begin{align*}
  \alphaX[1]^{2}
  &=
  \left\|
    \Vec{\zX} - \ThresholdSet'{\Supp( \Vec{\zX} ) \cup \Supp( \Vec{\uX} )}( \Vec{\vX} + \Vec{\wX} )
  \right\|_{2}^{2}
  \\
  &\dCmt{by the definition of \(  \alphaX[1]  \)}
  \\
  &=
  \sum_{\jIx=1}^{\n}
  \left(
    \Vec*{\zX}_{\jIx}
    -
    \ThresholdSet'{\Supp( \Vec{\zX} ) \cup \Supp( \Vec{\uX} )}( \Vec{\vX} + \Vec{\wX} )_{\jIx}
  \right)^{2}
  \\
  &\dCmt{by expanding out the definition of the \(  \lnorm{2}  \)-norm}
  \\
  &=
  \sum_{\jIx \in \Supp( \Vec{\zX} ) \cup \Supp( \Vec{\uX} )}
  \left(
    \Vec*{\zX}_{\jIx}
    -
    \Vec*{\vX}_{\jIx}
    -
    \Vec*{\wX}_{\jIx}
  \right)^{2}
  \\
  &\dCmt{since \(  \Supp( \Vec{\zX} - \ThresholdSet'{\Supp( \Vec{\zX} ) \cup \Supp( \Vec{\uX} )}( \Vec{\vX} + \Vec{\wX} ) ) \subseteq \Supp( \Vec{\zX} ) \cup \Supp( \Vec{\uX} )  \)}
  \\
  &=
  \sum_{\jIx \in \Supp( \Vec{\uX} ) \setminus \Supp( \Vec{\zX} )}
  \left(
    \Vec*{\zX}_{\jIx}
    -
    \Vec*{\vX}_{\jIx}
    -
    \Vec*{\wX}_{\jIx}
  \right)^{2}
  +
  \sum_{\jIx \in \Supp( \Vec{\zX} )}
  \left(
    \Vec*{\zX}_{\jIx}
    -
    \Vec*{\vX}_{\jIx}
    -
    \Vec*{\wX}_{\jIx}
  \right)^{2}
  \\
  &\dCmt{\(  \because ( \Supp( \Vec{\uX} ) \setminus \Supp( \Vec{\zX} ) ) \sqcup \Supp( \Vec{\zX} ) = \Supp( \Vec{\zX} ) \cup \Supp( \Vec{\uX} )  \) is a disjoint partition}
  \\
  &=
  \sum_{\jIx \in \Supp( \Vec{\uX} ) \setminus \Supp( \Vec{\zX} )}
  \left(
    -
    \Vec*{\vX}_{\jIx}
    -
    \Vec*{\wX}_{\jIx}
  \right)^{2}
  +
  \sum_{\jIx \in \Supp( \Vec{\zX} )}
  \left(
    \Vec*{\zX}_{\jIx}
    -
    \Vec*{\vX}_{\jIx}
    -
    \Vec*{\wX}_{\jIx}
  \right)^{2}
  \\
  &\dCmt{\(  \because \Vec*{\zX}_{\jIx} = 0  \) for any \(  \jIx \in [\n] \setminus \Supp( \Vec{\zX} )  \)}
  \\
  &=
  \sum_{\jIx \in \Supp( \Vec{\uX} ) \setminus \Supp( \Vec{\zX} )}
  \left(
    \Vec*{\vX}_{\jIx}
    +
    \Vec*{\wX}_{\jIx}
  \right)^{2}
  +
  \sum_{\jIx \in \Supp( \Vec{\zX} )}
  \left(
    \Vec*{\zX}_{\jIx}
    -
    \Vec*{\vX}_{\jIx}
    -
    \Vec*{\wX}_{\jIx}
  \right)^{2}
  \\
  &\dCmt{by squaring the \(  -1  \) factor}
  \\
  &=
  \sum_{\jIx=1}^{\n}
  \ThresholdSet'{\Supp( \Vec{\uX} ) \setminus \Supp( \Vec{\zX} )}( \Vec{\vX} + \Vec{\wX} )_{\jIx}^{2}
  +
  \sum_{\jIx=1}^{\n}
  \ThresholdSet'{\Supp( \Vec{\zX} )}( \Vec{\zX} - ( \Vec{\vX} + \Vec{\wX} ) )_{\jIx}^{2}
  \\
  &\dCmt{due to the definition of the hard thresholding operation, \(  \ThresholdOp  \)}
  \\
  &=
  \left\|
    \ThresholdSet'{\Supp( \Vec{\uX} ) \setminus \Supp( \Vec{\zX} )}( \Vec{\vX} + \Vec{\wX} )
  \right\|_{2}^{2}
  +
  \left\|
    \ThresholdSet'{\Supp( \Vec{\zX} )}( \Vec{\zX} - ( \Vec{\vX} + \Vec{\wX} ) )
  \right\|_{2}^{2}
  \\
  &\dCmt{by condensing notation via the definition of the \(  \lnorm{2}  \)-norm}
\end{align*}
%
In short,
\begin{align*}
  \alphaX[1]^{2}
  =
  \left\|
    \Vec{\zX} - \ThresholdSet'{\Supp( \Vec{\zX} ) \cup \Supp( \Vec{\uX} )}( \Vec{\vX} + \Vec{\wX} )
  \right\|_{2}^{2}
  =
  \left\|
    \ThresholdSet'{\Supp( \Vec{\uX} ) \setminus \Supp( \Vec{\zX} )}( \Vec{\vX} + \Vec{\wX} )
  \right\|_{2}^{2}
  +
  \left\|
    \ThresholdSet'{\Supp( \Vec{\zX} )}( \Vec{\zX} - ( \Vec{\vX} + \Vec{\wX} ) )
  \right\|_{2}^{2}
.\end{align*}
%
Rearranging the terms obtains:
\begin{align*}
  \left\|
    \ThresholdSet'{\Supp( \Vec{\uX} ) \setminus \Supp( \Vec{\zX} )}( \Vec{\vX} + \Vec{\wX} )
  \right\|_{2}^{2}
  &=
  \left\|
    \Vec{\zX} - \ThresholdSet'{\Supp( \Vec{\zX} ) \cup \Supp( \Vec{\uX} )}( \Vec{\vX} + \Vec{\wX} )
  \right\|_{2}^{2}
  -
  \left\|
    \ThresholdSet'{\Supp( \Vec{\zX} )}( \Vec{\zX} - ( \Vec{\vX} + \Vec{\wX} ) )
  \right\|_{2}^{2}
  \\
  &\leq
  \left\|
    \Vec{\zX} - \ThresholdSet'{\Supp( \Vec{\zX} ) \cup \Supp( \Vec{\uX} )}( \Vec{\vX} + \Vec{\wX} )
  \right\|_{2}^{2}
.\end{align*}
%
Hence, after taking a square root:
\begin{align*}
  \alphaX[2]'
  =
  \left\|
    \ThresholdSet'{\Supp( \Vec{\uX} ) \setminus \Supp( \Vec{\zX} )}( \Vec{\vX} + \Vec{\wX} )
  \right\|_{2}
  \leq
  \left\|
    \Vec{\zX} - \ThresholdSet'{\Supp( \Vec{\zX} ) \cup \Supp( \Vec{\uX} )}( \Vec{\vX} + \Vec{\wX} )
  \right\|_{2}
  =
  \alphaX[1]
\end{align*}
%
The final task for \STEP \ref{enum:pf:lemma:error:deterministic:b} is bounding \(  \alphaX[1]  \).
For this purpose, observe the following.
(Note that the comments throughout the derivation below take
\(  J \subseteq [\n]  \), \(  \Vec{a}, \Vec{b} \in \R^{\n}  \).)
\begin{align*}
  &
  \left\|
    \Vec{\zX} - \ThresholdSet'{\Supp( \Vec{\zX} ) \cup \Supp( \Vec{\uX} ) \cup \Supp( \Vec{\vX} )}
    ( \Vec{\vX} + \Vec{\wX} )
  \right\|_{2}^{2}
  \\
  &=
  \left\|
    \ThresholdSet'{\Supp( \Vec{\zX} ) \cup \Supp( \Vec{\uX} ) \cup \Supp( \Vec{\vX} )}( \Vec{\zX} )
    -
    \ThresholdSet'{\Supp( \Vec{\zX} ) \cup \Supp( \Vec{\uX} ) \cup \Supp( \Vec{\vX} )}
    ( \Vec{\vX} + \Vec{\wX} )
  \right\|_{2}^{2}
  \\
  &\dCmt{\(  \because \ThresholdSet'{J}( \Vec{a} ) = \Vec{a}  \) if \(  J \supseteq \Supp( \Vec{a} )  \)}
  \\
  &=
  \left\|
    \ThresholdSet'{\Supp( \Vec{\zX} ) \cup \Supp( \Vec{\uX} ) \cup \Supp( \Vec{\vX} )}
    ( \Vec{\zX} - \Vec{\vX} - \Vec{\wX} )
  \right\|_{2}^{2}
  \\
  &\dCmt{\(  \because \ThresholdSet'{J}( \Vec{a} ) + \ThresholdSet'{J}( \Vec{b} ) = \ThresholdSet'{J}( \Vec{a} + \Vec{b} )  \)}
  \\
  &=
  \left\|
    \ThresholdSet'{\Supp( \Vec{\zX} ) \cup \Supp( \Vec{\uX} )}
    ( \Vec{\zX} - \Vec{\vX} - \Vec{\wX} )
    +
    \ThresholdSet'{\Supp( \Vec{\vX} ) \setminus ( \Supp( \Vec{\zX} ) \cup \Supp( \Vec{\uX} ) )}
    ( \Vec{\zX} - \Vec{\vX} - \Vec{\wX} )
  \right\|_{2}^{2}
  \\
  &\dCmt{\(  ( \Supp( \Vec{\zX} ) \cup \Supp( \Vec{\uX} ) ) \sqcup ( \Supp( \Vec{\vX} ) \setminus ( \Supp( \Vec{\zX} ) \cup \Supp( \Vec{\uX} ) ) ) = \Supp( \Vec{\zX} ) \cup \Supp( \Vec{\uX} ) \cup \Supp( \Vec{\vX} )  \)}
  \\
  &\dCmtIndent \text{is a (disjoint) partition}
  \\
  &=
  \left\|
    \ThresholdSet'{\Supp( \Vec{\zX} ) \cup \Supp( \Vec{\uX} )}
    ( \Vec{\zX} - \Vec{\vX} - \Vec{\wX} )
  \right\|_{2}^{2}
  +
  \left\|
    \ThresholdSet'{\Supp( \Vec{\vX} ) \setminus ( \Supp( \Vec{\zX} ) \cup \Supp( \Vec{\uX} ) )}
    ( \Vec{\zX} - \Vec{\vX} - \Vec{\wX} )
  \right\|_{2}^{2}
  \\
  &\dCmt{by the Pythagorean theorem}
  \\
  &\dCmt{note that \(  \ThresholdSet'{\Supp( \Vec{\zX} ) \cup \Supp( \Vec{\uX} )}( \Vec{\zX} - \Vec{\vX} - \Vec{\wX} )  \) and \(  \ThresholdSet'{\Supp( \Vec{\vX} ) \setminus ( \Supp( \Vec{\zX} ) \cup \Supp( \Vec{\uX} ) )}( \Vec{\zX} - \Vec{\vX} - \Vec{\wX} )  \)}
  \\
  &\dCmtIndent \text{are orthogonal since their support sets are disjoint}
  \\
  &\geq
  \left\|
    \ThresholdSet'{\Supp( \Vec{\zX} ) \cup \Supp( \Vec{\uX} )}
    ( \Vec{\zX} - \Vec{\vX} - \Vec{\wX} )
  \right\|_{2}^{2}
  \\
  &\dCmt{since both terms in the preceding line are nonnegative,} \\
  &\dCmtIndent \text{deleting one cannot increase the value of the expression}
  \\
  &=
  \left\|
    \ThresholdSet'{\Supp( \Vec{\zX} ) \cup \Supp( \Vec{\uX} )}
    ( \Vec{\zX} )
    -
    \ThresholdSet'{\Supp( \Vec{\zX} ) \cup \Supp( \Vec{\uX} )}
    ( \Vec{\vX} + \Vec{\wX} )
  \right\|_{2}^{2}
  \\
  &\dCmt{\(  \because \ThresholdSet'{J}( \Vec{a} ) + \ThresholdSet'{J}( \Vec{b} ) = \ThresholdSet'{J}( \Vec{a} + \Vec{b} )  \)}
  \\
  &=
  \left\|
    \Vec{\zX}
    -
    \ThresholdSet'{\Supp( \Vec{\zX} ) \cup \Supp( \Vec{\uX} )}
    ( \Vec{\vX} + \Vec{\wX} )
  \right\|_{2}^{2}
  \\
  &\dCmt{\(  \because \ThresholdSet'{J}( \Vec{a} ) = \Vec{a}  \) if \(  J \supseteq \Supp(  \Vec{a} )  \)}
  \\
  &=
  \alphaX[1]^{2}
  \\
  &\dCmt{by the definition of \(  \alphaX[1]  \)}
\end{align*}
%
Thus,
\begin{align*}
  \alphaX[1]
  =
  \left\|
    \Vec{\zX}
    -
    \ThresholdSet'{\Supp( \Vec{\zX} ) \cup \Supp( \Vec{\uX} )}
    ( \Vec{\vX} + \Vec{\wX} )
  \right\|_{2}
  \leq
  \left\|
    \Vec{\zX} - \ThresholdSet'{\Supp( \Vec{\zX} ) \cup \Supp( \Vec{\uX} ) \cup \Supp( \Vec{\vX} )}
    ( \Vec{\vX} + \Vec{\wX} )
  \right\|_{2}
  =
  \alphaX[1]'
\end{align*}
as claimed.
To summarize, this step has shown:
\begin{gather*}
  \alphaX[2] \leq \alphaX[2]' \leq \alphaX[1] \leq \alphaX[1]'
,\end{gather*}
%
\paragraph{\STEP \ref{enum:pf:lemma:error:deterministic:c}.} 
%
By combining the arguments of \STEPS \ref{enum:pf:lemma:error:deterministic:a} and
\ref{enum:pf:lemma:error:deterministic:b},
\EQUATION \eqref{eqn:claim:pf:lemma:error:deterministic:1:2} follows:
\begin{align*}
  \| \Vec{\zX} - \Vec{\uX} \|_{2}
  \leq
  2 \alphaX[1] + 2 \alphaX[2]
  \leq
  4 \alphaX[1]
  \leq
  4 \alphaX[1]'
  &\leq
  4
  \left\|
    \Vec{\zX}
    -
    \ThresholdSet'{\Supp( \Vec{\zX} ) \cup \Supp( \Vec{\uX} ) \cup \Supp( \Vec{\vX} )}
    ( \Vec{\vX} + \Vec{\wX} )
  \right\|_{2}
  \\
  &=
  4
  \left\|
    ( \Vec{\zX} - \Vec{\vX} )
    -
    \ThresholdSet'{\Supp( \Vec{\zX} ) \cup \Supp( \Vec{\uX} ) \cup \Supp( \Vec{\vX} )}
    ( \Vec{\wX} )
  \right\|_{2}
.\end{align*}
%
This completes the proof of \CLAIM \ref{claim:pf:lemma:error:deterministic:1}.
\end{subproof}

By the discussion at the beginning of this proof, due to the proof of
\CLAIM \ref{claim:pf:lemma:error:deterministic:1},
\LEMMA \ref{lemma:error:deterministic} also holds.
\end{proof}
\subsection{Proof of \LEMMA \ref{lemma:error:recurrence}} 
\label{section:|>pf-deterministic|>pf-recurrence}         

Before the proof of \LEMMA \ref{lemma:error:recurrence} is laid out,
the following fact is stated to facilitate this.
The proof of this fact can be found in \cite{matsumoto2022binary}.

\begin{fact}[{\cite[\FACT 4.1]{matsumoto2022binary}}]
\label{fact:recurrence}
Let
\(  \ux, \vx, \wx, \wOx \in \R_{+}  \),
where
\(  \ux = \frac{1}{2} ( 1 + \sqrt{1+4\wx} )  \) and
\(  \ux \in [1,\sqrt{\frac{2}{\vx}}]  \).
%
Let
\(  \fx{1}, \fx{2} : \Z_{\geq 0} \to \R  \)
be the functions given by
\begin{subequations}
\label{eqn:fact:recurrence:1}
\begin{gather}
  \label{eqn:fact:recurrence:1:1}
  \fx{1}( 0 ) = 2
  ,\\ \label{eqn:fact:recurrence:1:2}
  \fx{1}( t ) = \vx \wx + \sqrt{\vx \fx{1}( t-1 )}
  ,\quad t \in \Z_{+}
,\end{gather}
\end{subequations}
\begin{gather}
\label{eqn:fact:recurrence:2}
  \fx{2}( t ) = 2^{2^{-t}} ( \ux^2 \vx )^{1-2^{-t}}
  ,\quad t \in \Z_{\geq 0}
.\end{gather}
%
The functions, \(  \fx{1}  \) and \(  \fx{2}  \), are strictly decreasing and satisfy
\begin{gather}
  \label{eqn:fact:recurrence:3}
  \fx{1}( t ) \leq \fx{2}( t )
  ,\quad \forall t \in \Z_{\geq 0}
  ,\\ \label{eqn:fact:recurrence:4}
  \lim_{t \to \infty} \fx{2}( t ) = \lim_{t \to \infty} \fx{1}( t ) = \ux^{2} \vx
.\end{gather}
\end{fact}

\begin{proof}
{\LEMMA \ref{lemma:error:recurrence}}
The lemma's results follow from an argument nearly identical to the proofs of
\cite[\LEMMAS 4.2 and 4.3]{matsumoto2022binary} with just a couple changes in constants.
The (combined) proofs are reproduced below with the appropriate adjustments to constants.
The results are derived simply via \FACT \ref{fact:recurrence}.
Recall the definition of the function,
\(  \Err{} : \Z_{\geq 0} \to \R  \),
by the recurrence relation:
\begin{subequations}
\label{eqn:pf:lemma:error:recurrence:1}
\begin{gather}
\label{eqn:pf:lemma:error:recurrence:1:1}
  \Err{0} = 2
  \\ \label{eqn:pf:lemma:error:recurrence:1:2}
  \Err{\Iter}
  = 4\cC[1] \sqrt{\frac{\gammaX}{\cC} \Err{\Iter-1}}
  + \frac{4 \cX{2} \gammaX}{\cC}
  ,\qquad
  \Iter \in \Z_{+}
.\end{gather}
\end{subequations}
%
%
The first task is writing \EQUATIONS \eqref{eqn:pf:lemma:error:recurrence:1}
in the form of \EQUATIONS \eqref{eqn:fact:recurrence:1}.
When \(  \Iter = 0  \), then trivially, \EQUATION \eqref{eqn:pf:lemma:error:recurrence:1:1}
matches \EQUATION \eqref{eqn:fact:recurrence:1:1}, whereas for
\(  \Iter \in \Z_{+}  \), \EQUATION \eqref{eqn:pf:lemma:error:recurrence:1:2}
can match the form of \EQUATION \eqref{eqn:fact:recurrence:1:2} by simply writing:
\begin{align*}
  \Err{\Iter}
  &=
  4\cC[1] \sqrt{\frac{\gammaX}{\cC} \Err{\Iter-1}}
  +
  \frac{4 \cX{2} \gammaX}{\cC}
  \\
  &=
  \sqrt{\frac{16 \cX{1}^{2} \gammaX}{\cC} \Err{\Iter-1}}
  +
  \frac{4 \cX{2} \gammaX}{\cC}
  \\
  &=
  \sqrt{\frac{16 \cX{1}^{2} \gammaX}{\cC} \Err{\Iter-1}}
  +
  \frac{4 \cX{2}}{16 \cX{1}^{2}} \frac{16 \cX{1}^{2} \gammaX}{\cC}
  \\
  &=
  \sqrt{\frac{16 \cX{1}^{2} \gammaX}{\cC} \Err{\Iter-1}}
  +
  \frac{\cX{2}}{4 \cX{1}^{2}} \frac{16 \cX{1}^{2} \gammaX}{\cC}
  \\
  &=
  \frac{16 \cX{1}^{2} \gammaX}{\cC} \frac{\cX{2}}{4 \cX{1}^{2}}
  +
  \sqrt{\frac{16 \cX{1}^{2} \gammaX}{\cC} \Err{\Iter-1}}
  \\
  &=
  \vx \wx + \sqrt{\vx \fx{1}( \Iter-1 )}
\end{align*}
where in the last line,
\begin{subequations}
\label{eqn:pf:lemma:error:recurrence:2}
\begin{gather}
  \label{eqn:pf:lemma:error:recurrence:2:1}
  \fx{1} = \Err{}
  ,\\ \label{eqn:pf:lemma:error:recurrence:2:2}
  \vx = \frac{16 \cX{1}^{2} \gammaX}{\cC}
  ,\\ \label{eqn:pf:lemma:error:recurrence:2:3}
  \wx = \frac{\cX{2}}{4 \cX{1}^{2}}
.\end{gather}
\end{subequations}
%
Now we have that \EQUATION \eqref{eqn:pf:lemma:error:recurrence:1:2},
\begin{gather*}
  \Err{\Iter}
  = 4\cC[1] \sqrt{\frac{\gammaX}{\cC} \Err{\Iter-1}}
  + \frac{4 \cX{2} \gammaX}{\cC}
,\end{gather*}
is equivalent to
\begin{gather*}
  \fx{1}( \Iter ) = \vx \wx + \sqrt{\vx \fx{1}( \Iter-1 )}
,\end{gather*}
the latter of which is precisely the form of \EQUATION \eqref{eqn:fact:recurrence:1:2}.
%
\par 
%
Before \FACT \ref{fact:recurrence} can be applied, it is necessary to verify
that the fact's conditions are satisfied when the parameters \(  \vx, \wx  \)
are chosen as in \eqref{eqn:pf:lemma:error:recurrence:2}.
Specifically, writing
\(    \ux = \frac{1}{2} ( 1 + \sqrt{1 + 4\wx} )  \),
\FACT \ref{fact:recurrence} requires that
\(  1 \leq \ux \leq \sqrt{\frac{2}{\vx}}  \).
%
Clearly,
\(  \ux \geq 1  \)
since
\(  \frac{1}{2} ( 1+\sqrt{1+z} ) \geq \frac{1}{2} ( 1+1 ) = 1  \)
for \(  z \geq 0  \).
Towards verifying the other side of the bound,
\(  \ux \leq \sqrt{\frac{2}{\vx}}  \),
expand out \(  \ux  \) and \(  \frac{1}{\sqrt{\vx}}  \) as follows.
For \(  \ux  \), observe:
\begin{align*}
  \ux
  &= \frac{1}{2} \left( 1 + \sqrt{1 + 4\wx} \right)
  \\
  &= \frac{1}{2}
    \left(
      1 + \sqrt{1 + 4 \cdot \frac{\cX{2}}{4 \cX{1}^{2}}}
    \right)
  \\
  &= \frac{1}{2}
    \left(
      1 + \sqrt{1 + \frac{\cX{2}}{\cX{1}^{2}}}
    \right)
  \\
  &= \frac{1}{2}
    \left(
      1 + \frac{1}{\cX{1}} \sqrt{\cX{1}^{2} + \cX{2}}
    \right)
  \\
  &= \frac{1}{2}
    \left(
      \frac{1}{\cX{1}} \cX{1} + \frac{1}{\cX{1}} \sqrt{\cX{1}^{2} + \cX{2}}
    \right)
  \\
  &= \frac{1}{2\cX{1}}
    \left(
      \cX{1} + \sqrt{\cX{1}^{2} + \cX{2}}
    \right)
  \\
  &= \frac{\cX{1} + \sqrt{\cX{1}^{2} + \cX{2}}}{2\cX{1}}
\end{align*}
%
Next, \(  \frac{1}{\sqrt{\vx}}  \) is rewritten as:
\begin{align*}
  \frac{1}{\sqrt{\vx}}
  &=
  \sqrt{\frac{\cX{}}{16 \cX{1}^{2} \gammaX}}
  \\
  &=
  \sqrt{\frac{\cX{}}{4 \cX{1}^{2} \cdot 4 \gammaX}}
  \\
  &=
  \frac{1}{2 \cX{1}} \sqrt{\frac{\cX{}}{4 \gammaX}}
  \\
  &=
  \frac{1}{\sqrt{\gammaX}}
  \cdot
  \frac{\sqrt{\cX{}/4}}{2 \cC[1]}
\end{align*}
%
Taking
\begin{gather*}
  \cX{}
  = 4 \left( \cX{1} + \sqrt{\cX{1}^{2} + \cX{2}} \right)^{2}
,\end{gather*}
it follows that
\(  \ux \leq \sqrt{\frac{2}{\vx}}  \),
as required, since:
\begin{align*}
  \sqrt{\frac{2}{\vx}}
  &\geq
  \sqrt{\frac{2}{\gammaX}} \cdot \frac{\sqrt{\cX{}/4}}{2 \cC[1]}
  \\
  &=
  \sqrt{\frac{2}{\gammaX}} \cdot \frac{\sqrt{\frac{1}{4} \cdot 4 \left( \cX{1} + \sqrt{\cX{1}^{2} + \cX{2}} \right)^{2}}}{2 \cX{1}}
  \\
  &=
  \sqrt{\frac{2}{\gammaX}} \cdot \frac{\sqrt{\left( \cX{1} + \sqrt{\cX{1}^{2} + \cX{2}} \right)^{2}}}{2 \cX{1}}
  \\
  &=
  \sqrt{\frac{2}{\gammaX}} \cdot \frac{\cX{1} + \sqrt{\cX{1}^{2} + \cX{2}}}{2\cC[1]}
  \\
  &=
  \sqrt{\frac{2}{\gammaX}} \ux
  \\
  &\geq
  \ux
\end{align*}
%
Hence, the fact applies since
\(  1 \leq \ux \leq \sqrt{\frac{2}{\vx}}  \).
%
\par 
%
The lemma's results can now be obtained via \FACT \ref{fact:recurrence}.
Note that
\(  \sqrt{\frac{2}{\vx}} \geq \sqrt{\frac{2}{\gammaX}} \cdot \frac{\sqrt{\cX{}/4}}{2 \cX{1}}  \)
implies
\(  \sqrt{\frac{\gammaX}{\vx}} \geq \frac{\sqrt{\cX{}/4}}{2 \cX{1}}  \).
%
Then, observe:
\begin{align*}
  \sqrt{\frac{\gammaX}{\vx}}
  \geq
  \frac{\sqrt{\gammaX}}{\sqrt{\gammaX}} \cdot \frac{\sqrt{\cC/4}}{2 \cC[1]}
  =
  \frac{\sqrt{\cC/4}}{2 \cC[1]}
  =
  \frac{\cC[1] + \sqrt{\cX{1}^{2} + \cX{2}}}{2\cC[1]}
  =
  \ux
\end{align*}
%
To state the result briefly, we have established that
\(  \ux \leq \sqrt{\frac{\gammaX}{\vx}}  \).
%
%
Thus,
\begin{align*}
  \ux^{2} \vx
  \leq
  \left( \sqrt{\frac{\gammaX}{\vx}} \right)^{2} \vx
  = \frac{\gammaX}{\vx} \cdot \vx
  = \gammaX
.\end{align*}
%
Lastly, applying \FACT \ref{fact:recurrence} yields
\begin{gather*}
  \Err{\Iter} \leq 2^{2^{-t}} ( \ux^2 \vx )^{1-2^{-t}} \leq 2^{2^{-t}} \gammaX^{1-2^{-t}}
  ,\\
  \lim_{\Iter \to \infty} \Err{\Iter} = \ux^{2}\vx \leq \gammaX
,\end{gather*}
as desired.
\end{proof}

\end{appendix}

%
%
%
%
%

\end{document}